\documentclass[a4paper,onecolumn,11pt,accepted=2026-11-05]{quantumarticle}
\pdfoutput=1
\usepackage{my_style}
\begin{document}
\title{Learning state preparation circuits for quantum phases of matter}
\author{Hyun-Soo Kim}
\affiliation{Department of Physics, University of California, Davis, CA, 95616, USA}
\author{Isaac H. Kim}
\affiliation{Department of Computer Science, University of California, Davis, CA, 95616, USA}
\author{Daniel Ranard}
\affiliation{Walter Burke Institute for Theoretical Physics, Caltech, Pasadena, CA, 91125, USA}
\maketitle
\begin{abstract}
Many-body ground state preparation is an important subroutine used in the simulation of physical systems. In this paper, we introduce a flexible and efficient framework for obtaining a state preparation circuit for a large class of many-body ground states. We introduce polynomial-time classical algorithms that take reduced density matrices over $\mathcal{O}(1)$-sized balls as inputs, and output a circuit that prepares the global state. We introduce algorithms applicable to (i) short-range entangled states (e.g., states prepared by shallow quantum circuits in any number of dimensions, and more generally, invertible states) and (ii) long-range entangled ground states (e.g., the toric code on a disk). Both algorithms can provably find a circuit whose depth is asymptotically optimal. Our approach uses a variant of the quantum Markov chain condition that remains robust against constant-depth circuits. The robustness of this condition makes our method applicable to a large class of states, whilst ensuring a classically tractable optimization landscape.
\end{abstract}

\section{Introduction}
\label{sec:introduction}

One of the important applications of quantum computers is expected to be the simulation of physical systems that appear in nature. This idea, which dates back to the seminal work of Feynman~\cite{Feynman1982}, has been pursued vigorously over the past few decades. A large body of work in this direction has been focused on simulating the dynamics of physical systems that appear in nature. The fact the dynamics can be simulated efficiently was first pointed out by Lloyd~\cite{lloyd1996universal}. More recently, modern paradigms such as the linear combination of unitaries~\cite{Childs2012,berry2014exponential} and qubitization~\cite{low2019hamiltonian} led to essentially optimal methods for such simulation.

However, less understood is a method for preparing a state on which the dynamics occurs. Often the initial state of interest is some low energy state of the system, such as the ground state. The predominant approach for the preparation of such states on a quantum computer is based on the quantum phase estimation algorithm~\cite{Kitaev1995} or its modern variants~\cite{Lin2020,Lin2022}. However, these methods succeed with high probability only if the initial state has a high overlap with the target ground state. Therefore, it is desirable to develop a method for finding a low-complexity circuit that can prepare the state. In this work, we introduce methods for finding such a circuit.

Traditional approaches for finding a state preparation circuit often begin with an efficient classical description of the state, which is then converted to a circuit. There are known methods in the literature, applicable to stabilizer states~\cite{montanaro2017learningstabilizerstatesbell}, Gaussian states~\cite{Zhang2018}, and tensor network states~\cite{cramer2010efficient,Schwarz2012,Schwarz2013,Ge2016,Wei2023}. For one-dimensional matrix product states, circuits whose depth is linear~\cite{Schon2005} and, for normal matrix product states, optimally logarithmic~\cite{Malz2024} in the system size are known. 

More recently, Ref.~\cite{Huang2024} introduced a method capable of finding a state preparation circuit even from a more limited information, provided that the underlying state can be prepared by a constant-depth circuit. Their method only requires knowing the reduced density matrices over $\mathcal{O}(1)$-sized subsystems, which is a weaker requirement than knowing the global state. Because such information can be obtained from an efficient classical description of the state, their method can be also used to convert such a description to a circuit, provided that the state is preparable by a constant-depth circuit. Moreover, their work opens up a new possibility of learning the state preparation circuit directly from a quantum state prepared in experiments. Because local reduced density matrices can be often readily obtained from such states (e.g., using classical shadows~\cite{Huang2020,Hu2023,Tran2023}), their approach can be used to efficiently learn state preparation circuits from many-body ground state wavefunctions prepared on a quantum computer~\cite{Huang2020}, or even a quantum simulator~\cite{Hu2023,Tran2023}.

However, from a physical point of view, the class of states that are amenable to the method of Ref.~\cite{Huang2024} is rather restrictive. This method only works in one and two spatial dimensions, leaving the higher dimensional case open. More importantly, their method only works for states that can be prepared by constant-depth circuits. This assumption can be restrictive at times, because there are physical ground states of interest that cannot be prepared by any constant-depth quantum circuit~\cite{bravyi2006lieb}. These are known as topologically ordered ground states, a well-known example being Kitaev's toric code~\cite{Kitaev2003}. 

In the context of many-body physics, the dichotomy between the state preparable by a constant-depth circuit and the toric code can be understood in terms of \emph{quantum phases} of matter. These are equivalence classes of many-body ground states (with a spectral gap), which are classified broadly into short-range entangled states and long-range entangled states~\cite{Chen2010}. Roughly speaking, short-range entangled states can be understood as states that can be prepared by constant-depth circuits, whereas long-range entangled states are the ground states that cannot be prepared by such circuits. From this point of view, the results of Ref.~\cite{cramer2010efficient,Huang2024} can be understood as a statement that state preparation circuit for short-range entangled states in one and two dimensions are efficiently learnable. Thus a natural question is whether the same statement holds true for all short-range entangled states, and more generally, even long-range entangled states. Our paper aims to address this question.

To that end, we introduce two new efficient algorithms for learning state preparation circuits. Our first algorithm is applicable to short-range entangled states. These are states prepared by a geometrically local constant-depth circuit on a finite-dimensional lattice. Compared to the prior work~\cite{Huang2024}, our work generalizes their result in two ways and uses different techniques. First, while the method of Ref.~\cite{Huang2024} only works in one and two dimensions, our work applies to short-range entangled states on any finite-dimensional lattice. Second, we remark that the class of states considered in the prior work~\cite{Huang2024} are of the following form:
\begin{equation}
    |\psi\rangle =U|0\ldots 0\rangle, \label{eq:sre_def}
\end{equation}
where $U$ is a constant-depth circuit. On the other hand, our work is applicable to the following (more general) class of states that involves ancillary qubits:
\begin{equation}
    |\psi\rangle \otimes |\phi\rangle= U|0\ldots 0\rangle,\label{eq:invertible_def}
\end{equation}
where $U$ is again a constant-depth circuit. Here, only the state $|\psi\rangle$ is given to us and the state of the ancilla $|\varphi\rangle$ is discarded. States $|\psi\rangle$ of the form in Eq.~\eqref{eq:sre_def} are called ``invertible'' states \cite{Kitaev2013}.  (Often invertible states are included in the definition of short-range entangled states.)
While Eq.~\eqref{eq:sre_def} and~\eqref{eq:invertible_def} may seem similar, there is a significant difference between the two. If we consider a more general \emph{quasi-local} unitary $U$, it is believed that there is a state that cannot be prepared using Eq.~\eqref{eq:sre_def} but can be using Eq.~\eqref{eq:invertible_def}. (This is known as Kitaev's $E_8$ state~\cite{kitaev2006anyons}.) For simplicity, we take $U$ in Eq.~\eqref{eq:invertible_def} to be a constant-depth circuit in this paper. Extending this analysis to quasi-local unitaries is expected to be possible using techniques such as the Lieb-Robinson bound~\cite{hastings2010localityquantumsystems}, which we leave for future work.

Our second algorithm is applicable to long-range entangled states, modulo some caveats discussed below. Broadly speaking, long-range entangled states are ground states of gapped systems which cannot be prepared by any constant-depth circuit, even with ancillas.  Well-known examples in two dimensions (2D) include the toric code~\cite{Kitaev2003}, and more generally, the string-net models~\cite{Levin2005}. In higher dimensions, more exotic models such as fractons exist~\cite{Haah2011}. Our algorithm is applicable to a subclass of long-range entangled states which possess  so-called ``liquid topological order.'' These are states that satisfy a certain homogeneity condition~\cite{DraftOther}. Examples of such states include topologically ordered ground states in 2D. This class also includes models such as the higher-dimensional generalization of the toric code, but does not include fractons; see Ref.~\cite{DraftOther} for more discussion. An additional caveat is that our result is only applicable to states that are a unique ground state, e.g., the toric code on a sphere, or subsystems of such states.

Similar to Ref.~\cite{Huang2024}, both of our algorithms take reduced density matrices over $\mathcal{O}(1)$-size ball as inputs, and reconstruct a unitary that prepares the global state. This is possible because from the reduced density matrices one can learn a set of local physical processes that glue the local reduced density matrices together~\cite{DraftOther}. By stitching up those local processes together, one can obtain a circuit that prepares the global state. 

Thus, our results lead to a number of new applications. As said before, our method can be applied to learning the state preparation of circuits for the quantum states that appear in experiments. Alternatively, one may use our method to find a state preparation circuit with a rigorous guarantee from a classical description of the quantum state. This is because our method only requires knowledge of reduced density matrices over $\mathcal{O}(1)$-size balls. Such information can be obtained from any efficient classical description of the state, because such a description by definition implies an ability to efficiently compute the reduced density matrices.

In many cases, given only a classical description of the Hamiltonian (and not yet the state), our construction of the circuit preparing the ground state remains efficient.
In particular, for a class of Hamiltonians that satisfy the local topological quantum order condition~\cite{michalakis2013stability}, the reduced density matrix of a $\mathcal{O}(1)$-size region with respect to the global state is indistinguishable from the reduced density matrix of a local ground state, defined with respect to the Hamiltonian truncated to a slightly larger $\mathcal{O}(1)$-size subsystem. Therefore, for those ground states such systems, each reduced density matrix can be obtained in $\mathcal{O}(1)$ time, which can be then fed into our algorithm.

The rest of this paper is structured as follows. In Section~\ref{sec:summary}, we provide an executive summary of our main results and approach. In Section~\ref{sec:locally_extendible_states}, we introduce and discuss a notion of ``locally extendible'' states, the central concept we use in this paper. In Section~\ref{sec:extendibility_quantum_phase}, we explain the relation between local extendibility and the state preparation circuits of known subclasses of quantum phases. In Section~\ref{sec:learning}, we introduce our polynomial-time algorithms for learning the state preparation circuits. We end with a discussion in Section~\ref{sec:discussion}.

\section{Summary of results}
\label{sec:summary}

The main results of this paper are algorithms for learning the state preparation circuit of a given many-body state. We remark that there cannot be an efficient general-purpose method that works for any quantum state. For exponentially complex states, even the process of writing down the circuit will take an exponentially long time. Thus we need to identify a subclass of states which are amenable to an efficient method. 

The class of states we consider, as we alluded to in Section~\ref{sec:introduction}, are short-range and long-range entangled states in various dimensions. Because the short-range entangled states [Eq.~\eqref{eq:sre_def}], and more generally, invertible states [Eq.~\eqref{eq:invertible_def}] can be prepared by a constant-depth circuit, the goal of the algorithm would be to find a constant-depth circuit that prepares the same state. On the other hand, long-range entangled states such as the toric code~\cite{Kitaev2003} cannot be prepared by a constant-depth circuit~\cite{bravyi2006lieb,haah2016invariant}. Assuming the circuit is geometrically local, the requisite circuit depth in two spatial dimensions must scale as $\Omega(\sqrt{n})$, where $n$ is the number of qubits. Therefore, the best one can hope for is to find a circuit that achieves the same depth asymptotically. For this reason, our algorithms for the short-range entangled and long-range entangled states will be different. 

Nonetheless, there is a unifying theoretical framework that underpins both of these algorithms. (We defer this discussion to Section~\ref{subsec:techniques}.) As such, both algorithms share a similar overall structure; both algorithms take reduced density matrices over $\mathcal{O}(1)$-size balls as inputs, and outputs a state preparation circuit. The input to the algorithms are in fact the same; it is merely the classical post-processing that differs. 

Our first algorithm is an efficient algorithm for learning the state preparation circuits of invertible states [Eq.~\eqref{eq:invertible_def}]. Since invertible states include states that are preparable by constant-depth circuits without ancilla [Eq.~\eqref{eq:sre_def}], this algorithm is also applicable to such states.
\begin{restatable}[]{theorem}{mainone}
    \label{thm:main_invertible}
    Consider an unknown $n$-qubit invertible state prepared by a depth-$d$ geometrically local circuit in $k$ spatial dimensions. There is an algorithm that learns a circuit that prepares the same state up to a trace distance of $\epsilon$ with high probability from $\widetilde{O}(n^2 \log (n)/\epsilon^2 \cdot ke^{\mathcal{O}(d^k)})$ samples of the state, with the following complexity:
    \begin{itemize}
        \item One can obtain a circuit of depth $\mathcal{O}(ke^{\mathcal{O}(d^k)})$ using $\widetilde{\mathcal{O}}(e^{\mathcal{O}(d^k)}n^3/\epsilon^2)$ classical running time, or
        \item One can obtain a circuit of depth $\mathcal{O}(kd)$ using $\mathcal{O}((n/\epsilon)^{\mathcal{O}(d^k)})$ classical running time.
    \end{itemize}
\end{restatable}
\noindent
(Here $\widetilde{O}$ hides the polylogarithmic factor in $n$ and $1/\epsilon$.) Compared to the existing method~\cite{Huang2024}, there are three main differences. First, our algorithm works in any number of dimensions. Second, our algorithm is applicable to invertible states, i.e., states preparable in constant depth with ancillas, which are more general than those where ancillas are not allowed. Thirdly, our algorithm may require a lower overall classical running time, at the expense of an increased circuit depth. With the increased run-time, we can match the asymptotically optimal circuit depth in $d$ for $k=\mathcal{O}(1)$.

Our second algorithm is an efficient algorithm for learning the state preparation circuits of long-range entangled states ~\cite{DraftOther}. While generalizations are possible, we specialize to two spatial dimensions now. Explaining our definition of long-range entangled states requires a longer exposition, which we defer to Section~\ref{subsec:lre}. For now, we simply note that the states we consider include string-net models~\cite{Levin2005} on a disk with uniform boundary conditions, or the states obtained by applying a constant-depth circuit to such states. As a special case, this includes the toric code ground state on a disk with a smooth boundary condition, or a state obtained by applying a constant-depth circuit to it. 
\begin{theorem}[Informal version of Theorem~\ref{thm:main_lre}]
    Given an unknown $n$-qubit state of the form $U|\psi\rangle$, where $|\psi\rangle$ is a long-range entangled state on a disk with a uniform boundary condition and $U$ is a depth-$d$ quantum geometrically local circuit in two spatial dimensions, one can learn a geometrically local circuit that prepares the same state up to a trace distance of $\epsilon$ with high probability from $\widetilde{O}(n^2 \log (n)/\epsilon^2 \cdot e^{\mathcal{O}(d^2)})$  samples of the state, with the following complexity: 
    \begin{itemize}
        \item One can obtain a circuit of depth $\mathcal{O}(\sqrt{n} e^{\mathcal{O}(d^2)})$ using $\mathcal{O}(e^{\mathcal{O}(d^2)}n^3 \log (n)/\epsilon^2)$ running time, or
        \item One can obtain a circuit of depth $\mathcal{O}(\sqrt{n} d)$ using $\mathcal{O}((n/\epsilon)^{\mathcal{O}(d^2)})$ classical running time.
    \end{itemize}
    \label{theorem:main2}
\end{theorem}
\noindent
We remark that the $\sqrt{n}$-dependence in the circuit depth is optimal. It is known that using geometrically local circuits, the circuit depth to prepare ground states of topologically ordered state must be $\Omega(\sqrt{n})$~\cite{bravyi2006lieb,haah2016invariant}. Our algorithm therefore recovers a circuit with asymptotically optimal depth. Meanwhile, the necessity of some exponential dependence on depth $d$ can be inferred from the result of Ref.~\cite{schuster2024random}, which constructs low-depth pseudorandom unitaries.

Our algorithms can also be applied to arbitrary states, in which case the algorithm may not find a preparation circuit in polynomial time, but when it does, the result can be efficiently certified. That is, even without any guarantees about the input state, if the algorithm succeeds in finding a circuit, the circuit provably constructs the state with a rigorous approximation guarantee. This is due to the fact that our algorithm outputs a correct state preparation circuit if a certain set of efficiently checkable local conditions are satisfied. These are precisely the conditions that are used in finding the state preparation circuit, as we explain in Section~\ref{sec:locally_extendible_states} and~\ref{sec:extendibility_quantum_phase}. Therefore, applied to the states that lie outside of the class we consider (e.g., Haar-random states), our algorithms will simply have a runtime exponential in $n$. On the other hand, if the local conditions are satisfied, our algorithm will output a circuit, whose output is guaranteed to be close to the target state.

Due to the flexibility of our approach, our algorithms can be used for a number of different purposes. One attractive possibility is to use the quantum simulators in the near term~\cite{Zhang2017,Ebadi2021,Scholl2021}. It was recently shown that the shadow tomography approach~\cite{Huang2020} can be applied to quantum simulators using a scrambling dynamics~\cite{Hu2023,Tran2023}. This leads to a possibility of directly learning the reduced density matrices of many-body ground states that are prepared using quantum simulators, which in turn leads to an efficient state preparation circuit using our method. The circuits learned this way can be later used as a subroutine for various quantum algorithms, such as the quantum phase estimation algorithm or simulation of dynamics. 

Another possibility is to use our method to efficiently find a state-preparation circuit for states that already have an efficient classical description. These are states for which the expectation values of local observables with respect to the ground state wavefunction can be estimated efficiently, such as the neural network state~\cite{carleo2017solving}, isometric tensor network~\cite{Zaletel2020}, multi-scale entanglement renormalization ansatz~\cite{Vidal2008}, or approaches based on quantum Monte Carlo~\cite{becca2017quantum}. In all these cases, provided that the underlying state belong to a class of states we consider, our method should be applicable. Thus, our approach can be used to efficiently construct the state preparation circuit even when the circuit representation is not evident in the first place (e.g.,~\cite{carleo2017solving,becca2017quantum}). Alternatively, it may potentially lead to a more compressed circuit even when a different circuit description is known~\cite{Zaletel2020,Vidal2008}. 

In fact, our method can be used even when an efficient classical description is not known. The most common setup we may consider is when we are given a description of the parent Hamiltonian, but not the state. If the Hamiltonian satisfies the local topological quantum order (LTQO) condition~\cite{michalakis2013stability}, the required density matrix can be obtained approximately by numerically calculating the ground state of a Hamiltonian localized to a slightly larger $\mathcal{O}(1)$-size subsystem and then erasing the boundary. This is because the LTQO condition guarantees the consistency between the reduced density matrix obtained from the localized Hamiltonian and the one obtained from the global Hamiltonian. Although the LTQO condition may be difficult to verify, it is a condition that remains robust against perturbation~\cite{Cirac2013}.\footnote{The stability bound in ~\cite[Theorem 9]{Cirac2013} requires stability of a local gap, which was proved in ~\cite[Theorem 6.8]{nachtergaele2022quasi}.} As such, for a large class of physical Hamiltonians of interest, we anticipate our method to be applicable.

\subsection{Approach}
\label{subsec:techniques}
The main approach we use is based on a notion of recoverability for density matrices, recently introduced in Ref.~\cite{DraftOther}. This is related to the notion of quantum Markov chains~\cite{Petz1988,Petz2003}, though there are also important differences. In order to provide some contexts, let us first briefly review quantum Markov chains and its relevance to the circuit learning problem. A Quantum Markov chain refers to a tripartite quantum state $\rho_{ABC}$ that satisfies the following property. Let $S(\rho)= -\text{Tr}(\rho \log \rho)$ be the von Neumann entropy of $\rho$. We say $\rho_{ABC}$ is a quantum Markov chain if $I(A:C|B):=S(\rho_{AB}) + S(\rho_{BC}) - S(\rho_B) - S(\rho_{ABC})$ is zero. An important property of quantum Markov chains is that $\rho_{ABC}$ admits the following recovery map $\Phi_{B\to BC}$~\cite{Petz1988,Petz2003}:
\begin{equation}
    \rho_{ABC} = \mathcal{I}_A\otimes \Phi_{B\to BC} (\rho_{AB}).
    \label{eq:recovery_qmc}
\end{equation}
\sloppy
(This map, which is known as the Petz map~\cite{Petz1988}, has a closed-form expression $\Phi_{B\to BC}(\cdot) = \rho_{BC}^{\frac{1}{2}}\rho_B^{-\frac{1}{2}}(\cdot)\rho_B^{-\frac{1}{2}} \rho_{BC}^{\frac{1}{2}}$.) Conversely, if Eq.~\eqref{eq:recovery_qmc} is satisfied, $\rho_{ABC}$ is a quantum Markov chain. 

\begin{figure}[h]
    \centering
    \begin{tikzpicture}[line width=1pt, scale=0.8]
        \foreach \x in {1, ..., 5}
        {
        \draw[] (\x, 1) -- (\x, 5);
        }
        \foreach \y in {1, ..., 5}
        {
        \draw[] (1, \y) -- (5, \y);
        }
        \draw[line width=2pt, red] (1,1) -- (2,1);
        \draw[line width=2pt, red] (1,1) -- (1,2);
        \draw[line width=2pt, green] (2,1) -- (2,2);
        \draw[line width=2pt, green] (1,2) -- (2,2);
        \draw[line width=2pt, green] (2,1) -- (3,1);
        \draw[line width=2pt, blue] (3,1) -- (3,2);
        \draw[line width=2pt, blue] (2,2) -- (3,2);

        \node[] () at (3, 0.5) {(a)};

        \begin{scope}[xshift=7cm]
        \foreach \x in {1, ..., 5}
        {
        \draw[] (\x, 1) -- (\x, 5);
        }
        \foreach \y in {1, ..., 5}
        {
        \draw[] (1, \y) -- (5, \y);
        }
        \draw[line width=2pt, red] (1,1) -- (2,1);
        \draw[line width=2pt, red] (1,1) -- (1,2);
        \draw[line width=2pt, red] (2,1) -- (2,2);
        \draw[line width=2pt, red] (1,2) -- (2,2);
        \draw[line width=2pt, green] (2,1) -- (3,1);
        \draw[line width=2pt, green] (3,1) -- (3,2);
        \draw[line width=2pt, green] (2,2) -- (3,2);
        \draw[line width=2pt, blue] (3,1) -- (4,1);

        \node[] () at (3, 0.5) {(b)};
            
        \end{scope}
    \end{tikzpicture}
    \caption{Examples of quantum Markov chains for the toric code with a smooth boundary condition~\cite{Kitaev2003,bravyi1998quantumcodeslatticeboundary}. The qubits are placed on edges. The subsystems $A$, $B$, and $C$ are colored in red, green, and blue, respectively.}
    \label{fig:qmc_toric_code}
\end{figure}
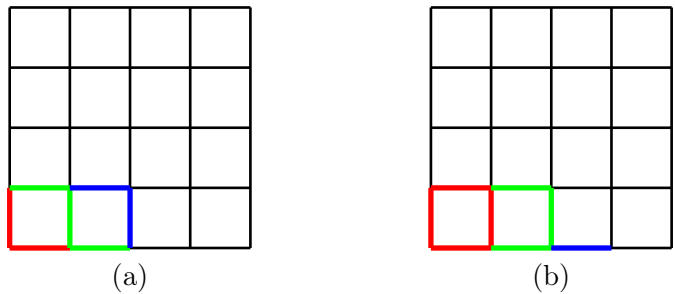

In the context of the circuit learning problem, the quantum Markov chain is relevant -- though with a caveat, as we discuss further below -- because many physical states of interest have the quantum Markov chain structure~\cite{Kim2013,shi2020fusion}. A nontrivial example is the toric code on a disk, with a uniform boundary condition (e.g., smooth boundary condition everywhere). For an appropriately chosen set of subsystems $A, B,$ and $C$, one can show that the ground state reduced density matrix over such $ABC$ forms a quantum Markov chain [Fig.~\ref{fig:qmc_toric_code}]~\cite{Kitaev2006,Levin2006,Kim2013}. Because the recovery map in Eq.~\eqref{eq:recovery_qmc} has a form that depends only on the reduced density matrix over $\rho_{BC}$, if we know this reduced density matrix, one can infer the recovery map in Eq.~\eqref{eq:recovery_qmc}. By sequentially applying such recovery maps associated with the subsystems in Fig.~\ref{fig:qmc_toric_code}, we can find a sequence of local channels that prepares the given state. This can be easily turned into a circuit by recalling that each local channel can be purified to a unitary, using an ancillary system of bounded dimension.

While the above delineated approach works for the toric code, it does not in more general contexts. If we were to apply a geometrically local constant-depth circuit to the state, the quantum Markov chain structure may no longer hold~\cite{Zou2016,Williamson2019,Kato2020}. This is undesirable because the two states are in the same phase. Ideally one would want a method that works for any state within the same phase, but the approach based on quantum Markov chain fails. This fragility of the quantum Markov chain structure under constant-depth circuits is what makes it inadequate for the circuit learning problem. 

However, there is a more relaxed variant of the quantum Markov chain that can avoid this difficulty~\cite{DraftOther}. In this variant, instead of demanding the recovery of the global state, one simply demands a recovery of the same state over a smaller subsystem. More precisely, note that in Eq.~\eqref{eq:recovery_qmc} $\Phi_{B\to BC}$ recovers the global state $\rho_{ABC}$ from $\rho_{AB}$. Instead, in the more relaxed notion one would aim to apply a channel $B\to B'C$, where $B'\subsetneq B$, with the goal of producing a state $\rho_{AB'C}$. Since $B'\subsetneq B$, we are demanding a condition strictly weaker than Eq.~\eqref{eq:recovery_qmc}.

The key advantage of this relaxed notion is twofold. First, though it is weaker than Eq.~\eqref{eq:recovery_qmc}, it is still strong enough to build up a circuit in the spirit of Fig.~\ref{fig:qmc_toric_code}, as we explain in Section~\ref{sec:extendibility_quantum_phase}. Second, unlike the quantum Markov chain condition, the weaker notion in Ref.~\cite{DraftOther} remains true even after applying a constant-depth circuit, as we prove in Section~\ref{subsec:stability_circuit}. The invariance of this condition under constant-depth circuit is what makes it suitable for the circuit learning problem.

Thus, the bulk of our analysis concerns this more relaxed notion of the quantum Markov chain. We refer to them as \emph{locally extendible states}, and discuss many of their properties in Section~\ref{sec:locally_extendible_states}. As a particularly pertinent property, we show the stability of extendibility under constant-depth circuits in Section~\ref{subsec:stability_circuit}. From the properties of these states, our learning algorithm construction follows, and we explain them in Sections~\ref{sec:extendibility_quantum_phase} and~\ref{sec:learning}.

\subsection{Notation and convention}
\label{subsec:notation}

Throughout this paper, we shall use the following convention. We denote a Hilbert space over a subsystem $S$ as $\mathcal{H}_S$. We denote the set of density matrices acting on a Hilbert space $\mathcal{H}_S$ as $\mathcal{D}_S$. We denote the identity channel acting on subsystem $S$ as $\mathcal{I}_S$. A \emph{channel} is a completely positive trace-preserving map. For any subsystem $S$, we denote the set of sites that are distance $x>0$ or less away from $S$ as $S(x)$. If $x<0$, $S(x)$ is a subsystem $S' \subset S$ such that $S'(-x)=S$. For two subsystems $S$ and $T$, we denote the minimum distance between a site of $S$ and a site of $T$ as $\operatorname{dist}(S,T)$.

\section{Locally extendible states}\label{sec:locally_extendible_states}

In this Section, we introduce the notion of \emph{locally extendible states}, or \emph{extendible states} for short.\footnote{This notion was recently introduced as \emph{partial recoverability} in Ref.~\cite{DraftOther}.} Though the precise definition of locally extendible state may seem abstract, we can obtain a more intuitive understanding by applying the definition to a reduced density matrix over an $\mathcal{O}(1)$-size ball for the states discussed in Section~\ref{sec:summary}. Below we will discuss these examples and then discuss properties of locally extendible states pertinent to this paper.

Let us first define locally extendible states as follows.
\begin{definition}\label{definition:locally_extendible_exact}
    A density matrix $\rho_{BCDE}$ is \emph{locally extendible} (or \emph{extendible}) from $BE$ to $BC$ if there exists a channel $\Phi:\mathcal{D}_{BE} \to \mathcal{D}_{BC}$ such that
    \begin{equation}
        \rho_{ABC} = \mathcal{I}_A\otimes \Phi( \rho_{ABE}),\label{eq:locally_extendible}
    \end{equation}
    where $\rho_{ABCDE}$ is a purification of $\rho_{BCDE}$. 
\end{definition}

\begin{definition}\label{definition:realizing_extension_exact}
    A channel $\Phi: \mathcal{D}_{BE} \to \mathcal{D}_{BC}$ \emph{extends} $\rho_{BCDE}$, if $\rho_{ABC} = \mathcal{I}_A \otimes \Phi (\rho_{ABE})$ for a purification $\rho_{ABCDE}$. Such a channel $\Phi$ is called an \emph{extending channel} or \emph{extending map} of $\rho_{BCDE}$ from $BE$ to $BC$.
\end{definition}

A guiding example of a locally extendible state is a reduced density matrix of the states discussed in Section~\ref{sec:summary}, for a judiciously chosen set of subsystems. For concreteness, we will consider two specific examples: a product state and Kitaev's toric code~\cite{Kitaev2003}. Consider the subsystems $BCDE$ shown in Fig.~\ref{fig:extendible_state_examples}(a) and (c), which can be viewed as a disk-like subsystem of a 2D lattice. For concreteness, let us set the purifying space $A$ to be the rest of the system, assuming the global state is pure [Figs.~\ref{fig:extendible_state_examples}(b) and~\ref{fig:extendible_state_examples}(d)].

For the subsystem $BCDE$ in Fig.~\ref{fig:extendible_state_examples}(a), it is easy to see that a product state is locally extendible from $BE$ to $BC$. The purifiying space $A$ is also a product state, decoupled from the rest. Therefore, we can construct an extending channel by appending the reduced density matrix of $C$ and taking a partial trace over $E$. The resulting state is a reduced density matrix over $ABC$, and as such, the reduced density matrix over $BCDE$ is extendible from $BE$ to $BC$. 

A more interesting example is the toric code. In this case, the purifying space $A$ is no longer decoupled from the rest. However, a simple calculation of entanglement entropy shows that 
\begin{equation}
\label{eq:weak_monotonicity_saturation}
    S(\rho_{BCE}) + S(\rho_{CD}) - S(\rho_{BE}) - S(\rho_{D})=0.
\end{equation}
This is because the entanglement entropy of the toric code has a form of $S(\rho_X) = \alpha |\partial X| - \log 2$ for any disk-like subsystem $A$, where $|\partial X|$ is the boundary area of $X$~\cite{Hamma2005,Kitaev2006,Levin2006}. It is an easy exercise to verify that this formula, applied to the linear combination of entanglement entropies in Eq.~\eqref{eq:weak_monotonicity_saturation}, yields zero~\cite{shi2020fusion}. From the strong subadditivity of entropy (SSA)~\cite{Lieb1973}, it then follows that
\begin{equation}
    I(A:C|BE)_{\rho} \leq S(\rho_{BCE}) + S(\rho_{CD}) - S(\rho_{BE}) - S(\rho_{D}),
\end{equation}
where $A$ is the purifying space. Again by the SSA, $I(A:C|BE)_{\rho}\geq 0$, and as such, $I(A:C|BE)_{\rho}$ must be zero. This equality is satisfied if and only if there is a channel $\Phi_{BE\to BEC}$ from $BE$ to $BEC$ such that $\rho_{ABCE} = \mathcal{I}_A\otimes \Phi_{BE\to BEC}(\rho_{ABE})$~\cite{Petz1988,Petz2003}. Therefore, Eq.~\eqref{eq:weak_monotonicity_saturation}, which the toric code satisfies, implies that there is an extending map from $BE$ to $BCE$. We can then simply take a partial trace over $E$, obtaining an extending map from $BE$ to $BC$. Thus, the toric code reduced density matrix over the subsystem $BCDE$ in Fig.~\ref{fig:extendible_state_examples} is also extendible from $BE$ to $BC$.

On the other hand, for the subsystem $BCDE$ in Fig.~\ref{fig:extendible_state_examples}(c), we arrive at different conclusions for the two examples. For a product state, we can again conclude that the reduced density matrix over $BCDE$ is extendible from $BE$ to $BC$. On the other hand, for the toric code, such a recovery map cannot exist due to the presence of long-range entanglement~\cite{Kim2013}.\footnote{To see why, consider two ground states of the toric code. If $BCDE$ is extendible from $BE$ to $BC$, that leads to the conclusion that the two states are identical, which is a contradiction.} Thus, the local extendibility condition for the subsystems in Fig.~\ref{fig:extendible_state_examples}(c) are satisfied by product states, but not necessarily by long-range entangled states such as the toric code.

Let us make a remark on the relevance of these examples to the circuit learning problem. Our approach to solving the circuit learning problem is to use the extending maps for the locally extendible states discussed above. A sequential application of those maps can  build up the global state, which in its entirety defines a state preparation circuit. (We defer the discussion on how the global state is built up to Section~\ref{sec:learning}.) 

In order for this approach to work, we must address the following problems. First, given a density matrix, we should be able to verify (or disprove) that the density matrix is locally extendible. Moreover, if the density matrix is extendible, we must have an efficient procedure for finding the extending map. Second, note that the class of states we consider includes not just the product state and the toric code, but also the states obtained by applying a constant-depth circuit to them. As such, it is important to verify that a locally extendible state remains locally extendible under such a circuit. Once these problems are solved, we are done; for the class of states we consider, we can simply obtain the extending map from the local reduced density matrices, from which the state preparation circuit can be obtained. We focus on these remaining tasks in the rest of this section.

\begin{figure}[t]
    \centering
    \begin{tikzpicture}
    \draw[line width=1pt] (0,0) circle (0.5cm);
    \draw[line width=1pt] (0,0) circle (1cm);
    \draw[line width=1pt] (75:0.5cm) -- (75:1cm);
    \draw[line width=1pt] (105:0.5cm) -- (105:1cm);
    \draw[line width=1pt] (-75:0.5cm) -- (-75:1cm);
    \draw[line width=1pt] (-105:0.5cm) -- (-105:1cm);

    \node[] () at (90:0.75cm) {$E$};
    \node[] () at (-90:0.75cm) {$E$};
    \node[] () at (0,0) {$C$};
    \node[] () at (0:0.75cm) {$D$};
    \node[] () at (180:0.75cm) {$B$};
    \node[below] () at (0, -1.5) {(a)};
    \begin{scope}[xshift=3.5cm,scale=0.375]
    \draw[line width=1pt] (-5,-2.5) -- ++ (10, 0) -- ++ (0, 5) -- ++ (-10, 0)--cycle;
    \node[] () at (-3, 0) {$A$};
    \node[below] () at (0, -4) {(b)};
    \begin{scope}[xshift=3cm, yshift=0.25cm]
    \draw[line width=1pt] (0,0) circle (0.5cm);
    \draw[line width=1pt] (0,0) circle (1cm);
    \draw[line width=1pt] (75:0.5cm) -- (75:1cm);
    \draw[line width=1pt] (105:0.5cm) -- (105:1cm);
    \draw[line width=1pt] (-75:0.5cm) -- (-75:1cm);
    \draw[line width=1pt] (-105:0.5cm) -- (-105:1cm);
    \end{scope}
        
    \end{scope}
        \begin{scope}[xshift=7cm]
            \draw[line width=1pt] (0,0) circle (0.5cm);
            \draw[line width=1pt] (0,0) circle (1cm);
            \draw[line width=1pt] (30:0.5cm) -- (30:1cm);
            \draw[line width=1pt] (150:0.5cm) -- (150:1cm);
            \draw[line width=1pt] (-30:0.5cm) -- (-30:1cm);
            \draw[line width=1pt] (-150:0.5cm) -- (-150:1cm);
            \draw[line width=1pt] (60:0.5cm) -- (60:1cm);
            \draw[line width=1pt] (120:0.5cm) -- (120:1cm);
            \draw[line width=1pt] (-60:0.5cm) -- (-60:1cm);
            \draw[line width=1pt] (-120:0.5cm) -- (-120:1cm);
    \node[below] () at (0, -1.5) {(c)};
    \node[] () at (90:0.75cm) {$D$};
    \node[] () at (0:0.75cm) {$B$};
    \node[] () at (0,0) {$C$};
    \node[] () at (180:0.75cm) {$B$};
    \node[] () at (270:0.75cm) {$D$};
    \node[] () at (45:0.75cm) {$E$};
    \node[] () at (-45:0.75cm) {$E$};
    \node[] () at (135:0.75cm) {$E$};
    \node[] () at (-135:0.75cm) {$E$};
        \end{scope}

    \begin{scope}[xshift=10.5cm,scale=0.375]
    \draw[line width=1pt] (-5,-2.5) -- ++ (10, 0) -- ++ (0, 5) -- ++ (-10, 0)--cycle;
    \begin{scope}[xshift=1cm, yshift=0.25cm]
    \draw[line width=1pt] (0,0) circle (0.5cm);
            \draw[line width=1pt] (0,0) circle (1cm);
            \draw[line width=1pt] (30:0.5cm) -- (30:1cm);
            \draw[line width=1pt] (150:0.5cm) -- (150:1cm);
            \draw[line width=1pt] (-30:0.5cm) -- (-30:1cm);
            \draw[line width=1pt] (-150:0.5cm) -- (-150:1cm);
            \draw[line width=1pt] (60:0.5cm) -- (60:1cm);
            \draw[line width=1pt] (120:0.5cm) -- (120:1cm);
            \draw[line width=1pt] (-60:0.5cm) -- (-60:1cm);
            \draw[line width=1pt] (-120:0.5cm) -- (-120:1cm);
    \end{scope}
    \node[below] () at (0, -4) {(d)};
    \node[] () at (-3, 0) {$A$};
    \end{scope}
    \end{tikzpicture}
    \caption{(a) The reduced density matrices of a product state/toric code over $BCDE$ is locally extendible from $BE$ to $BC$. (b) The subsystem $BCDE$ in (a) can be viewed as a disk in a 2D lattice. (c) The reduced density matrix of a product state over $BCDE$ is locally extendible from $BE$ to $BC$; this is not the case for the toric code. (d) The subsystem $BCDE$ in (c) can be viewed as a disk in a 2D lattice.}
    \label{fig:extendible_state_examples}
\end{figure}
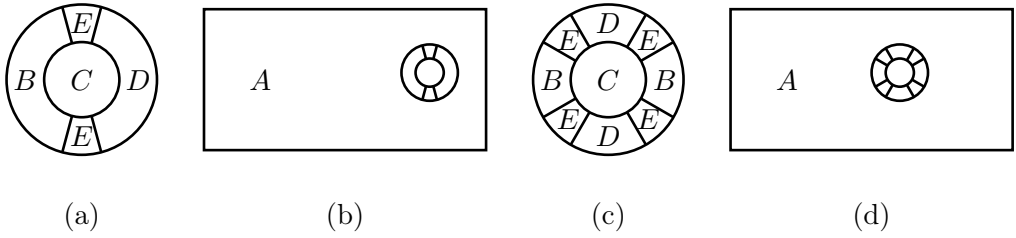

In Section~\ref{subsec:checking_extendibility}, we introduce an efficient method to verify that a given state is locally extendible. In particular, we introduce an efficient method for finding the extending map. In Section~\ref{subsec:stability_approximation}, we prove several useful continuity property of extendible states. In Section~\ref{subsec:stability_circuit}, we prove that the notion of local extendibility for the subsystems in Fig.~\ref{fig:extendible_state_examples} are stable under geometrically local constant-depth circuits, provided that the diameters of the subsystems are large compared to the circuit depth.

\subsection{Verifying extendibility}
\label{subsec:checking_extendibility}

If we were given an extending map of a locally extendible state, it would be straightforward to verify that the state is locally extendible. However, how to carry out this task is less clear if the input is simply a state, without the extending map. In this Section, we introduce an efficient method to check the extendibility of a state. As a byproduct of our analysis, we will also obtain an efficient method for finding a (near-optimal) extending map. 

To that end, we introduce a notion of \emph{approximately extendible states}. Typically, no state will exactly satisfy the definition of extendibility [Definition~\ref{definition:locally_extendible_exact}]. Therefore, we now extend the notion of local extendibility to one that allows for an approximation error. We will use the Bures metric to quantify this error. Let $F(\rho, \sigma) = \text{Tr}((\sigma^{\frac{1}{2}} \rho \sigma^{\frac{1}{2}})^{\frac{1}{2}})$ be the fidelity. Define $\bures{\rho}{\sigma} := \sqrt{1-F(\rho,\sigma)}$, which is known as the \emph{Bures distance}. Bures distance is a distance measure, and it is related to the trace distance as follows~\cite{Fuchs1999}:
\begin{equation} \label{eq:fuchs}
    2\bures{\rho}{\sigma}^2 \leq \|\rho - \sigma \|_1 \leq 2\sqrt{2} \bures{\rho}{\sigma}.
\end{equation}
We also remark that the Bures distance is non-increasing under channels.

The following is our definition of \emph{approximately (locally) extendible state}.
\begin{definition}\label{definition:locally_extendible_approximate}
    A density matrix $\rho_{BCDE}$ is $\epsilon$-\emph{locally-extendible}, or $\epsilon$-\emph{extendible} for short, from $BE$ to $BC$ if
    \begin{equation}
    \min_{\Phi: \mathcal{D}_{BE}\to \mathcal{D}_{BC}} \bures{\rho_{ABC}}{\mathcal{I}_A \otimes \Phi(\rho_{ABE})} \leq \epsilon,
    \label{eq:locally_extendible_approximate}
    \end{equation}
    where $\rho_{ABCDE}$ is a purification of $\rho_{BCDE}$, and $\epsilon \geq 0$.
\end{definition}

\begin{definition}\label{definition:realizing_extension_approximate}
    A channel $\Phi:\mathcal{D}_{BE}\to \mathcal{D}_{BC}$ \emph{realizes an $\epsilon$-extension} of $\rho_{BCDE}$, or \emph{$\epsilon$-extends} $\rho_{BCDE}$, if, for a purification $\rho_{ABCDE}$, 
    \begin{equation}
        \bures{\rho_{ABC}}{\mathcal{I}_A\otimes \Phi(\rho_{ABE})}\leq \epsilon.
    \end{equation}
\end{definition}

Both the extendibility [Definition~\ref{definition:locally_extendible_exact}] and approximate extendibility [Definition~\ref{definition:locally_extendible_approximate}] requires a purification. Fortunately, these definitions do not depend on the choice of purification.
\begin{lemma} 
If Eq.~\eqref{eq:locally_extendible_approximate} holds for one purification $\rho_{ABCDE}$, it also holds true for another purification $\rho_{A'BCDE}$.
\end{lemma}
\begin{proof}
    By Uhlmann's theorem~\cite{Uhlmann1976}, there is an isometry $V: \mathcal{H}_A \to \mathcal{H}_{A'}$ such that
    \begin{equation}
    \rho_{A'BCDE} = V \rho_{ABCDE} V^{\dagger}.
    \end{equation}
    Because $V$ acts trivially on $BCE$, and the Bures distance is invariant under an isometry, from Eq.~\eqref{eq:locally_extendible_approximate} we obtain:
    \begin{equation}
    \min_{\Phi: \mathcal{D}_{BE}\to \mathcal{D}_{BC}} \bures{\rho_{A'BC}}{\mathcal{I}_{A'} \otimes \Phi(\rho_{A'BE})}  = \min_{\Phi: \mathcal{D}_{BE}\to \mathcal{D}_{BC}} \bures{\rho_{ABC}}{\mathcal{I}_{A} \otimes \Phi(\rho_{ABE})} \leq \epsilon.
    \end{equation}
\end{proof}
\noindent
With a similar reasoning, if a  channel $\Phi: \mathcal{D}_{BE}\to \mathcal{D}_{BC}$ $\epsilon$-extends $\rho_{BCDE}$ for a purification of $\rho_{ABCDE}$, it also $\epsilon$-extends the same state for a different purification $\rho_{A'BCDE}$. Therefore, even without specifying the purifying system, we can unambiguously say that a channel realizes an $\epsilon$-extension [Definition~\ref{definition:realizing_extension_approximate}].

Because the extendibility does not depend on the choice of the purifying space, we can simply choose it to be the one with the minimal dimension, which is the rank of $\rho_{BCDE}$. By choosing such a purification, the Bures distance between $\rho_{ABC}$ and $\mathcal{I}_A\otimes \Phi(\rho_{ABE})$ for a given $\Phi$ can be computed in time $\text{poly}(\text{rk}(\rho_{BCDE}))$, where $\text{rk}(M)$ is the rank of a matrix $M$.

Now it remains to explain how one can solve the optimization problem in Eq.~\eqref{eq:locally_extendible_approximate}, and in particular find the optimal solution $\Phi: \mathcal{D}_{BE} \to \mathcal{D}_{BC}$. This problem can be solved by using known facts about the \emph{fidelity of recovery}~\cite{Fawzi2015}, defined as:
\begin{equation}
F_{C\to B}(\rho_{AB}\| \sigma_{AC}) := \max_{\Phi:\mathcal{D}_{C} \to \mathcal{D}_B}F(\rho_{AB}, \mathcal{I}_A\otimes \Phi(\sigma_{AC})).
\end{equation}
Brought to our setup, we see that the local extendibility condition can be phrased in terms of the fidelity of recovery:
\begin{equation}
    \min_{\Phi: \mathcal{D}_{BE}\to \mathcal{D}_{BC}} \bures{\rho_{ABC}}{\mathcal{I}_A \otimes \Phi(\rho_{ABE})} = \sqrt{1-F_{BE\to BC}(\rho_{ABC} \| \sigma_{ABE})}.
\end{equation}
This formulation is useful because the fidelity of recovery can be recast as a semi-definite programming (SDP)~\cite{Berta2016}, for which there is an efficient algorithm for solving the optimization problem~\cite{Boyd2004}.

More specifically, the square root of the fidelity of recovery is a solution to the following optimization problem~\cite{Berta2016}:
\begin{equation}
    \begin{aligned}
        \text{maximize}&: \frac{1}{2}\text{Tr}(Z_{AB} + Z_{AB}^{\dagger}) \\
        \text{subject to}&: \tau_{ABD}\geq 0, Z_{AB} \in \mathcal{L}_{AB} \\
        &\quad \text{Tr}_B(\tau_{ABD})= I_{AD} \\
        &\quad \begin{pmatrix}
            \rho_{AB} & Z_{AB}\\
            Z_{AB}^{\dagger} & \text{Tr}_D(\sqrt{\sigma_{AD}} \tau_{ABD}\sqrt{\sigma_{AD}})
        \end{pmatrix}
        \geq 0.
    \end{aligned}
    \label{eq:sdp}
\end{equation}
Here $\mathcal{L}_{AB}$ is a set of linear operators acting on $\mathcal{H}_A \otimes \mathcal{H}_B$, $\sigma_{ACD}$ is a purification of $\sigma_{AC}$, and $\tau_{ABD}$ is a Choi-Jamio\l{}kowski state~\cite{Choi1975,Jamiolkowski}. That is, $\tau_{ABD}:=  \Phi(\Psi_{AD:C})$, where $\Psi_{AD:C}$ is an unnormalized maximally entangled state between $AD$ and $C$. Eq.~\eqref{eq:sdp} can be brought into the standard SDP form~\cite{Berta2016}, after which it can be solved in time polynomial in the dimensions and $\log (1/\delta)$, where $\delta$ is the target precision. In particular, the minimizer $\tau_{ABD}$ can be obtained, from which the channel $\Phi$ can be obtained via the Choi-Jamio\l{}kowski isomorphism~\cite{Choi1975,Jamiolkowski}. 

\begin{proposition}
\label{eq:for_sdp}
    The fidelity of recovery $F_{C\to B}(\rho_{AB}\| \sigma_{AC})$ (and the channel that achieves the minimum) can be computed in time polynomial in the dimension of $A, B, C$, and $\log 1/\delta$, where $\delta$ is the target additive precision. 
\end{proposition}

Therefore, by solving an SDP of the form of Eq.~\eqref{eq:sdp}, we can compute the degree to which the given state is locally extendible. Moreover, we can obtain a channel that minimizes the error in Eq~\eqref{eq:locally_extendible_approximate}. Both take time polynomial in the dimension of $BCDE$ and $\log (1/\delta)$, where $\delta$ is the target precision. Later we shall choose the subsystems $BCDE$ so that their dimension is $\mathcal{O}(1)$. In such cases, both tasks take time at most polynomial in $\log (1/\delta)$. In particular, we arrive at the following corollary.

\begin{corollary}\label{corollary:finding_extender}
    Let $\rho_{BCDE}$ be a $\epsilon$-extendible state from $BE$ to $BC$. There is an algorithm that, given a description of $\rho_{BCDE}$, outputs a channel $\Phi: \mathcal{D}_{BE} \to \mathcal{D}_{BC}$ such that
    \begin{equation}
        \bures{\rho_{ABC}}{\mathcal{I}_A\otimes \Phi(\rho_{ABE})}\leq 2\epsilon,
    \end{equation}
    where $\rho_{ABCDE}$ is a purification of $\rho_{BCDE}$, in time polynomial in the dimension of $BCDE$ and $\log 1/\epsilon$.
\end{corollary}

We remark that, though the SDP formulation of extendibility has the advantage of being efficiently solvable, there are setups in which a more brute-force approach is preferred. Since our ultimate goal is circuit learning, the channel obtained from the solution of the SDP must be converted into a unitary using Stinespring dilation~\cite{Stinespring1955}. Unfortunately, we can only guarantee the depth of this unitary (once decomposed into a circuit) to be at most exponential in the number of qubits in the subsystems involved. However, in many of the setups we consider [Section~\ref{sec:summary}], we will be guaranteed that there exists a unitary with a much lower depth. With such a guarantee, one can simply brute-force search over all possible circuits with the prescribed depth. This method will have a higher complexity compared to the SDP approach, but will succeed in finding a circuit of lower depth. We shall revisit this discussion in Section~\ref{subsec:learning_extending_map}.

\subsection{Continuity}
\label{subsec:stability_approximation}

In Section~\ref{subsec:checking_extendibility}, we introduced a method to check the extendibility of a given state. Moreover, we also introduced a method to find a channel that minimizes the optimization problem in Eq.~\eqref{eq:locally_extendible_approximate} (up to an error which we can choose). While this is useful, to use these results for the circuit learning problem, we need to prove some continuity bounds. This is because, even if a state of interest is extendible, we will at best have a finite-precision approximation of the state. For instance, if we were to obtain the state using state tomography, there will be a statistical error associated with the finite number of samples we take. If were to use an exact diagonalization to compute the state, we will still have only a finite number bits to specify the state. In both cases, the state we obtain will merely be an approximation of the true state. 

To that end, we prove two results in this Section. First, we prove that if a state is close to an extendible state, it is also approximately extendible [Lemma~\ref{lemma:stability}]. Second, we prove that a good extending map for one state is also a good extending map for another state if the two states are close to each other [Lemma~\ref{lemma:extender}]. Therefore, if we have a good enough approximation of a given state, an extending map for the approximate state is an approximate extending map of the given state. 

\begin{lemma}
\label{lemma:stability}
Suppose $\rho_{BCDE}$ is $\epsilon$-extendible from $BE$ to $BC$ and $\bures{\rho_{BCDE}}{\rho_{BCDE}'}\leq \delta$ for some $\delta \geq 0$. Then, $\rho_{BCDE}'$ is $(\epsilon + 2\delta)$-extendible.
\end{lemma}
\begin{proof}
\sloppy
    By Uhlmann's theorem~\cite{Uhlmann1976}, there exist purifications $\rho_{ABCDE}$ and $\rho_{ABCDE}'$ that satisfies $\bures{\rho_{ABCDE}}{\rho_{ABCDE}'}\leq \delta$. Because the Bures distance is non-increasing under a partial trace (over $D$), we get:
    \begin{equation}
        \begin{aligned}
            \bures{\rho_{ABC}}{\rho_{ABC}'} \leq \delta \quad \text{ and } \quad \bures{\rho_{ABE}}{\rho_{ABE}'} \leq \delta. \label{eq:uhlmann_and_then_partial_trace}
        \end{aligned}
    \end{equation}

    \noindent By the extendibility assumption, there exists a channel $\Phi_0: \mathcal{D}_{BE} \to \mathcal{D}_{BC}$ such that
    \begin{equation}
        \bures{\rho_{ABC}}{\mathcal{I}_A\otimes \Phi_0(\rho_{ABE})}\leq \epsilon.
    \end{equation}
    Using the triangle inequality for the Bures distance,
    \begin{equation}
    \begin{aligned}
        \bures{\rho_{ABC}'}{\mathcal{I}_A\otimes \Phi_0(\rho_{ABE}')} &\leq \bures{\rho_{ABC}'} {\rho_{ABC}} + \bures{\rho_{ABC}}{\mathcal{I}_A\otimes \Phi_0(\rho_{ABE})} \\&+ \bures{\mathcal{I}_A\otimes \Phi_0(\rho_{ABE})}{\mathcal{I}_A\otimes \Phi_0(\rho_{ABE}')} \\
        &\leq \epsilon + 2\delta,
    \end{aligned}
    \end{equation}
\end{proof}

Using the exact same argument, we also obtain the following result.
\begin{lemma}
\label{lemma:extender}
    Suppose $\Phi:\mathcal{D}_{BE}\to \mathcal{D}_{BC}$ realizes an $\epsilon$-extension of $\rho_{BCDE}$ and $\mathcal{B}(\rho_{BCDE}, \rho_{BCDE}')\leq \delta$ for some $\delta \geq 0$. Then, $\Phi$ realizes an $(\epsilon+2\delta)$-extension of $\rho_{BCDE}'$.
\end{lemma}

\subsection{Learning an extending map}
\label{subsec:learning_extending_map}

Thanks to the SDP formulation of the fidelity of recovery [Eq.~\eqref{eq:sdp}] and its continuity [Section~\ref{subsec:stability_approximation}], we can now describe algorithms that learn an extending map associated with a density matrix $\rho_{BCDE}$, given its approximate description.

We consider two situations. The first situation is when no information about the extending map is known to the user. In this case, one can simply solve the SDP [Eq.~\eqref{eq:sdp}] from which the extending map is obtained. 

\begin{lemma}
\label{lemma:extender_sdp}
Let $\rho_{BCDE}$ be an $\epsilon$-extendible state from $BE$ to $BC$. Given a description of $\rho_{BCDE}'$ such that $\mathcal{B}(\rho_{BCDE}, \rho_{BCDE}')\leq \delta$, there is an algorithm running in time polynomial in the dimension of $BCDE$, $\log 1/\epsilon$, and $\log 1/\delta$ that outputs a $(2\epsilon+6\delta)$-extending map of $\rho_{BCDE}$.
\end{lemma}
\begin{proof}
    From Lemma~\ref{lemma:stability}, it follows that $\rho_{BCDE}'$ is $\epsilon + 2\delta$-extendible. From Corollary~\ref{corollary:finding_extender}, it follows that there is an algorithm with a runtime polynomial in  $BCDE$, $\log 1/\epsilon$, and $\log 1/\delta$ that outputs an $(2\epsilon + 4\delta)$-extending map of $\rho_{BCDE}'$, say $\Phi$. Using Lemma~\ref{lemma:extender}, we conclude that $\Phi$ is an $(2\epsilon + 6\delta)$-extending map of $\rho_{BCDE}$.
\end{proof}

The second situation is when we know about the extending map, specifically, its circuit architecture. In this case, one can simply brute-force search over all possible circuits consistent with this circuit architecture. (\textbf{Note}: Below we denote the dimension of a subsystem $X$ as $d_X$.)
\begin{lemma}
\label{lemma:extender_less_gates}
    Let $\Phi_{BE\to BC}$ be a channel that $\epsilon$-extends $\rho_{BCDE}$ from $BE$ to $BC$. Suppose there is a dilation of $\Phi_{BE\to BC}$ such that
    \begin{equation}
        \Phi_{BE\to BC}(\cdot) = \text{Tr}_{EF}( U_{BCEF}(\omega_{F}\otimes (\cdot)) U_{BCEF}^{\dagger})
    \end{equation}
    for a known state $\omega_F$ and a unitary $U_{BCEF}$ consisting of $M$ gates, with a known architecture. Given $\rho_{BCDE}'$ such that $\mathcal{B}(\rho_{BCDE}, \rho_{BCDE}')\leq \epsilon$, there is an algorithm running in time $\text{poly}(d_{BCDE})\cdot \left(\frac{M}{\epsilon} \right)^{\mathcal{O}(M)}$ that finds $U_{BCEF}'$ with the following properties:
    \begin{enumerate}
        \item $U_{BCEF}'$ and $U_{BCEF}$ have the same architecture.
        \item $\mathcal{B}(\rho_{ABC}, \mathcal{I}_A\otimes \Phi_{BE\to BC}'(\rho_{ABE}))\leq 3\epsilon$, where 
        \begin{equation}
        \Phi_{BE\to BC}'(\cdot) = \text{Tr}_{EF}( U_{BCEF}'(\omega_{F}\otimes (\cdot)) {U_{BCEF}'}^{\dagger}).
    \end{equation}
    \end{enumerate}
\end{lemma}
\begin{proof}
Recall that the $\epsilon$-net of $SU(m)$ for any $m=\mathcal{O}(1)$ has a cardinality of at most $(c_0/\epsilon)^{c_1}$ for some constants $c_0, c_1>0$. Therefore, by choosing the gate set from the $\delta/M$-net, one can guarantee that there exists a $U_{BCEF}'$ created from that gate set that satisfies $\| U_{BCEF} - U_{BCEF}' \|\leq \delta$. For such $U_{BCEF}'$, we obtain:
\begin{equation}
    \mathcal{B}(\rho_{ABC}', \mathcal{I}_A\otimes \Phi_{BE\to BC}'(\rho_{ABE}')) \leq 2\delta, \label{eq:requirement}
\end{equation}
where $\rho_{ABCDE}'$ is a purification of $\rho_{BCDE}'$. 

Since the search space has a cardinality of $\left(\frac{M}{\delta} \right)^{\mathcal{O}(M)}$, with brute-force search one can in fact find such $U_{BCEF}'$ by checking Eq.~\eqref{eq:requirement}. Evaluating the left hand side of Eq.~\eqref{eq:requirement} takes $\text{poly}(d_{BCDE})$ time, and as such, the computation time is bounded by $\text{poly}(d_{BCDE})\cdot \left(\frac{M}{\delta} \right)^{\mathcal{O}(M)}$. Using the continuity [Lemma~\ref{lemma:extender}], we conclude that $\Phi_{BE\to BC}'$ is realizes an $(2\epsilon + 2\delta)$-extension of $\rho_{BCDE}$. By setting $\delta=\frac{1}{2}\epsilon$, the claim follows.
\end{proof}

There are pros and cons to Lemma~\ref{lemma:extender_sdp} and Lemma~\ref{lemma:extender_less_gates}. The overall running time of the algorithm in Lemma~\ref{lemma:extender_sdp} is lower than that of Lemma~\ref{lemma:extender_less_gates}. On the other hand, the gate complexity of the extending map obtained from Lemma~\ref{lemma:extender_less_gates} may be lower than that of Lemma~\ref{lemma:extender_sdp}. (This can happen when there is a promise that the gate complexity of the extending map is low, which is the setups considered in Section~\ref{sec:extendibility_quantum_phase}.) Therefore, in order to reduce the overall running time of the algorithm, it will be better to use Lemma~\ref{lemma:extender_sdp}. For minimizing the gate complexity of the extending map, Lemma~\ref{lemma:extender_less_gates} should be preferred.

For later purposes, we shall focus on learning an approximate extending map of an extendible state, from copies of the extendible state. Since solving this problem entails using a state tomography algorithm (in terms of Bures distance), let us briefly review the relevant literature. In the state tomography literature, the Bures distance is studied under the name of \emph{infidelity}. (More precisely, infidelity is square of the Bures distance.) Recently, Flammia and O'Donnell came up with a protocol that estimates infidelity up to a precision of $\epsilon$ with complexity $\widetilde{O}(d_{\mathcal{H}}^3/\epsilon)$, where $d_{\mathcal{H}}$ is the dimension of the underlying Hilbert space and $\widetilde{O}$ hides the polylogarithmic factor~\cite[Corollary 1.10 and (2)]{Flammia2024quantumchisquared}. We remark that there are other methods that improve the dependence on the dimension~\cite{haah2016sample,o2017efficient}, but they require measurement over multiple copies. On the other hand, the method in Ref.~\cite{Flammia2024quantumchisquared} only requires single-copy measurements. Lastly, one can use the standard median trick to ensure the correctness with failure probability at most $\delta$, with an overhead of at most $\mathcal{O}(\log 1/\delta)$. Thus the overall complexity would be $\widetilde{\mathcal{O}}(d_{\mathcal{H}}^3 \log (1/\delta)/\epsilon)$. 

\begin{lemma}
\label{lemma:extend_from_samples_fast}
\sloppy
    Let $\rho_{BCDE}$ be an extendible state from $BE$ to $BC$. There is an algorithm that uses $\widetilde{\mathcal{O}}(\text{poly}(d_{BCDE})\log (1/\delta)/\epsilon^2)$ single-copy measurements of $\rho_{BCDE}$ and outputs an $\epsilon$-extending map of $\rho_{BCDE}$ with probability of at least $1-\delta$, with the runtime of $\widetilde{\mathcal{O}}(\text{poly}(d_{BCDE}) \log (1/\delta)/\epsilon^2)$. 
\end{lemma}

\begin{lemma}
    \label{lemma:extend_from_samples_slow}
     Let $\Phi_{BE\to BC}$ be a channel that extends $\rho_{BCDE}$ from $BE$ to $BC$. Suppose there is a dilation of $\Phi_{BE\to BC}$ such that
    \begin{equation}
        \Phi_{BE\to BC}(\cdot) = \text{Tr}_{EF}( U_{BCEF}(\omega_{F}\otimes (\cdot)) U_{BCEF}^{\dagger})
    \end{equation}
    for a known state $\omega_F$ and a unitary $U_{BCEF}$ consisting of $M$ gates, with a known architecture. There is an algorithm that uses $\widetilde{O}(\text{poly}(d_{BCDE}) \log(1/\delta)/\epsilon^2)$ single-copy measurements of $\rho_{BCDE}$ and runs in time $\text{poly}(d_{BCDE})\cdot \left(\frac{M}{\epsilon} \right)^{\mathcal{O}(M)}$ that finds $U_{BCEF}'$ with probability at least $1-\delta$, with the following properties:
    \begin{enumerate}
        \item $U_{BCEF}'$ and $U_{BCEF}$ have the same architecture.
        \item $\mathcal{B}(\rho_{ABC}, \mathcal{I}_A\otimes \Phi_{BE\to BC}'(\rho_{ABE}))\leq \epsilon$, where 
        \begin{equation}
        \Phi_{BE\to BC}'(\cdot) = \text{Tr}_{EF}( U_{BCEF}'(\omega_{F}\otimes (\cdot)) {U_{BCEF}'}^{\dagger}).
    \end{equation}
    \end{enumerate}
\end{lemma}

\subsection{Stability under constant-depth circuits}\label{subsec:stability_circuit}

So far, we have discussed general properties of locally extendible states. In this Section, we restrict our attention to extendible states on a finite-dimensional lattice, but in turn show another important property of such extendible states: the stability of extendibility under constant-depth circuits. We will consider a state $\rho$ that is supported on a finite-dimensional lattice and satisfies Eq.~\eqref{eq:locally_extendible} for an appropriate partition of the system. For such a state, we will apply a geometrically local, constant-depth circuit $U$. Then, we will show that the post-circuit state $U \rho U^{\dagger}$ also satisfies Eq.~\eqref{eq:locally_extendible} for another, topologically equivalent partition.

Our result consists of two theorems, respectively corresponding to each distinct partition in Fig.~\ref{fig:extendible_state_examples}. We first prove Theorem~\ref{theorem:circuit_stability_1}, the theorem concerning the partition in Fig.~\ref{fig:extendible_state_examples}(a)-(b). We then provide a shorter proof for Theorem~\ref{theorem:circuit_stability_3}, pertaining to the partitions in Fig.~\ref{fig:extendible_state_examples}(c)-(d), which follows a reasoning similar to Theorem~\ref{theorem:circuit_stability_1}. For concreteness, we prove our theorems in two spatial dimensions. At the end of this Section, we briefly discuss a generalization of the theorems to higher dimensions, alongside the significance of our theorems towards circuit learning.

\begin{figure}[H]
    \centering
    \begin{tikzpicture}[every path/.style={thick},scale=0.32]
    \begin{scope}
    \foreach \x in {-7,...,28}
    {
    \draw[] (\x * 0.5+0.25, 0.7) -- ++ (0,2.5);
    }
    \draw[fill=yellow] (-2.5-0.45, 1.0) -- ++ (0.9, 0) -- ++ (0, 0.4) -- ++ (-0.9, 0) -- cycle;
    \draw[fill=yellow] (-1.5-0.45, 1.0) -- ++ (0.9, 0) -- ++ (0, 0.4) -- ++ (-0.9, 0) -- cycle;
    \draw[fill=yellow] (-0.5-0.45, 1.0) -- ++ (0.9, 0) -- ++ (0, 0.4) -- ++ (-0.9, 0) -- cycle;
    \draw[fill=yellow] (0.5-0.45, 1.0) -- ++ (0.9, 0) -- ++ (0, 0.4) -- ++ (-0.9, 0) -- cycle;
    \draw[fill=yellow] (1.5-0.45, 1.0) -- ++ (0.9, 0) -- ++ (0, 0.4) -- ++ (-0.9, 0) -- cycle;
    \draw[fill=yellow] (2.5-0.45, 1.0) -- ++ (0.9, 0) -- ++ (0, 0.4) -- ++ (-0.9, 0) -- cycle;
    \draw[fill=green!60!white] (3.5-0.45, 1.0) -- ++ (0.9, 0) -- ++ (0, 0.4) -- ++ (-0.9, 0) -- cycle;
    \draw[fill=green!60!white] (4.5-0.45, 1.0) -- ++ (0.9, 0) -- ++ (0, 0.4) -- ++ (-0.9, 0) -- cycle;
    \draw[fill=green!60!white] (5.5-0.45, 1.0) -- ++ (0.9, 0) -- ++ (0, 0.4) -- ++ (-0.9, 0) -- cycle;
    \draw[fill=green!60!white] (6.5-0.45, 1.0) -- ++ (0.9, 0) -- ++ (0, 0.4) -- ++ (-0.9, 0) -- cycle;
    \draw[fill=green!60!white] (7.5-0.45, 1.0) -- ++ (0.9, 0) -- ++ (0, 0.4) -- ++ (-0.9, 0) -- cycle;
    \draw[fill=yellow] (8.5-0.45, 1.0) -- ++ (0.9, 0) -- ++ (0, 0.4) -- ++ (-0.9, 0) -- cycle;
    \draw[fill=yellow] (9.5-0.45, 1.0) -- ++ (0.9, 0) -- ++ (0, 0.4) -- ++ (-0.9, 0) -- cycle;
    \draw[fill=yellow] (10.5-0.45, 1.0) -- ++ (0.9, 0) -- ++ (0, 0.4) -- ++ (-0.9, 0) -- cycle;
    \draw[fill=yellow] (11.5-0.45, 1.0) -- ++ (0.9, 0) -- ++ (0, 0.4) -- ++ (-0.9, 0) -- cycle;
    \draw[fill=yellow] (12.5-0.45, 1.0) -- ++ (0.9, 0) -- ++ (0, 0.4) -- ++ (-0.9, 0) -- cycle;
    \draw[fill=yellow] (13.5-0.45, 1.0) -- ++ (0.9, 0) -- ++ (0, 0.4) -- ++ (-0.9, 0) -- cycle;
    \draw[fill=yellow] (-3-0.45, 1.0) -- ++ (0.4, 0) -- ++ (0, 0.4) -- ++ (-0.4, 0);
    \draw[fill=yellow] (14+0.45, 1.0) -- ++ (-0.4, 0) -- ++ (0, 0.4) -- ++ (0.4, 0);
    \draw[fill=yellow] (-3-0.45, 1.5) -- ++ (0.9, 0) -- ++ (0, 0.4) -- ++ (-0.9, 0) -- cycle;
    \draw[fill=yellow] (-2-0.45, 1.5) -- ++ (0.9, 0) -- ++ (0, 0.4) -- ++ (-0.9, 0) -- cycle;
    \draw[fill=yellow] (-1-0.45, 1.5) -- ++ (0.9, 0) -- ++ (0, 0.4) -- ++ (-0.9, 0) -- cycle;
    \draw[fill=yellow] (0-0.45, 1.5) -- ++ (0.9, 0) -- ++ (0, 0.4) -- ++ (-0.9, 0) -- cycle;
    \draw[fill=yellow] (1-0.45, 1.5) -- ++ (0.9, 0) -- ++ (0, 0.4) -- ++ (-0.9, 0) -- cycle;
    \draw[fill=yellow] (2-0.45, 1.5) -- ++ (0.9, 0) -- ++ (0, 0.4) -- ++ (-0.9, 0) -- cycle;
    \draw[fill=yellow] (3-0.45, 1.5) -- ++ (0.9, 0) -- ++ (0, 0.4) -- ++ (-0.9, 0) -- cycle;
    \draw[fill=green!60!white] (4-0.45, 1.5) -- ++ (0.9, 0) -- ++ (0, 0.4) -- ++ (-0.9, 0) -- cycle;
    \draw[fill=green!60!white] (5-0.45, 1.5) -- ++ (0.9, 0) -- ++ (0, 0.4) -- ++ (-0.9, 0) -- cycle;
    \draw[fill=green!60!white] (6-0.45, 1.5) -- ++ (0.9, 0) -- ++ (0, 0.4) -- ++ (-0.9, 0) -- cycle;
    \draw[fill=green!60!white] (7-0.45, 1.5) -- ++ (0.9, 0) -- ++ (0, 0.4) -- ++ (-0.9, 0) -- cycle;
    \draw[fill=yellow] (8-0.45, 1.5) -- ++ (0.9, 0) -- ++ (0, 0.4) -- ++ (-0.9, 0) -- cycle;
    \draw[fill=yellow] (9-0.45, 1.5) -- ++ (0.9, 0) -- ++ (0, 0.4) -- ++ (-0.9, 0) -- cycle;
    \draw[fill=yellow] (10-0.45, 1.5) -- ++ (0.9, 0) -- ++ (0, 0.4) -- ++ (-0.9, 0) -- cycle;
    \draw[fill=yellow] (11-0.45, 1.5) -- ++ (0.9, 0) -- ++ (0, 0.4) -- ++ (-0.9, 0) -- cycle;
    \draw[fill=yellow] (12-0.45, 1.5) -- ++ (0.9, 0) -- ++ (0, 0.4) -- ++ (-0.9, 0) -- cycle;
    \draw[fill=yellow] (13-0.45, 1.5) -- ++ (0.9, 0) -- ++ (0, 0.4) -- ++ (-0.9, 0) -- cycle;
    \draw[fill=yellow] (14-0.45, 1.5) -- ++ (0.9, 0) -- ++ (0, 0.4) -- ++ (-0.9, 0) -- cycle;
    \draw[fill=yellow] (-2.5-0.45, 2.0) -- ++ (0.9, 0) -- ++ (0, 0.4) -- ++ (-0.9, 0) -- cycle;
    \draw[fill=yellow] (-1.5-0.45, 2.0) -- ++ (0.9, 0) -- ++ (0, 0.4) -- ++ (-0.9, 0) -- cycle;
    \draw[fill=yellow] (-0.5-0.45, 2.0) -- ++ (0.9, 0) -- ++ (0, 0.4) -- ++ (-0.9, 0) -- cycle;
    \draw[fill=yellow] (0.5-0.45, 2.0) -- ++ (0.9, 0) -- ++ (0, 0.4) -- ++ (-0.9, 0) -- cycle;
    \draw[fill=yellow] (1.5-0.45, 2.0) -- ++ (0.9, 0) -- ++ (0, 0.4) -- ++ (-0.9, 0) -- cycle;
    \draw[fill=yellow] (2.5-0.45, 2.0) -- ++ (0.9, 0) -- ++ (0, 0.4) -- ++ (-0.9, 0) -- cycle;
    \draw[fill=yellow] (3.5-0.45, 2.0) -- ++ (0.9, 0) -- ++ (0, 0.4) -- ++ (-0.9, 0) -- cycle;
    \draw[fill=green!60!white] (4.5-0.45, 2.0) -- ++ (0.9, 0) -- ++ (0, 0.4) -- ++ (-0.9, 0) -- cycle;
    \draw[fill=green!60!white] (5.5-0.45, 2.0) -- ++ (0.9, 0) -- ++ (0, 0.4) -- ++ (-0.9, 0) -- cycle;
    \draw[fill=green!60!white] (6.5-0.45, 2.0) -- ++ (0.9, 0) -- ++ (0, 0.4) -- ++ (-0.9, 0) -- cycle;
    \draw[fill=yellow] (7.5-0.45, 2.0) -- ++ (0.9, 0) -- ++ (0, 0.4) -- ++ (-0.9, 0) -- cycle;
    \draw[fill=yellow] (8.5-0.45, 2.0) -- ++ (0.9, 0) -- ++ (0, 0.4) -- ++ (-0.9, 0) -- cycle;
    \draw[fill=yellow] (9.5-0.45, 2.0) -- ++ (0.9, 0) -- ++ (0, 0.4) -- ++ (-0.9, 0) -- cycle;
    \draw[fill=yellow] (10.5-0.45, 2.0) -- ++ (0.9, 0) -- ++ (0, 0.4) -- ++ (-0.9, 0) -- cycle;
    \draw[fill=yellow] (11.5-0.45, 2.0) -- ++ (0.9, 0) -- ++ (0, 0.4) -- ++ (-0.9, 0) -- cycle;
    \draw[fill=yellow] (12.5-0.45, 2.0) -- ++ (0.9, 0) -- ++ (0, 0.4) -- ++ (-0.9, 0) -- cycle;
    \draw[fill=yellow] (13.5-0.45, 2.0) -- ++ (0.9, 0) -- ++ (0, 0.4) -- ++ (-0.9, 0) -- cycle;
    \draw[fill=yellow] (-3-0.45, 2.0) -- ++ (0.4, 0) -- ++ (0, 0.4) -- ++ (-0.4, 0);
    \draw[fill=yellow] (14+0.45, 2.0) -- ++ (-0.4, 0) -- ++ (0, 0.4) -- ++ (0.4, 0);
    \draw[fill=yellow] (-3-0.45, 2.5) -- ++ (0.9, 0) -- ++ (0, 0.4) -- ++ (-0.9, 0) -- cycle;
    \draw[fill=yellow] (-2-0.45, 2.5) -- ++ (0.9, 0) -- ++ (0, 0.4) -- ++ (-0.9, 0) -- cycle;
    \draw[fill=yellow] (-1-0.45, 2.5) -- ++ (0.9, 0) -- ++ (0, 0.4) -- ++ (-0.9, 0) -- cycle;
    \draw[fill=yellow] (0-0.45, 2.5) -- ++ (0.9, 0) -- ++ (0, 0.4) -- ++ (-0.9, 0) -- cycle;
    \draw[fill=yellow] (1-0.45, 2.5) -- ++ (0.9, 0) -- ++ (0, 0.4) -- ++ (-0.9, 0) -- cycle;
    \draw[fill=yellow] (2-0.45, 2.5) -- ++ (0.9, 0) -- ++ (0, 0.4) -- ++ (-0.9, 0) -- cycle;
    \draw[fill=yellow] (3-0.45, 2.5) -- ++ (0.9, 0) -- ++ (0, 0.4) -- ++ (-0.9, 0) -- cycle;
    \draw[fill=yellow] (4-0.45, 2.5) -- ++ (0.9, 0) -- ++ (0, 0.4) -- ++ (-0.9, 0) -- cycle;
    \draw[fill=green!60!white] (5-0.45, 2.5) -- ++ (0.9, 0) -- ++ (0, 0.4) -- ++ (-0.9, 0) -- cycle;
    \draw[fill=green!60!white] (6-0.45, 2.5) -- ++ (0.9, 0) -- ++ (0, 0.4) -- ++ (-0.9, 0) -- cycle;
    \draw[fill=yellow] (7-0.45, 2.5) -- ++ (0.9, 0) -- ++ (0, 0.4) -- ++ (-0.9, 0) -- cycle;
    \draw[fill=yellow] (8-0.45, 2.5) -- ++ (0.9, 0) -- ++ (0, 0.4) -- ++ (-0.9, 0) -- cycle;
    \draw[fill=yellow] (9-0.45, 2.5) -- ++ (0.9, 0) -- ++ (0, 0.4) -- ++ (-0.9, 0) -- cycle;
    \draw[fill=yellow] (10-0.45, 2.5) -- ++ (0.9, 0) -- ++ (0, 0.4) -- ++ (-0.9, 0) -- cycle;
    \draw[fill=yellow] (11-0.45, 2.5) -- ++ (0.9, 0) -- ++ (0, 0.4) -- ++ (-0.9, 0) -- cycle;
    \draw[fill=yellow] (12-0.45, 2.5) -- ++ (0.9, 0) -- ++ (0, 0.4) -- ++ (-0.9, 0) -- cycle;
    \draw[fill=yellow] (13-0.45, 2.5) -- ++ (0.9, 0) -- ++ (0, 0.4) -- ++ (-0.9, 0) -- cycle;
    \draw[fill=yellow] (14-0.45, 2.5) -- ++ (0.9, 0) -- ++ (0, 0.4) -- ++ (-0.9, 0) -- cycle;
    \draw [decorate, decoration = {calligraphic brace}] (4.5, 3.275) --  (6.5, 3.275);
    \node[] at (5.5, 4.25) {$S$};
    \draw [decorate, decoration = {calligraphic brace,mirror}] (2.5, 0.625) --  (8.5, 0.625);
    \node[] at (5.5, -0.35) {$S(d)$};
    \node[] at (-4.2, 2.0) {$\cdots$};
    \node[] at (15.2, 2.0) {$\cdots$};
    \draw[->, line width=1pt] (-5.2, 0.9) -- node[left] {\small time} (-5.2, 3.0);
    \node[] at (5.5, -1.95) {(a)};

    \end{scope}
    \begin{scope}[shift={(21.5,0)}]
    \foreach \x in {-7,...,28}
    {
    \draw[] (\x * 0.5+0.25, 0.7) -- ++ (0,2.5);
    }
    \draw[fill=yellow] (-2.5-0.45, 1.0) -- ++ (0.9, 0) -- ++ (0, 0.4) -- ++ (-0.9, 0) -- cycle;
    \draw[fill=yellow] (-1.5-0.45, 1.0) -- ++ (0.9, 0) -- ++ (0, 0.4) -- ++ (-0.9, 0) -- cycle;
    \draw[fill=yellow] (-0.5-0.45, 1.0) -- ++ (0.9, 0) -- ++ (0, 0.4) -- ++ (-0.9, 0) -- cycle;
    \draw[fill=yellow] (0.5-0.45, 1.0) -- ++ (0.9, 0) -- ++ (0, 0.4) -- ++ (-0.9, 0) -- cycle;
    \draw[fill=yellow] (1.5-0.45, 1.0) -- ++ (0.9, 0) -- ++ (0, 0.4) -- ++ (-0.9, 0) -- cycle;
    \draw[fill=yellow] (2.5-0.45, 1.0) -- ++ (0.9, 0) -- ++ (0, 0.4) -- ++ (-0.9, 0) -- cycle;
    \draw[fill=yellow] (3.5-0.45, 1.0) -- ++ (0.9, 0) -- ++ (0, 0.4) -- ++ (-0.9, 0) -- cycle;
    \draw[fill=green!60!white] (4.5-0.45, 1.0) -- ++ (0.9, 0) -- ++ (0, 0.4) -- ++ (-0.9, 0) -- cycle;
    \draw[fill=green!60!white] (5.5-0.45, 1.0) -- ++ (0.9, 0) -- ++ (0, 0.4) -- ++ (-0.9, 0) -- cycle;
    \draw[fill=green!60!white] (6.5-0.45, 1.0) -- ++ (0.9, 0) -- ++ (0, 0.4) -- ++ (-0.9, 0) -- cycle;
    \draw[fill=yellow] (7.5-0.45, 1.0) -- ++ (0.9, 0) -- ++ (0, 0.4) -- ++ (-0.9, 0) -- cycle;
    \draw[fill=yellow] (8.5-0.45, 1.0) -- ++ (0.9, 0) -- ++ (0, 0.4) -- ++ (-0.9, 0) -- cycle;
    \draw[fill=yellow] (9.5-0.45, 1.0) -- ++ (0.9, 0) -- ++ (0, 0.4) -- ++ (-0.9, 0) -- cycle;
    \draw[fill=yellow] (10.5-0.45, 1.0) -- ++ (0.9, 0) -- ++ (0, 0.4) -- ++ (-0.9, 0) -- cycle;
    \draw[fill=yellow] (11.5-0.45, 1.0) -- ++ (0.9, 0) -- ++ (0, 0.4) -- ++ (-0.9, 0) -- cycle;
    \draw[fill=yellow] (12.5-0.45, 1.0) -- ++ (0.9, 0) -- ++ (0, 0.4) -- ++ (-0.9, 0) -- cycle;
    \draw[fill=yellow] (13.5-0.45, 1.0) -- ++ (0.9, 0) -- ++ (0, 0.4) -- ++ (-0.9, 0) -- cycle;
    \draw[fill=yellow] (-3-0.45, 1.0) -- ++ (0.4, 0) -- ++ (0, 0.4) -- ++ (-0.4, 0);
    \draw[fill=yellow] (14+0.45, 1.0) -- ++ (-0.4, 0) -- ++ (0, 0.4) -- ++ (0.4, 0);
    \draw[fill=yellow] (-3-0.45, 1.5) -- ++ (0.9, 0) -- ++ (0, 0.4) -- ++ (-0.9, 0) -- cycle;
    \draw[fill=yellow] (-2-0.45, 1.5) -- ++ (0.9, 0) -- ++ (0, 0.4) -- ++ (-0.9, 0) -- cycle;
    \draw[fill=yellow] (-1-0.45, 1.5) -- ++ (0.9, 0) -- ++ (0, 0.4) -- ++ (-0.9, 0) -- cycle;
    \draw[fill=yellow] (0-0.45, 1.5) -- ++ (0.9, 0) -- ++ (0, 0.4) -- ++ (-0.9, 0) -- cycle;
    \draw[fill=yellow] (1-0.45, 1.5) -- ++ (0.9, 0) -- ++ (0, 0.4) -- ++ (-0.9, 0) -- cycle;
    \draw[fill=yellow] (2-0.45, 1.5) -- ++ (0.9, 0) -- ++ (0, 0.4) -- ++ (-0.9, 0) -- cycle;
    \draw[fill=yellow] (3-0.45, 1.5) -- ++ (0.9, 0) -- ++ (0, 0.4) -- ++ (-0.9, 0) -- cycle;
    \draw[fill=green!60!white] (4-0.45, 1.5) -- ++ (0.9, 0) -- ++ (0, 0.4) -- ++ (-0.9, 0) -- cycle;
    \draw[fill=green!60!white] (5-0.45, 1.5) -- ++ (0.9, 0) -- ++ (0, 0.4) -- ++ (-0.9, 0) -- cycle;
    \draw[fill=green!60!white] (6-0.45, 1.5) -- ++ (0.9, 0) -- ++ (0, 0.4) -- ++ (-0.9, 0) -- cycle;
    \draw[fill=green!60!white] (7-0.45, 1.5) -- ++ (0.9, 0) -- ++ (0, 0.4) -- ++ (-0.9, 0) -- cycle;
    \draw[fill=yellow] (8-0.45, 1.5) -- ++ (0.9, 0) -- ++ (0, 0.4) -- ++ (-0.9, 0) -- cycle;
    \draw[fill=yellow] (9-0.45, 1.5) -- ++ (0.9, 0) -- ++ (0, 0.4) -- ++ (-0.9, 0) -- cycle;
    \draw[fill=yellow] (10-0.45, 1.5) -- ++ (0.9, 0) -- ++ (0, 0.4) -- ++ (-0.9, 0) -- cycle;
    \draw[fill=yellow] (11-0.45, 1.5) -- ++ (0.9, 0) -- ++ (0, 0.4) -- ++ (-0.9, 0) -- cycle;
    \draw[fill=yellow] (12-0.45, 1.5) -- ++ (0.9, 0) -- ++ (0, 0.4) -- ++ (-0.9, 0) -- cycle;
    \draw[fill=yellow] (13-0.45, 1.5) -- ++ (0.9, 0) -- ++ (0, 0.4) -- ++ (-0.9, 0) -- cycle;
    \draw[fill=yellow] (14-0.45, 1.5) -- ++ (0.9, 0) -- ++ (0, 0.4) -- ++ (-0.9, 0) -- cycle;
    \draw[fill=yellow] (-2.5-0.45, 2.0) -- ++ (0.9, 0) -- ++ (0, 0.4) -- ++ (-0.9, 0) -- cycle;
    \draw[fill=yellow] (-1.5-0.45, 2.0) -- ++ (0.9, 0) -- ++ (0, 0.4) -- ++ (-0.9, 0) -- cycle;
    \draw[fill=yellow] (-0.5-0.45, 2.0) -- ++ (0.9, 0) -- ++ (0, 0.4) -- ++ (-0.9, 0) -- cycle;
    \draw[fill=yellow] (0.5-0.45, 2.0) -- ++ (0.9, 0) -- ++ (0, 0.4) -- ++ (-0.9, 0) -- cycle;
    \draw[fill=yellow] (1.5-0.45, 2.0) -- ++ (0.9, 0) -- ++ (0, 0.4) -- ++ (-0.9, 0) -- cycle;
    \draw[fill=yellow] (2.5-0.45, 2.0) -- ++ (0.9, 0) -- ++ (0, 0.4) -- ++ (-0.9, 0) -- cycle;
    \draw[fill=green!60!white] (3.5-0.45, 2.0) -- ++ (0.9, 0) -- ++ (0, 0.4) -- ++ (-0.9, 0) -- cycle;
    \draw[fill=green!60!white] (4.5-0.45, 2.0) -- ++ (0.9, 0) -- ++ (0, 0.4) -- ++ (-0.9, 0) -- cycle;
    \draw[fill=green!60!white] (5.5-0.45, 2.0) -- ++ (0.9, 0) -- ++ (0, 0.4) -- ++ (-0.9, 0) -- cycle;
    \draw[fill=green!60!white] (6.5-0.45, 2.0) -- ++ (0.9, 0) -- ++ (0, 0.4) -- ++ (-0.9, 0) -- cycle;
    \draw[fill=green!60!white] (7.5-0.45, 2.0) -- ++ (0.9, 0) -- ++ (0, 0.4) -- ++ (-0.9, 0) -- cycle;
    \draw[fill=yellow] (8.5-0.45, 2.0) -- ++ (0.9, 0) -- ++ (0, 0.4) -- ++ (-0.9, 0) -- cycle;
    \draw[fill=yellow] (9.5-0.45, 2.0) -- ++ (0.9, 0) -- ++ (0, 0.4) -- ++ (-0.9, 0) -- cycle;
    \draw[fill=yellow] (10.5-0.45, 2.0) -- ++ (0.9, 0) -- ++ (0, 0.4) -- ++ (-0.9, 0) -- cycle;
    \draw[fill=yellow] (11.5-0.45, 2.0) -- ++ (0.9, 0) -- ++ (0, 0.4) -- ++ (-0.9, 0) -- cycle;
    \draw[fill=yellow] (12.5-0.45, 2.0) -- ++ (0.9, 0) -- ++ (0, 0.4) -- ++ (-0.9, 0) -- cycle;
    \draw[fill=yellow] (13.5-0.45, 2.0) -- ++ (0.9, 0) -- ++ (0, 0.4) -- ++ (-0.9, 0) -- cycle;
    \draw[fill=yellow] (-3-0.45, 2.0) -- ++ (0.4, 0) -- ++ (0, 0.4) -- ++ (-0.4, 0);
    \draw[fill=yellow] (14+0.45, 2.0) -- ++ (-0.4, 0) -- ++ (0, 0.4) -- ++ (0.4, 0);
    \draw[fill=yellow] (-3-0.45, 2.5) -- ++ (0.9, 0) -- ++ (0, 0.4) -- ++ (-0.9, 0) -- cycle;
    \draw[fill=yellow] (-2-0.45, 2.5) -- ++ (0.9, 0) -- ++ (0, 0.4) -- ++ (-0.9, 0) -- cycle;
    \draw[fill=yellow] (-1-0.45, 2.5) -- ++ (0.9, 0) -- ++ (0, 0.4) -- ++ (-0.9, 0) -- cycle;
    \draw[fill=yellow] (0-0.45, 2.5) -- ++ (0.9, 0) -- ++ (0, 0.4) -- ++ (-0.9, 0) -- cycle;
    \draw[fill=yellow] (1-0.45, 2.5) -- ++ (0.9, 0) -- ++ (0, 0.4) -- ++ (-0.9, 0) -- cycle;
    \draw[fill=yellow] (2-0.45, 2.5) -- ++ (0.9, 0) -- ++ (0, 0.4) -- ++ (-0.9, 0) -- cycle;
    \draw[fill=green!60!white] (3-0.45, 2.5) -- ++ (0.9, 0) -- ++ (0, 0.4) -- ++ (-0.9, 0) -- cycle;
    \draw[fill=green!60!white] (4-0.45, 2.5) -- ++ (0.9, 0) -- ++ (0, 0.4) -- ++ (-0.9, 0) -- cycle;
    \draw[fill=green!60!white] (5-0.45, 2.5) -- ++ (0.9, 0) -- ++ (0, 0.4) -- ++ (-0.9, 0) -- cycle;
    \draw[fill=green!60!white] (6-0.45, 2.5) -- ++ (0.9, 0) -- ++ (0, 0.4) -- ++ (-0.9, 0) -- cycle;
    \draw[fill=green!60!white] (7-0.45, 2.5) -- ++ (0.9, 0) -- ++ (0, 0.4) -- ++ (-0.9, 0) -- cycle;
    \draw[fill=green!60!white] (8-0.45, 2.5) -- ++ (0.9, 0) -- ++ (0, 0.4) -- ++ (-0.9, 0) -- cycle;
    \draw[fill=yellow] (9-0.45, 2.5) -- ++ (0.9, 0) -- ++ (0, 0.4) -- ++ (-0.9, 0) -- cycle;
    \draw[fill=yellow] (10-0.45, 2.5) -- ++ (0.9, 0) -- ++ (0, 0.4) -- ++ (-0.9, 0) -- cycle;
    \draw[fill=yellow] (11-0.45, 2.5) -- ++ (0.9, 0) -- ++ (0, 0.4) -- ++ (-0.9, 0) -- cycle;
    \draw[fill=yellow] (12-0.45, 2.5) -- ++ (0.9, 0) -- ++ (0, 0.4) -- ++ (-0.9, 0) -- cycle;
    \draw[fill=yellow] (13-0.45, 2.5) -- ++ (0.9, 0) -- ++ (0, 0.4) -- ++ (-0.9, 0) -- cycle;
    \draw[fill=yellow] (14-0.45, 2.5) -- ++ (0.9, 0) -- ++ (0, 0.4) -- ++ (-0.9, 0) -- cycle;
    \draw [decorate, decoration = {calligraphic brace}] (4.5, 3.275) --  (6.5, 3.275);
    \node[] at (5.5, 4.25) {$S$};
    \draw [decorate, decoration = {calligraphic brace,mirror}] (2.5, 0.625) --  (8.5, 0.625);
    \node[] at (5.5, -0.35) {$S(d)$};
    \node[] at (-4.2, 2.0) {$\cdots$};
    \node[] at (15.2, 2.0) {$\cdots$};
    \node[] at (5.5, -1.95) {(b)};
    \end{scope}
    \end{tikzpicture}
    \caption{(a) The past light-cone $U_{S}$ (green) of subsystem $S$ under a geometrically local depth-$d$ circuit $U$. (b) The unitary $F_S$ (green), consisting of the gates whose future light-cone stays within $S(d)$. Here $F_S$ coincides with the future light-cone of $S$. Time flows from bottom to top.}
    \label{fig:past_light_cone}
\end{figure}

Before we continue, let us introduce an important concept -- the \textit{past light-cone} $U_{X}$ of a subsystem $X$ under a geometrically local, depth-$d$ circuit $U$. (In our circuit diagrams, time flows from bottom to top.) We first collect the gates in the last layer of $U$ that act on $X$. Then, we collect the gates in the previous layer that act on $X$ or share a qubit with an already collected gate, and so on, down to the first layer. The product of the collected gates, in the order in which they appear in $U$, is $U_X$. This $U_X$ is a depth-$d$ circuit supported on $X(d)$ [Fig.~\ref{fig:past_light_cone}(a)]. We shall use the following properties of the past light-cone.
\begin{lemma}\label{lemma:light_cone}
    Let $U$ be a geometrically local depth-$d$ circuit on $\Lambda$, $\rho$ a state on $\Lambda$, and $X, S \subseteq \Lambda$.
    \begin{enumerate}
        \item $U = W_X U_X$, where $W_X$ is the product of the remaining gates of $U$ and is supported on $\Lambda\setminus X$.
        \item The post-circuit reduced density matrix on $X$ only depends on $U_X$:
        \begin{equation}\label{eq:past_light_cone_prop2}
            \Tr_{\Lambda \setminus X} \Bigl( U \rho U^{\dagger} \Bigr) = \Tr_{\Lambda \setminus X} \left( U_{X} \rho U_{X}^{\dagger} \right).
        \end{equation}
        \item Let $F_S := W_{\Lambda\setminus S(d)}$. Then $F_S$ is supported on $S(d)$, $U_{\Lambda \setminus S(d)}$ is supported on $\Lambda\setminus S$, and
        \begin{equation}\label{eq:past_light_cone_prop1}
            F_S^{\dagger} \left( U \rho U^{\dagger} \right) F_S = U_{\Lambda\setminus S(d)}\, \rho\, U_{\Lambda\setminus S(d)}^{\dagger}.
        \end{equation}
    \end{enumerate}
\end{lemma}
\begin{proof}
    (i) If a gate belongs to $U_X$, so does any earlier gate that shares a qubit with it. Thus the gates of $U_X$ can be moved before those of $W_X$. Moreover, every gate acting on $X$ belongs to $U_X$. (ii) follows from (i). (iii) Apply (i) to $\Lambda\setminus S(d)$. Then $F_S$ is supported on $S(d)$, and $U_{\Lambda\setminus S(d)}$ is supported on $(\Lambda\setminus S(d))(d)$, which is disjoint from $S$.
\end{proof}

Note that $F_S$ consists of the gates whose future light-cone stays within $S(d)$ [Fig.~\ref{fig:past_light_cone}(b)]. A similar construction was used in Ref.~\cite{Hu2026} for mixed states.

Let us now state and prove Theorem~\ref{theorem:circuit_stability_1} -- extendible states on the partition in Fig.~\ref{fig:extendible_state_examples}(a)-(b) are also extendible under a depth-$d$ circuit $U$, on a topologically equivalent partition. Consider a pure state $\rho$ supported on $ABCDE$ as in Fig.~\ref{fig:circuit_stability_partitions_1}(a), equivalent to the one in Fig.~\ref{fig:extendible_state_examples}(a)-(b). Let $BC\widetilde{D}E$ be a subsystem of $ABCDE$, as defined in Fig.~\ref{fig:circuit_stability_partitions_1}(b), where $\widetilde{D}$ consists of all sites in $D$ that are at distance more than $d$ from $A$. If $\rho_{BC\widetilde{D}E}$ is extendible from $BE$ to $BC$, then the post-circuit state $\rho' := U \rho U^{\dagger}$ has a reduced density matrix $\rho_{B'C'D'E'}'$ that is extendible from $B'E'$ to $B'C'$. Here $A'B'C'D'E'$ is a partition defined in Fig.~\ref{fig:circuit_stability_partitions_1}(c), and is topologically equivalent to $ABCDE$. We assume that the subsystems are thick compared to $d$, so that $\operatorname{dist}(A, C) > d$ and $\operatorname{dist}(B, D) > 2d$. 

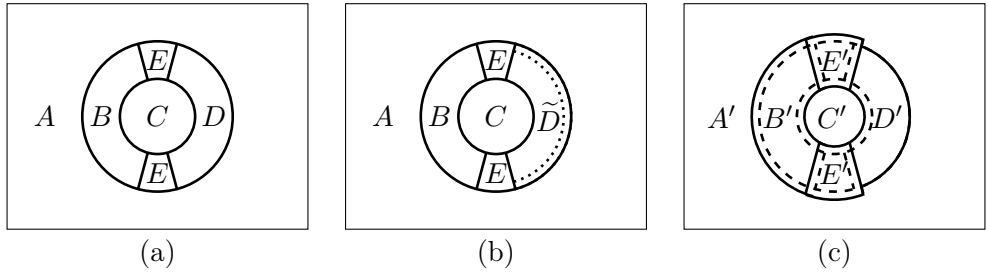
\begin{figure}[h]
    \centering
    \begin{tikzpicture}[scale=1.0]
    \begin{scope}[xshift=0.0cm]
    \draw[fill=white] (-2, -1.5) -- ++ (4, 0) -- ++ (0, 3) -- ++ (-4, 0) -- cycle;
    \draw[line width=1pt] (0,0) circle (0.5cm);
    \draw[line width=1pt] (0,0) circle (1cm);
    \draw[line width=1pt] (75:0.5cm) -- (75:1cm);
    \draw[line width=1pt] (105:0.5cm) -- (105:1cm);
    \draw[line width=1pt] (-75:0.5cm) -- (-75:1cm);
    \draw[line width=1pt] (-105:0.5cm) -- (-105:1cm);
        
    \node[] () at (90:0.75cm) {$E$};
    \node[] () at (-90:0.75cm) {$E$};
    \node[] () at (0,0) {$C$};
    \node[] () at (0:0.75cm) {$D$};
    \node[] () at (180:0.75cm) {$B$};
    \node[] () at (-1.5, 0) {$A$};
    \node[below] () at (0, -1.5) {(a)};
    \end{scope}

    \begin{scope}[xshift=4.5cm]
    \draw[fill=white] (-2, -1.5) -- ++ (4, 0) -- ++ (0, 3) -- ++ (-4, 0) -- cycle;
    \draw[line width=1pt] (0,0) circle (0.5cm);
    \draw[line width=1pt] (0,0) circle (1cm);
    \draw[line width=1pt] (75:0.5cm) -- (75:1cm);
    \draw[line width=1pt] (105:0.5cm) -- (105:1cm);
    \draw[line width=1pt] (-75:0.5cm) -- (-75:1cm);
    \draw[line width=1pt] (-105:0.5cm) -- (-105:1cm);
    \draw[line width=1pt, style=dotted] (75:0.9cm) arc(75:-75:0.9cm) -- (-75:1.0cm) arc(-75:75:1.0cm) -- (75:0.9cm);
    
    \node[] () at (90:0.75cm) {$E$};
    \node[] () at (-90:0.75cm) {$E$};
    \node[] () at (0,0) {$C$};
    \node[] () at (0:0.7cm) {$\widetilde{D}$};
    \node[] () at (180:0.75cm) {$B$};
    \node[] () at (-1.5, 0) {$A$};
    \node[below] () at (0, -1.5) {(b)};
    \end{scope}
    
    \begin{scope}[xshift=9.0cm]
    \draw[fill=white] (-2, -1.5) -- ++ (4, 0) -- ++ (0, 3) -- ++ (-4, 0) -- cycle;
    \draw[line width=1pt] (-105:0.4cm) arc(255:105:0.4cm);
    \draw[line width=1pt] (70:1cm) arc(70:-70:1.0cm);
    \draw[line width=1pt] (0.355291cm, 0.925967cm) -- (0.381173cm, 1.02256cm) arc(69.56:290.44:1.1cm) -- (0.355291cm, -0.925967cm);
    \draw[draw=none] (105:0.5cm) -- (105:1cm) arc(105:255:1cm) -- (255:0.5cm) arc(-105:105:0.5cm);
    \draw[line width=2.0pt, color=white] (-0.381173cm, 1.02256cm) arc(110.4:249.6:1.1cm);
    \draw[line width=1pt] (105:1.1cm) arc(105:255:1.1cm);
    \draw[line width=2.0pt, color=white] (75:0.5cm) arc(75:115:0.5cm);
    \draw[line width=2.0pt, color=white] (70.5:1.0cm) arc(70.5:112:1.0cm);
    \draw[line width=2.0pt, color=white] (-75:0.5cm) arc(-75:-115:0.5cm);
    \draw[line width=2.0pt, color=white] (-70.5:1.0cm) arc(-70.5:-112:1.0cm);
    \draw[line width=2.0pt, color=white] (75:0.5cm) arc(75:-75:0.5cm);
    \draw[line width=1pt] (60:0.4cm) -- (0.381173cm, 1.02256cm);
    \draw[line width=1pt] (-60:0.4cm) -- (0.381173cm, -1.02256cm);
    \draw[line width=1pt] (249.56:1.11cm) -- (240:0.4cm) arc(240:480:0.4cm) -- (-0.381173cm, 1.02256cm);
    \draw[line width=1pt, style=dashed] (0,0) circle (0.5cm);
    \draw[line width=1pt, style=dashed] (0,0) circle (1cm);
    \draw[line width=1pt, style=dashed] (75:0.5cm) -- (75:1cm);
    \draw[line width=1pt, style=dashed] (105:0.5cm) -- (105:1cm);
    \draw[line width=1pt, style=dashed] (-75:0.5cm) -- (-75:1cm);
    \draw[line width=1pt, style=dashed] (-105:0.5cm) -- (-105:1cm);
        
    \node[] () at (90:0.75cm) {$E'$};
    \node[] () at (-90:0.75cm) {$E'$};
    \node[] () at (-0.03,0) {$C'$};
    \node[] () at (0:0.7cm) {$D'$};
    \node[] () at (180:0.75cm) {$B'$};
    \node[] () at (-1.5, 0) {$A'$};
    \node[below] () at (0, -1.5) {(c)};
    \end{scope}
    \end{tikzpicture}
    \caption{(a) The pre-circuit partition $ABCDE$, equivalent to the one in Fig.~\ref{fig:extendible_state_examples}(b). $A$ is a purifying system of $BCDE$. (b) The same partition $ABCDE$, with subsystem $\widetilde{D} \subset D$ shown. $\widetilde{D}$ consists of all sites in $D$ that are at distance more than $d$ from $A$.
    (c) The topologically equivalent post-circuit partition $A'B'C'D'E'$. The individual subsystems are: $B' := BE(d) \setminus E(d)$, $C' := C(-d)$, $D' := BCDE \setminus BE(d)C(-d)$, $E' := E(d)$, and $A'$ is a purifying system of $B'C'D'E'$. (Recall that $S(x)$ for $x \geq 0$ is the subsystem of sites of distance $x$ or less away from $S$, and for $x < 0$ it is the subsystem $S'$ such that $S'(-x) = S$ [Section~\ref{subsec:notation}].) For comparison, the partition $ABCDE$ is overlaid with dashed lines.}
    \label{fig:circuit_stability_partitions_1}
    \end{figure}

\begin{theorem}\label{theorem:circuit_stability_1}
    Let us partition a system $\Lambda$ as in Fig.~\ref{fig:circuit_stability_partitions_1}(a)-(b) for $\Lambda = ABCDE$, and as in Fig.~\ref{fig:circuit_stability_partitions_1}(c) for $\Lambda = A'B'C'D'E'$. Let $\rho$ be a state on $\Lambda$, $U$ be a geometrically local depth-$d$ circuit supported on $\Lambda$, and $\rho' := U \rho U^{\dagger}$ be the post-circuit state. Suppose $\operatorname{dist}(A, C) > d$ and $\operatorname{dist}(B, D) > 2d$. If $\rho_{BC\widetilde{D}E}$ is locally extendible from $BE$ to $BC$, then $\rho_{B'C'D'E'}'$ is locally extendible from $B'E'$ to $B'C'$.
\end{theorem}

\begin{proof}
Our goal is to prove the existence of an extending channel $\Gamma:\mathcal{D}_{B'E'} \rightarrow \mathcal{D}_{B'C'}$ that satisfies:
\begin{equation}\label{eq:circuit_stability_1}
    \rho_{A'B'C'}' = \mathcal{I}_{A'} \otimes \Gamma(\rho_{A'B'E'}').
\end{equation}
\noindent In this proof, we construct such a channel out of three channels -- $\Gamma_{1}, \Gamma_{2}, \Gamma_{3}$. We explicitly apply the three channels to $\rho_{A'B'E'}'$, and show that the resulting state $\mathcal{I}_{A'} \otimes \Gamma_{3}\circ \Gamma_{2} \circ \Gamma_{1}(\rho_{A'B'E'}')$ satisfies Eq.~\eqref{eq:circuit_stability_1}.

An overview of the construction is as follows. By Lemma~\ref{lemma:light_cone}(iii), $U = F_{BE} V_{ACD}$, where $F_{BE}$ is supported on $BE(d) = B'E'$ and $V_{ACD}$ is supported on $ACD$. We start from the reduced density matrix $\rho_{A'B'E'}'$. To it, we first apply $F_{BE}^{\dagger}$, which leaves behind $V_{ACD}\rho V_{ACD}^{\dagger}$. Such $F_{BE}^{\dagger}$ is our first channel $\Gamma_1$. Then, we take the partial trace $\Tr_{CD \setminus C'D'}$. Together with the already applied $\Tr_{C'D'}$ in $\rho_{A'B'E'}'$, $\Tr_{CD \setminus C'D'}$ yields a state with an important property: its reduced density matrix over $BC\widetilde{D}E$ is indistinguishable from $\rho_{BC\widetilde{D}E}$.

Because $\rho_{BC\widetilde{D}E}$ has an extending channel $\Phi$, we can apply $\Phi$ and recover a state identical to the pre-circuit state $\rho_{ABC}$ conjugated by some new unitary $V_{AD\setminus\widetilde{D}}$ (which originates from the unitary $V_{ACD}$). The composition of $\Tr_{CD \setminus C'D'}$ and the extending channel $\Phi$, in that order, is our second channel $\Gamma_2$. Finally, we apply $K$, a unitary consisting of the gates of $U$ that were removed in the previous steps but are in the past light-cone of $A'B'C'$, followed by the partial trace $\Tr_{D'E' \setminus DE}$. This is our third channel $\Gamma_3$, which recovers the desired post-circuit reduced density matrix $\rho_{A'B'C'}'$. The action of the three channels composed is the post-circuit extending channel $\Gamma_{B'E' \rightarrow B'C'}$.

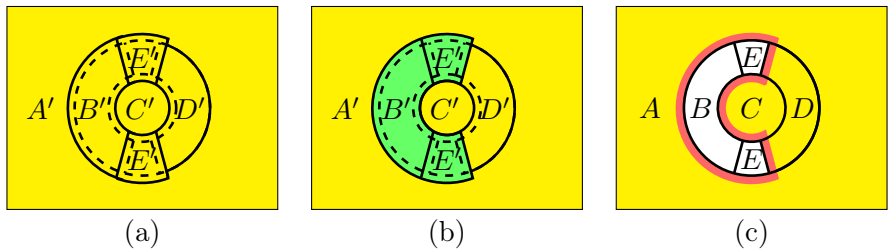
\begin{figure}[h]
    \centering
    \begin{tikzpicture}[scale=0.9]
    \begin{scope}[xshift=0.0cm]
    \draw[fill=yellow] (-2, -1.5) -- ++ (4, 0) -- ++ (0, 3) -- ++ (-4, 0) -- cycle;
    \draw[line width=1pt] (-105:0.4cm) arc(255:105:0.4cm);
    \draw[line width=1pt] (70:1cm) arc(70:-70:1.0cm);
    \draw[line width=1pt] (0.355291cm, 0.925967cm) -- (0.381173cm, 1.02256cm) arc(69.56:290.44:1.1cm) -- (0.355291cm, -0.925967cm);
    \draw[draw=none] (105:0.5cm) -- (105:1cm) arc(105:255:1cm) -- (255:0.5cm) arc(-105:105:0.5cm);
    \draw[line width=2.0pt, color=yellow] (-0.381173cm, 1.02256cm) arc(110.4:249.6:1.1cm);
    \draw[line width=1pt] (105:1.1cm) arc(105:255:1.1cm);
    \draw[line width=2.0pt, color=yellow] (75:0.5cm) arc(75:115:0.5cm);
    \draw[line width=2.0pt, color=yellow] (70.5:1.0cm) arc(70.5:112:1.0cm);
    \draw[line width=2.0pt, color=yellow] (-75:0.5cm) arc(-75:-115:0.5cm);
    \draw[line width=2.0pt, color=yellow] (-70.5:1.0cm) arc(-70.5:-112:1.0cm);
    \draw[line width=2.0pt, color=yellow] (75:0.5cm) arc(75:-75:0.5cm);
    \draw[line width=1pt] (60:0.4cm) -- (0.381173cm, 1.02256cm);
    \draw[line width=1pt] (-60:0.4cm) -- (0.381173cm, -1.02256cm);
    \draw[line width=1pt] (249.56:1.11cm) -- (240:0.4cm) arc(240:480:0.4cm) -- (-0.381173cm, 1.02256cm);
    
    \draw[line width=1pt, style=dashed] (0,0) circle (0.5cm);
    \draw[line width=1pt, style=dashed] (0,0) circle (1cm);
    \draw[line width=1pt, style=dashed] (75:0.5cm) -- (75:1cm);
    \draw[line width=1pt, style=dashed] (105:0.5cm) -- (105:1cm);
    \draw[line width=1pt, style=dashed] (-75:0.5cm) -- (-75:1cm);
    \draw[line width=1pt, style=dashed] (-105:0.5cm) -- (-105:1cm);
    
    \node[] () at (90:0.75cm) {$E'$};
    \node[] () at (-90:0.75cm) {$E'$};
    \node[] () at (-0.03,0) {$C'$};
    \node[] () at (0:0.7cm) {$D'$};
    \node[] () at (180:0.75cm) {$B'$};
    \node[] () at (-1.5, 0) {$A'$};
    \node[below] () at (0, -1.5) {(a)};
    \end{scope}

    \begin{scope}[xshift=4.5cm]
    \draw[fill=yellow] (-2, -1.5) -- ++ (4, 0) -- ++ (0, 3) -- ++ (-4, 0) -- cycle;
    \draw[line width=1pt] (70:1cm) arc(70:-70:1.0cm);
    \draw[line width=1pt, fill=green!60!white] (60:0.4cm) -- (0.381173cm, 1.02256cm) arc(69.56:290.44:1.1cm) -- (300:0.4cm) arc(300:60:0.4cm);
    \draw[line width=1pt, style=dashed] (0.12940952255cm, 0.48296291314cm) arc(75:-75:0.5cm);
    \draw[line width=1pt, style=dashed] (0.12940952255cm, 0.48296291314cm) arc(75:285:0.5cm);
    \draw[line width=1pt, style=dashed] (0,0) circle (1cm);
    \draw[line width=1pt] (0,0) circle (0.4cm);
    \draw[line width=1pt] (120:0.4cm) -- (-0.381173cm, 1.02256cm);
    \draw[line width=1pt] (-120:0.4cm) -- (-0.381173cm, -1.02256cm);
    
    \draw[line width=1pt, style=dashed] (75:0.5cm) -- (75:1cm);
    \draw[line width=1pt, style=dashed] (-75:0.5cm) -- (-75:1cm);
    \draw[line width=1pt, style=dashed] (105:0.5cm) -- (105:1cm);
    \draw[line width=1pt, style=dashed] (-105:0.5cm) -- (-105:1cm);

    \node[] () at (90:0.75cm) {$E'$};
    \node[] () at (-90:0.75cm) {$E'$};
    \node[] () at (-0.03,0) {$C'$};
    \node[] () at (0:0.7cm) {$D'$};
    \node[] () at (180:0.75cm) {$B'$};
    \node[] () at (-1.5, 0) {$A'$};
    \node[below] () at (0, -1.5) {(b)};
    \end{scope}

    \begin{scope}[xshift=9cm]
    \draw[fill=yellow] (-2, -1.5) -- ++ (4, 0) -- ++ (0, 3) -- ++ (-4, 0) -- cycle;
    \draw[line width=1pt] (75:0.5cm) -- (75:1.0cm) arc(75:-75:1.0cm) -- (-75:0.5cm) arc(285:75:0.5cm);
    \draw[line width=1pt, color=red!60!white, fill=red!60!white] (60:0.4cm) -- (0.381173cm, 1.02256cm) arc(69.56:290.44:1.1cm) -- (300:0.4cm) arc(300:60:0.4cm);
    \draw[draw=none, fill=white] (75:0.5cm) -- (75:1cm)  arc(75:285:1cm) -- (285:0.5cm) arc(285:75:0.5cm);
    \draw[line width=1pt] (0.12940952255cm, 0.48296291314cm) arc(75:-75:0.5cm);
    \draw[line width=1pt] (0.12940952255cm, 0.48296291314cm) arc(75:285:0.5cm);
    \draw[line width=1pt] (0,0) circle(1cm);
    \draw[line width=1pt] (75:0.5cm) -- (75:1cm);
    \draw[line width=1pt] (105:0.5cm) -- (105:1cm);
    \draw[line width=1pt] (-75:0.5cm) -- (-75:1cm);
    \draw[line width=1pt] (-105:0.5cm) -- (-105:1cm);
    
    \node[] () at (90:0.75cm) {$E$};
    \node[] () at (-90:0.75cm) {$E$};
    \node[] () at (0,0) {$C$};
    \node[] () at (0:0.75cm) {$D$};
    \node[] () at (180:0.75cm) {$B$};
    \node[] () at (-1.5, 0) {$A$};
    \node[below] () at (0, -1.5) {(c)};
    \end{scope}
    \end{tikzpicture}
    \caption{The first step $\Gamma_1$ illustrated. (a) $\rho' := U\rho U^{\dagger}$ over $A'B'C'D'E' = ABCDE$. The yellow shading depicts application of the gates in $U$. (b) The unitary $F_{BE}$, depicted in green. Note $B'E'$ as its support. (c) $F_{BE}^{\dagger}$ applied, resulting in state $V_{ACD} \rho V_{ACD}^{\dagger}$. The red shading indicates the gates supported on $B'E'$ not cancelled out by $F_{BE}^{\dagger}$.}
    \label{fig:circuit_stability_1_1}
\end{figure}

We now proceed with the detailed proof, beginning with the first step $\mathcal{I}_{A'} \otimes \Gamma_{1}$. (See Fig.~\ref{fig:circuit_stability_1_1} for its graphical illustration.) By Lemma~\ref{lemma:light_cone}(iii) with $S = BE$,
\begin{equation}
    U = F_{BE} V_{ACD},
\end{equation}
where $F_{BE}$ is supported on $BE(d) = B'E'$ and $V_{ACD} := U_{\Lambda\setminus BE(d)}$ is a constant-depth circuit supported on $ACD$ [Fig.~\ref{fig:circuit_stability_1_1}(c)]. We choose $\Gamma_1$ to be $\Gamma_1(\cdot) = F_{BE}^{\dagger} (\cdot )F_{BE}$. Since $F_{BE}$ is supported in the complement of $A'C'D'$, we obtain:
\begin{equation}\label{eq:end_of_step1}
    \mathcal{I}_{A'} \otimes \Gamma_{1}\left(\rho_{A'B'E'}'\right) 
        = \Tr_{C'D'}\left( V_{ACD} \rho V_{ACD}^{\dagger} \right).
\end{equation}

Before we continue, let us draw a schematic circuit diagram which shall be useful for the remainder of the proof. As we apply $F_{BE}^{\dagger}$, we may visualize changes to the gates in $U$ by taking a cross-section view of the circuit diagram. As an example, let us view the horizontal cross-section through $A$, $B$, $C$, and $D$, according to the dotted line in Fig.~\ref{fig:circuit_stability_1_1_plc}(a). The action of $F_{BE}^{\dagger}$ is depicted in Fig.~\ref{fig:circuit_stability_1_1_plc}(b)-(c), providing a schematic illustration for how the gates on this cross-section may have changed. We will update this picture through the following steps of the proof.

\begin{figure}[h]
\centering
\begin{tikzpicture}
    \begin{scope}[xshift=1.0cm, yshift=-0.9cm, scale=0.9]
    \draw[fill=yellow] (-2, -1.5) -- ++ (4, 0) -- ++ (0, 3) -- ++ (-4, 0) -- cycle;
    \draw[line width=1pt] (75:0.5cm) -- (75:1.0cm) arc(75:-75:1.0cm) -- (-75:0.5cm) arc(285:75:0.5cm);
    \draw[line width=1pt, color=red!60!white, fill=red!60!white] (60:0.4cm) -- (0.381173cm, 1.02256cm) arc(69.56:290.44:1.1cm) -- (300:0.4cm) arc(300:60:0.4cm);
    \draw[draw=none, fill=white] (75:0.5cm) -- (75:1cm)  arc(75:285:1cm) -- (285:0.5cm) arc(285:75:0.5cm);
    \draw[line width=1pt] (0.12940952255cm, 0.48296291314cm) arc(75:-75:0.5cm);
    \draw[line width=1pt] (0.12940952255cm, 0.48296291314cm) arc(75:285:0.5cm);
    \draw[line width=1pt] (0,0) circle(1cm);
    \draw[line width=1pt] (75:0.5cm) -- (75:1cm);
    \draw[line width=1pt] (105:0.5cm) -- (105:1cm);
    \draw[line width=1pt] (-75:0.5cm) -- (-75:1cm);
    \draw[line width=1pt] (-105:0.5cm) -- (-105:1cm);

    \draw[line width=1.5pt, style=dotted] (-2, 0) -- (-0.15, 0);
    \draw[line width=1.5pt, style=dotted] (0.15, 0) -- (0.55, 0);
    \draw[line width=1.5pt, style=dotted] (0.95, 0) -- (2, 0);
    
    \node[] () at (90:0.75cm) {$E$};
    \node[] () at (-90:0.75cm) {$E$};
    \node[] () at (0,0) {$C$};
    \node[] () at (0:0.75cm) {$D$};
    \node[] () at (180:0.75cm) {$B$};
    \node[] () at (-1.5, 0) {$A$};
    \node[below] () at (0, -1.5) {(a)};
    \end{scope}
    \begin{scope}[xshift=6.8cm, yshift=-0.5cm, scale=0.45, every path/.style={thick}]
    \foreach \x in {-7,...,28}
    {
    \draw[] (\x * 0.5+0.25, 0.7) -- ++ (0,2.5);
    }
    \draw[fill=yellow] (-2.5-0.45, 1.0) -- ++ (0.9, 0) -- ++ (0, 0.4) -- ++ (-0.9, 0) -- cycle;
    \draw[fill=yellow] (-1.5-0.45, 1.0) -- ++ (0.9, 0) -- ++ (0, 0.4) -- ++ (-0.9, 0) -- cycle;
    \draw[fill=yellow] (-0.5-0.45, 1.0) -- ++ (0.9, 0) -- ++ (0, 0.4) -- ++ (-0.9, 0) -- cycle;
    \draw[fill=green!60!white] (0.5-0.45, 1.0) -- ++ (0.9, 0) -- ++ (0, 0.4) -- ++ (-0.9, 0) -- cycle;
    \draw[fill=green!60!white] (1.5-0.45, 1.0) -- ++ (0.9, 0) -- ++ (0, 0.4) -- ++ (-0.9, 0) -- cycle;
    \draw[fill=green!60!white] (2.5-0.45, 1.0) -- ++ (0.9, 0) -- ++ (0, 0.4) -- ++ (-0.9, 0) -- cycle;
    \draw[fill=yellow] (3.5-0.45, 1.0) -- ++ (0.9, 0) -- ++ (0, 0.4) -- ++ (-0.9, 0) -- cycle;
    \draw[fill=yellow] (4.5-0.45, 1.0) -- ++ (0.9, 0) -- ++ (0, 0.4) -- ++ (-0.9, 0) -- cycle;
    \draw[fill=yellow] (5.5-0.45, 1.0) -- ++ (0.9, 0) -- ++ (0, 0.4) -- ++ (-0.9, 0) -- cycle;
    \draw[fill=yellow] (6.5-0.45, 1.0) -- ++ (0.9, 0) -- ++ (0, 0.4) -- ++ (-0.9, 0) -- cycle;
    \draw[fill=yellow] (7.5-0.45, 1.0) -- ++ (0.9, 0) -- ++ (0, 0.4) -- ++ (-0.9, 0) -- cycle;
    \draw[fill=yellow] (8.5-0.45, 1.0) -- ++ (0.9, 0) -- ++ (0, 0.4) -- ++ (-0.9, 0) -- cycle;
    \draw[fill=yellow] (9.5-0.45, 1.0) -- ++ (0.9, 0) -- ++ (0, 0.4) -- ++ (-0.9, 0) -- cycle;
    \draw[fill=yellow] (10.5-0.45, 1.0) -- ++ (0.9, 0) -- ++ (0, 0.4) -- ++ (-0.9, 0) -- cycle;
    \draw[fill=yellow] (11.5-0.45, 1.0) -- ++ (0.9, 0) -- ++ (0, 0.4) -- ++ (-0.9, 0) -- cycle;
    \draw[fill=yellow] (12.5-0.45, 1.0) -- ++ (0.9, 0) -- ++ (0, 0.4) -- ++ (-0.9, 0) -- cycle;
    \draw[fill=yellow] (13.5-0.45, 1.0) -- ++ (0.9, 0) -- ++ (0, 0.4) -- ++ (-0.9, 0) -- cycle;
    \draw[fill=yellow] (-3-0.45, 1.0) -- ++ (0.4, 0) -- ++ (0, 0.4) -- ++ (-0.4, 0);
    \draw[fill=yellow] (14+0.45, 1.0) -- ++ (-0.4, 0) -- ++ (0, 0.4) -- ++ (0.4, 0);
    \draw[fill=yellow] (-3-0.45, 1.5) -- ++ (0.9, 0) -- ++ (0, 0.4) -- ++ (-0.9, 0) -- cycle;
    \draw[fill=yellow] (-2-0.45, 1.5) -- ++ (0.9, 0) -- ++ (0, 0.4) -- ++ (-0.9, 0) -- cycle;
    \draw[fill=yellow] (-1-0.45, 1.5) -- ++ (0.9, 0) -- ++ (0, 0.4) -- ++ (-0.9, 0) -- cycle;
    \draw[fill=green!60!white] (0-0.45, 1.5) -- ++ (0.9, 0) -- ++ (0, 0.4) -- ++ (-0.9, 0) -- cycle;
    \draw[fill=green!60!white] (1-0.45, 1.5) -- ++ (0.9, 0) -- ++ (0, 0.4) -- ++ (-0.9, 0) -- cycle;
    \draw[fill=green!60!white] (2-0.45, 1.5) -- ++ (0.9, 0) -- ++ (0, 0.4) -- ++ (-0.9, 0) -- cycle;
    \draw[fill=green!60!white] (3-0.45, 1.5) -- ++ (0.9, 0) -- ++ (0, 0.4) -- ++ (-0.9, 0) -- cycle;
    \draw[fill=yellow] (4-0.45, 1.5) -- ++ (0.9, 0) -- ++ (0, 0.4) -- ++ (-0.9, 0) -- cycle;
    \draw[fill=yellow] (5-0.45, 1.5) -- ++ (0.9, 0) -- ++ (0, 0.4) -- ++ (-0.9, 0) -- cycle;
    \draw[fill=yellow] (6-0.45, 1.5) -- ++ (0.9, 0) -- ++ (0, 0.4) -- ++ (-0.9, 0) -- cycle;
    \draw[fill=yellow] (7-0.45, 1.5) -- ++ (0.9, 0) -- ++ (0, 0.4) -- ++ (-0.9, 0) -- cycle;
    \draw[fill=yellow] (8-0.45, 1.5) -- ++ (0.9, 0) -- ++ (0, 0.4) -- ++ (-0.9, 0) -- cycle;
    \draw[fill=yellow] (9-0.45, 1.5) -- ++ (0.9, 0) -- ++ (0, 0.4) -- ++ (-0.9, 0) -- cycle;
    \draw[fill=yellow] (10-0.45, 1.5) -- ++ (0.9, 0) -- ++ (0, 0.4) -- ++ (-0.9, 0) -- cycle;
    \draw[fill=yellow] (11-0.45, 1.5) -- ++ (0.9, 0) -- ++ (0, 0.4) -- ++ (-0.9, 0) -- cycle;
    \draw[fill=yellow] (12-0.45, 1.5) -- ++ (0.9, 0) -- ++ (0, 0.4) -- ++ (-0.9, 0) -- cycle;
    \draw[fill=yellow] (13-0.45, 1.5) -- ++ (0.9, 0) -- ++ (0, 0.4) -- ++ (-0.9, 0) -- cycle;
    \draw[fill=yellow] (14-0.45, 1.5) -- ++ (0.9, 0) -- ++ (0, 0.4) -- ++ (-0.9, 0) -- cycle;
    \draw[fill=yellow] (-2.5-0.45, 2.0) -- ++ (0.9, 0) -- ++ (0, 0.4) -- ++ (-0.9, 0) -- cycle;
    \draw[fill=yellow] (-1.5-0.45, 2.0) -- ++ (0.9, 0) -- ++ (0, 0.4) -- ++ (-0.9, 0) -- cycle;
    \draw[fill=green!60!white] (-0.5-0.45, 2.0) -- ++ (0.9, 0) -- ++ (0, 0.4) -- ++ (-0.9, 0) -- cycle;
    \draw[fill=green!60!white] (0.5-0.45, 2.0) -- ++ (0.9, 0) -- ++ (0, 0.4) -- ++ (-0.9, 0) -- cycle;
    \draw[fill=green!60!white] (1.5-0.45, 2.0) -- ++ (0.9, 0) -- ++ (0, 0.4) -- ++ (-0.9, 0) -- cycle;
    \draw[fill=green!60!white] (2.5-0.45, 2.0) -- ++ (0.9, 0) -- ++ (0, 0.4) -- ++ (-0.9, 0) -- cycle;
    \draw[fill=green!60!white] (3.5-0.45, 2.0) -- ++ (0.9, 0) -- ++ (0, 0.4) -- ++ (-0.9, 0) -- cycle;
    \draw[fill=yellow] (4.5-0.45, 2.0) -- ++ (0.9, 0) -- ++ (0, 0.4) -- ++ (-0.9, 0) -- cycle;
    \draw[fill=yellow] (5.5-0.45, 2.0) -- ++ (0.9, 0) -- ++ (0, 0.4) -- ++ (-0.9, 0) -- cycle;
    \draw[fill=yellow] (6.5-0.45, 2.0) -- ++ (0.9, 0) -- ++ (0, 0.4) -- ++ (-0.9, 0) -- cycle;
    \draw[fill=yellow] (7.5-0.45, 2.0) -- ++ (0.9, 0) -- ++ (0, 0.4) -- ++ (-0.9, 0) -- cycle;
    \draw[fill=yellow] (8.5-0.45, 2.0) -- ++ (0.9, 0) -- ++ (0, 0.4) -- ++ (-0.9, 0) -- cycle;
    \draw[fill=yellow] (9.5-0.45, 2.0) -- ++ (0.9, 0) -- ++ (0, 0.4) -- ++ (-0.9, 0) -- cycle;
    \draw[fill=yellow] (10.5-0.45, 2.0) -- ++ (0.9, 0) -- ++ (0, 0.4) -- ++ (-0.9, 0) -- cycle;
    \draw[fill=yellow] (11.5-0.45, 2.0) -- ++ (0.9, 0) -- ++ (0, 0.4) -- ++ (-0.9, 0) -- cycle;
    \draw[fill=yellow] (12.5-0.45, 2.0) -- ++ (0.9, 0) -- ++ (0, 0.4) -- ++ (-0.9, 0) -- cycle;
    \draw[fill=yellow] (13.5-0.45, 2.0) -- ++ (0.9, 0) -- ++ (0, 0.4) -- ++ (-0.9, 0) -- cycle;
    \draw[fill=yellow] (-3-0.45, 2.0) -- ++ (0.4, 0) -- ++ (0, 0.4) -- ++ (-0.4, 0);
    \draw[fill=yellow] (14+0.45, 2.0) -- ++ (-0.4, 0) -- ++ (0, 0.4) -- ++ (0.4, 0);
    \draw[fill=yellow] (-3-0.45, 2.5) -- ++ (0.9, 0) -- ++ (0, 0.4) -- ++ (-0.9, 0) -- cycle;
    \draw[fill=yellow] (-2-0.45, 2.5) -- ++ (0.9, 0) -- ++ (0, 0.4) -- ++ (-0.9, 0) -- cycle;
    \draw[fill=green!60!white] (-1-0.45, 2.5) -- ++ (0.9, 0) -- ++ (0, 0.4) -- ++ (-0.9, 0) -- cycle;
    \draw[fill=green!60!white] (0-0.45, 2.5) -- ++ (0.9, 0) -- ++ (0, 0.4) -- ++ (-0.9, 0) -- cycle;
    \draw[fill=green!60!white] (1-0.45, 2.5) -- ++ (0.9, 0) -- ++ (0, 0.4) -- ++ (-0.9, 0) -- cycle;
    \draw[fill=green!60!white] (2-0.45, 2.5) -- ++ (0.9, 0) -- ++ (0, 0.4) -- ++ (-0.9, 0) -- cycle;
    \draw[fill=green!60!white] (3-0.45, 2.5) -- ++ (0.9, 0) -- ++ (0, 0.4) -- ++ (-0.9, 0) -- cycle;
    \draw[fill=green!60!white] (4-0.45, 2.5) -- ++ (0.9, 0) -- ++ (0, 0.4) -- ++ (-0.9, 0) -- cycle;
    \draw[fill=yellow] (5-0.45, 2.5) -- ++ (0.9, 0) -- ++ (0, 0.4) -- ++ (-0.9, 0) -- cycle;
    \draw[fill=yellow] (6-0.45, 2.5) -- ++ (0.9, 0) -- ++ (0, 0.4) -- ++ (-0.9, 0) -- cycle;
    \draw[fill=yellow] (7-0.45, 2.5) -- ++ (0.9, 0) -- ++ (0, 0.4) -- ++ (-0.9, 0) -- cycle;
    \draw[fill=yellow] (8-0.45, 2.5) -- ++ (0.9, 0) -- ++ (0, 0.4) -- ++ (-0.9, 0) -- cycle;
    \draw[fill=yellow] (9-0.45, 2.5) -- ++ (0.9, 0) -- ++ (0, 0.4) -- ++ (-0.9, 0) -- cycle;
    \draw[fill=yellow] (10-0.45, 2.5) -- ++ (0.9, 0) -- ++ (0, 0.4) -- ++ (-0.9, 0) -- cycle;
    \draw[fill=yellow] (11-0.45, 2.5) -- ++ (0.9, 0) -- ++ (0, 0.4) -- ++ (-0.9, 0) -- cycle;
    \draw[fill=yellow] (12-0.45, 2.5) -- ++ (0.9, 0) -- ++ (0, 0.4) -- ++ (-0.9, 0) -- cycle;
    \draw[fill=yellow] (13-0.45, 2.5) -- ++ (0.9, 0) -- ++ (0, 0.4) -- ++ (-0.9, 0) -- cycle;
    \draw[fill=yellow] (14-0.45, 2.5) -- ++ (0.9, 0) -- ++ (0, 0.4) -- ++ (-0.9, 0) -- cycle;
    \draw [decorate, decoration = {calligraphic brace}] (0.5, 3.275) --  (2.5, 3.275);
    \node[] at (1.5, 4.25) {$B$};
    \draw [decorate, decoration = {calligraphic brace}] (2.5, 3.275) --  (8.0, 3.275);
    \node[] at (5.25, 4.25) {$C$};
    \draw [decorate, decoration = {calligraphic brace}] (8.0, 3.275) --  (12.0, 3.275);
    \node[] at (10.0, 4.25) {$D$};
    \draw [decorate, decoration = {calligraphic brace,mirror}] (-1.5, 0.625) --  (4.5, 0.625);
    \node[] at (1.5, -0.35) {$B'E'$};
    \node[] at (-1.6, 4.25) {$A$};
    \node[] at (13.25, 4.25) {$A$};
    \node[] at (-4.2, 2.0) {$\cdots$};
    \node[] at (15.2, 2.0) {$\cdots$};
    \draw[->, line width=1pt] (-5.2, 0.9) -- node[left] {\small time} (-5.2, 3.0);
    \node[] at (5.5, -1.2) {(b)};

    \end{scope}
    \begin{scope}[xshift=6.8cm, yshift=-3.5cm, scale=0.45, every path/.style={thick}]
    \foreach \x in {-7,...,28}
    {
    \draw[] (\x * 0.5+0.25, 0.7) -- ++ (0,2.5);
    }
    \draw[fill=yellow] (-2.5-0.45, 1.0) -- ++ (0.9, 0) -- ++ (0, 0.4) -- ++ (-0.9, 0) -- cycle;
    \draw[fill=yellow] (-1.5-0.45, 1.0) -- ++ (0.9, 0) -- ++ (0, 0.4) -- ++ (-0.9, 0) -- cycle;
    \draw[fill=yellow] (-0.5-0.45, 1.0) -- ++ (0.9, 0) -- ++ (0, 0.4) -- ++ (-0.9, 0) -- cycle;
    \draw[fill=white] (0.5-0.45, 1.0) -- ++ (0.9, 0) -- ++ (0, 0.4) -- ++ (-0.9, 0) -- cycle;
    \draw[fill=white] (1.5-0.45, 1.0) -- ++ (0.9, 0) -- ++ (0, 0.4) -- ++ (-0.9, 0) -- cycle;
    \draw[fill=white] (2.5-0.45, 1.0) -- ++ (0.9, 0) -- ++ (0, 0.4) -- ++ (-0.9, 0) -- cycle;
    \draw[fill=yellow] (3.5-0.45, 1.0) -- ++ (0.9, 0) -- ++ (0, 0.4) -- ++ (-0.9, 0) -- cycle;
    \draw[fill=yellow] (4.5-0.45, 1.0) -- ++ (0.9, 0) -- ++ (0, 0.4) -- ++ (-0.9, 0) -- cycle;
    \draw[fill=yellow] (5.5-0.45, 1.0) -- ++ (0.9, 0) -- ++ (0, 0.4) -- ++ (-0.9, 0) -- cycle;
    \draw[fill=yellow] (6.5-0.45, 1.0) -- ++ (0.9, 0) -- ++ (0, 0.4) -- ++ (-0.9, 0) -- cycle;
    \draw[fill=yellow] (7.5-0.45, 1.0) -- ++ (0.9, 0) -- ++ (0, 0.4) -- ++ (-0.9, 0) -- cycle;
    \draw[fill=yellow] (8.5-0.45, 1.0) -- ++ (0.9, 0) -- ++ (0, 0.4) -- ++ (-0.9, 0) -- cycle;
    \draw[fill=yellow] (9.5-0.45, 1.0) -- ++ (0.9, 0) -- ++ (0, 0.4) -- ++ (-0.9, 0) -- cycle;
    \draw[fill=yellow] (10.5-0.45, 1.0) -- ++ (0.9, 0) -- ++ (0, 0.4) -- ++ (-0.9, 0) -- cycle;
    \draw[fill=yellow] (11.5-0.45, 1.0) -- ++ (0.9, 0) -- ++ (0, 0.4) -- ++ (-0.9, 0) -- cycle;
    \draw[fill=yellow] (12.5-0.45, 1.0) -- ++ (0.9, 0) -- ++ (0, 0.4) -- ++ (-0.9, 0) -- cycle;
    \draw[fill=yellow] (13.5-0.45, 1.0) -- ++ (0.9, 0) -- ++ (0, 0.4) -- ++ (-0.9, 0) -- cycle;
    \draw[fill=yellow] (-3-0.45, 1.0) -- ++ (0.4, 0) -- ++ (0, 0.4) -- ++ (-0.4, 0);
    \draw[fill=yellow] (14+0.45, 1.0) -- ++ (-0.4, 0) -- ++ (0, 0.4) -- ++ (0.4, 0);
    \draw[fill=yellow] (-3-0.45, 1.5) -- ++ (0.9, 0) -- ++ (0, 0.4) -- ++ (-0.9, 0) -- cycle;
    \draw[fill=yellow] (-2-0.45, 1.5) -- ++ (0.9, 0) -- ++ (0, 0.4) -- ++ (-0.9, 0) -- cycle;
    \draw[fill=yellow] (-1-0.45, 1.5) -- ++ (0.9, 0) -- ++ (0, 0.4) -- ++ (-0.9, 0) -- cycle;
    \draw[fill=white] (0-0.45, 1.5) -- ++ (0.9, 0) -- ++ (0, 0.4) -- ++ (-0.9, 0) -- cycle;
    \draw[fill=white] (1-0.45, 1.5) -- ++ (0.9, 0) -- ++ (0, 0.4) -- ++ (-0.9, 0) -- cycle;
    \draw[fill=white] (2-0.45, 1.5) -- ++ (0.9, 0) -- ++ (0, 0.4) -- ++ (-0.9, 0) -- cycle;
    \draw[fill=white] (3-0.45, 1.5) -- ++ (0.9, 0) -- ++ (0, 0.4) -- ++ (-0.9, 0) -- cycle;
    \draw[fill=yellow] (4-0.45, 1.5) -- ++ (0.9, 0) -- ++ (0, 0.4) -- ++ (-0.9, 0) -- cycle;
    \draw[fill=yellow] (5-0.45, 1.5) -- ++ (0.9, 0) -- ++ (0, 0.4) -- ++ (-0.9, 0) -- cycle;
    \draw[fill=yellow] (6-0.45, 1.5) -- ++ (0.9, 0) -- ++ (0, 0.4) -- ++ (-0.9, 0) -- cycle;
    \draw[fill=yellow] (7-0.45, 1.5) -- ++ (0.9, 0) -- ++ (0, 0.4) -- ++ (-0.9, 0) -- cycle;
    \draw[fill=yellow] (8-0.45, 1.5) -- ++ (0.9, 0) -- ++ (0, 0.4) -- ++ (-0.9, 0) -- cycle;
    \draw[fill=yellow] (9-0.45, 1.5) -- ++ (0.9, 0) -- ++ (0, 0.4) -- ++ (-0.9, 0) -- cycle;
    \draw[fill=yellow] (10-0.45, 1.5) -- ++ (0.9, 0) -- ++ (0, 0.4) -- ++ (-0.9, 0) -- cycle;
    \draw[fill=yellow] (11-0.45, 1.5) -- ++ (0.9, 0) -- ++ (0, 0.4) -- ++ (-0.9, 0) -- cycle;
    \draw[fill=yellow] (12-0.45, 1.5) -- ++ (0.9, 0) -- ++ (0, 0.4) -- ++ (-0.9, 0) -- cycle;
    \draw[fill=yellow] (13-0.45, 1.5) -- ++ (0.9, 0) -- ++ (0, 0.4) -- ++ (-0.9, 0) -- cycle;
    \draw[fill=yellow] (14-0.45, 1.5) -- ++ (0.9, 0) -- ++ (0, 0.4) -- ++ (-0.9, 0) -- cycle;
    \draw[fill=yellow] (-2.5-0.45, 2.0) -- ++ (0.9, 0) -- ++ (0, 0.4) -- ++ (-0.9, 0) -- cycle;
    \draw[fill=yellow] (-1.5-0.45, 2.0) -- ++ (0.9, 0) -- ++ (0, 0.4) -- ++ (-0.9, 0) -- cycle;
    \draw[fill=white] (-0.5-0.45, 2.0) -- ++ (0.9, 0) -- ++ (0, 0.4) -- ++ (-0.9, 0) -- cycle;
    \draw[fill=white] (0.5-0.45, 2.0) -- ++ (0.9, 0) -- ++ (0, 0.4) -- ++ (-0.9, 0) -- cycle;
    \draw[fill=white] (1.5-0.45, 2.0) -- ++ (0.9, 0) -- ++ (0, 0.4) -- ++ (-0.9, 0) -- cycle;
    \draw[fill=white] (2.5-0.45, 2.0) -- ++ (0.9, 0) -- ++ (0, 0.4) -- ++ (-0.9, 0) -- cycle;
    \draw[fill=white] (3.5-0.45, 2.0) -- ++ (0.9, 0) -- ++ (0, 0.4) -- ++ (-0.9, 0) -- cycle;
    \draw[fill=yellow] (4.5-0.45, 2.0) -- ++ (0.9, 0) -- ++ (0, 0.4) -- ++ (-0.9, 0) -- cycle;
    \draw[fill=yellow] (5.5-0.45, 2.0) -- ++ (0.9, 0) -- ++ (0, 0.4) -- ++ (-0.9, 0) -- cycle;
    \draw[fill=yellow] (6.5-0.45, 2.0) -- ++ (0.9, 0) -- ++ (0, 0.4) -- ++ (-0.9, 0) -- cycle;
    \draw[fill=yellow] (7.5-0.45, 2.0) -- ++ (0.9, 0) -- ++ (0, 0.4) -- ++ (-0.9, 0) -- cycle;
    \draw[fill=yellow] (8.5-0.45, 2.0) -- ++ (0.9, 0) -- ++ (0, 0.4) -- ++ (-0.9, 0) -- cycle;
    \draw[fill=yellow] (9.5-0.45, 2.0) -- ++ (0.9, 0) -- ++ (0, 0.4) -- ++ (-0.9, 0) -- cycle;
    \draw[fill=yellow] (10.5-0.45, 2.0) -- ++ (0.9, 0) -- ++ (0, 0.4) -- ++ (-0.9, 0) -- cycle;
    \draw[fill=yellow] (11.5-0.45, 2.0) -- ++ (0.9, 0) -- ++ (0, 0.4) -- ++ (-0.9, 0) -- cycle;
    \draw[fill=yellow] (12.5-0.45, 2.0) -- ++ (0.9, 0) -- ++ (0, 0.4) -- ++ (-0.9, 0) -- cycle;
    \draw[fill=yellow] (13.5-0.45, 2.0) -- ++ (0.9, 0) -- ++ (0, 0.4) -- ++ (-0.9, 0) -- cycle;
    \draw[fill=yellow] (-3-0.45, 2.0) -- ++ (0.4, 0) -- ++ (0, 0.4) -- ++ (-0.4, 0);
    \draw[fill=yellow] (14+0.45, 2.0) -- ++ (-0.4, 0) -- ++ (0, 0.4) -- ++ (0.4, 0);
    \draw[fill=yellow] (-3-0.45, 2.5) -- ++ (0.9, 0) -- ++ (0, 0.4) -- ++ (-0.9, 0) -- cycle;
    \draw[fill=yellow] (-2-0.45, 2.5) -- ++ (0.9, 0) -- ++ (0, 0.4) -- ++ (-0.9, 0) -- cycle;
    \draw[fill=white] (-1-0.45, 2.5) -- ++ (0.9, 0) -- ++ (0, 0.4) -- ++ (-0.9, 0) -- cycle;
    \draw[fill=white] (0-0.45, 2.5) -- ++ (0.9, 0) -- ++ (0, 0.4) -- ++ (-0.9, 0) -- cycle;
    \draw[fill=white] (1-0.45, 2.5) -- ++ (0.9, 0) -- ++ (0, 0.4) -- ++ (-0.9, 0) -- cycle;
    \draw[fill=white] (2-0.45, 2.5) -- ++ (0.9, 0) -- ++ (0, 0.4) -- ++ (-0.9, 0) -- cycle;
    \draw[fill=white] (3-0.45, 2.5) -- ++ (0.9, 0) -- ++ (0, 0.4) -- ++ (-0.9, 0) -- cycle;
    \draw[fill=white] (4-0.45, 2.5) -- ++ (0.9, 0) -- ++ (0, 0.4) -- ++ (-0.9, 0) -- cycle;
    \draw[fill=yellow] (5-0.45, 2.5) -- ++ (0.9, 0) -- ++ (0, 0.4) -- ++ (-0.9, 0) -- cycle;
    \draw[fill=yellow] (6-0.45, 2.5) -- ++ (0.9, 0) -- ++ (0, 0.4) -- ++ (-0.9, 0) -- cycle;
    \draw[fill=yellow] (7-0.45, 2.5) -- ++ (0.9, 0) -- ++ (0, 0.4) -- ++ (-0.9, 0) -- cycle;
    \draw[fill=yellow] (8-0.45, 2.5) -- ++ (0.9, 0) -- ++ (0, 0.4) -- ++ (-0.9, 0) -- cycle;
    \draw[fill=yellow] (9-0.45, 2.5) -- ++ (0.9, 0) -- ++ (0, 0.4) -- ++ (-0.9, 0) -- cycle;
    \draw[fill=yellow] (10-0.45, 2.5) -- ++ (0.9, 0) -- ++ (0, 0.4) -- ++ (-0.9, 0) -- cycle;
    \draw[fill=yellow] (11-0.45, 2.5) -- ++ (0.9, 0) -- ++ (0, 0.4) -- ++ (-0.9, 0) -- cycle;
    \draw[fill=yellow] (12-0.45, 2.5) -- ++ (0.9, 0) -- ++ (0, 0.4) -- ++ (-0.9, 0) -- cycle;
    \draw[fill=yellow] (13-0.45, 2.5) -- ++ (0.9, 0) -- ++ (0, 0.4) -- ++ (-0.9, 0) -- cycle;
    \draw[fill=yellow] (14-0.45, 2.5) -- ++ (0.9, 0) -- ++ (0, 0.4) -- ++ (-0.9, 0) -- cycle;
    \draw [decorate, decoration = {calligraphic brace}] (0.5, 3.275) --  (2.5, 3.275);
    \node[] at (1.5, 4.25) {$B$};
    \draw [decorate, decoration = {calligraphic brace}] (2.5, 3.275) --  (8.0, 3.275);
    \node[] at (5.25, 4.25) {$C$};
    \draw [decorate, decoration = {calligraphic brace}] (8.0, 3.275) --  (12.0, 3.275);
    \node[] at (10.0, 4.25) {$D$};
    \draw [decorate, decoration = {calligraphic brace,mirror}] (-1.5, 0.625) --  (4.5, 0.625);
    \node[] at (1.5, -0.35) {$B'E'$};
    \node[] at (-1.6, 4.25) {$A$};
    \node[] at (13.25, 4.25) {$A$};
    \node[] at (-4.2, 2.0) {$\cdots$};
    \node[] at (15.2, 2.0) {$\cdots$};
    \draw[->, line width=1pt] (-5.2, 0.9) -- node[left] {\small time} (-5.2, 3.0);
    \node[] at (5.5, -1.2) {(c)};

    \end{scope}
\end{tikzpicture}
    \caption{Schematic illustration of the cancellation of gates in $U$ by $F_{BE}^{\dagger}$, in a particular cross-section. (a) As an example, we take a cross-section view of the circuit diagram according to the horizontal dotted line through $A$, $B$, $C$, and $D$. (b) The unitary $F_{BE}$, whose gates are shaded in green. (c) The gates after $F_{BE}^{\dagger}$ is applied, where the cancelled gates have no colored shading. $V_{ACD}$ is the composition of all the remaining (colored) gates.}
    \label{fig:circuit_stability_1_1_plc}
\end{figure}

\begin{figure}[h]
    \centering
    \begin{tikzpicture}[scale=0.84]
    \begin{scope}[xshift=0cm]
    \draw[fill=yellow] (-2, -1.5) -- ++ (4, 0) -- ++ (0, 3) -- ++ (-4, 0) -- cycle;
    \draw[line width=1pt] (75:0.5cm) -- (75:1.0cm) arc(75:-75:1.0cm) -- (-75:0.5cm) arc(285:75:0.5cm);
    \draw[line width=1pt, color=red!60!white, fill=red!60!white] (60:0.4cm) -- (0.381173cm, 1.02256cm) arc(69.56:290.44:1.1cm) -- (300:0.4cm) arc(300:60:0.4cm);
    \draw[line width=1.5pt, color=red!20!white, fill=red!20!white] (-75:0.5cm) -- (-75:1.0cm) arc(285:290.44:1.0cm) -- (300:0.4cm) arc(300:60:0.4cm) -- (0.355291cm, 0.925967cm) arc(70:75:1.0cm) -- (75:0.5cm) arc(75:285:0.5cm);
    \draw[draw=none, fill=white] (75:0.5cm) -- (75:1cm)  arc(75:285:1cm) -- (285:0.5cm) arc(285:75:0.5cm);
    \draw[line width=1pt] (0.12940952255cm, 0.48296291314cm) arc(75:-75:0.5cm);
    \draw[line width=1pt] (0.12940952255cm, 0.48296291314cm) arc(75:285:0.5cm);
    \draw[line width=1pt] (0,0) circle(1cm);
    \draw[line width=1pt] (75:0.5cm) -- (75:1cm);
    \draw[line width=1pt] (105:0.5cm) -- (105:1cm);
    \draw[line width=1pt] (-75:0.5cm) -- (-75:1cm);
    \draw[line width=1pt] (-105:0.5cm) -- (-105:1cm);
    
    \node[] () at (90:0.75cm) {$E$};
    \node[] () at (-90:0.75cm) {$E$};
    \node[] () at (0,0) {$C$};
    \node[] () at (0:0.75cm) {$D$};
    \node[] () at (180:0.75cm) {$B$};
    \node[] () at (-1.5, 0) {$A$};
    \node[below] () at (0, -1.5) {(a)};
    \end{scope}

    \begin{scope}[xshift=4.5cm]
    \draw[fill=yellow] (-2, -1.5) -- ++ (4, 0) -- ++ (0, 3) -- ++ (-4, 0) -- cycle;
    \draw[line width=1pt, fill=yellow!50!white] (75:0.5cm) -- (75:1.0cm) arc(75:-75:1.0cm) -- (-75:0.5cm) arc(285:75:0.5cm);
    \draw[draw=none, fill=red!60!white] (0.355291cm, 0.925967cm) -- (0.381173cm, 1.02256cm) arc(69.56:290.44:1.1cm) -- (0.355291cm, -0.925967cm) arc(290.44:69.56:1.0cm);
    \draw[line width=1pt, color=red!60!white] (0.355291cm, 0.925967cm) -- (0.381173cm, 1.02256cm) arc(69.56:290.44:1.1cm) -- (0.355291cm, -0.925967cm);
    \draw[draw=none, fill=white] (75:0.5cm) -- (75:1cm) arc(75:285:1cm) -- (285:0.5cm) arc(285:75:0.5cm);
    \draw[line width=1.5pt, color=red!20!white, fill=red!20!white] (-75:0.5cm) -- (-75:1.0cm) arc(285:290.44:1.0cm) -- (300:0.4cm) arc(300:60:0.4cm) -- (0.355291cm, 0.925967cm) arc(70:75:1.0cm) -- (75:0.5cm) arc(75:285:0.5cm);
    \draw[line width=1pt] (0.12940952255cm, 0.48296291314cm) arc(75:-75:0.5cm);
    \draw[line width=1pt] (0.12940952255cm, 0.48296291314cm) arc(75:285:0.5cm);
    \draw[line width=1pt] (0,0) circle(1cm);
    \draw[line width=1pt] (75:0.5cm) -- (75:1cm);
    \draw[line width=1pt] (105:0.5cm) -- (105:1cm);
    \draw[line width=1pt] (-75:0.5cm) -- (-75:1cm);
    \draw[line width=1pt] (-105:0.5cm) -- (-105:1cm);
    
    \node[] () at (90:0.75cm) {$E$};
    \node[] () at (-90:0.75cm) {$E$};
    \node[text=gray] () at (0,0) {$C$};
    \node[text=gray] () at (0:0.75cm) {$D$};
    \node[] () at (180:0.75cm) {$B$};
    \node[] () at (-1.5, 0) {$A$};
    \node[below] () at (0, -1.5) {(b)};
    \end{scope}

    \begin{scope}[xshift=9cm]
    \draw[fill=yellow] (-2, -1.5) -- ++ (4, 0) -- ++ (0, 3) -- ++ (-4, 0) -- cycle;
    \draw[line width=1pt, fill=gray!10!white] (75:0.5cm) -- (75:1.0cm) arc(75:-75:1.0cm) -- (-75:0.5cm) arc(285:75:0.5cm);
    \draw[draw=none, fill=red!60!white] (0.355291cm, 0.925967cm) -- (0.381173cm, 1.02256cm) arc(69.56:290.44:1.1cm) -- (0.355291cm, -0.925967cm) arc(290.44:69.56:1.0cm);
    \draw[line width=1pt, color=red!60!white] (0.355291cm, 0.925967cm) -- (0.381173cm, 1.02256cm) arc(69.56:290.44:1.1cm) -- (0.355291cm, -0.925967cm);
    \draw[draw=none, fill=white] (75:0.5cm) -- (75:1cm) arc(75:285:1cm) -- (285:0.5cm) arc(285:75:0.5cm);
    \draw[line width=1pt, color=gray!10!white, fill=gray!10!white] (-75:0.5cm) -- (-75:1.0cm) arc(285:290.44:1.0cm) -- (300:0.4cm) arc(300:60:0.4cm) -- (0.355291cm, 0.925967cm) arc(70:75:1.0cm) -- (75:0.5cm) arc(75:285:0.5cm);
    \draw[line width=1pt, color=gray] (0.12940952255cm, 0.48296291314cm) arc(75:-75:0.5cm);
    \draw[line width=1pt] (0.12940952255cm, 0.48296291314cm) arc(75:285:0.5cm);
    \draw[draw=none, fill=brown!60!white] (75:0.9cm) arc(75:-75:0.9cm) -- (-75:1.0cm) arc(-75:75:1.0cm) -- (75:0.9cm);
    \draw[line width=1pt] (0,0) circle(1cm);
    \draw[line width=1pt] (75:0.5cm) -- (75:1cm);
    \draw[line width=1pt] (105:0.5cm) -- (105:1cm);
    \draw[line width=1pt] (-75:0.5cm) -- (-75:1cm);
    \draw[line width=1pt] (-105:0.5cm) -- (-105:1cm);
    \draw[line width=1pt, color=gray, style=dotted] (75:0.9cm) arc(75:-75:0.9cm);
    
    \node[] () at (90:0.75cm) {$E$};
    \node[] () at (-90:0.75cm) {$E$};
    \node[text=gray] () at (0,0) {$C$};
    \node[text=gray] () at (0.1cm:0.7cm) {$\widetilde{D}$};
    \node[] () at (180:0.75cm) {$B$};
    \node[] () at (-1.5, 0) {$A$};
    \node[below] () at (0, -1.5) {(c)};
    \end{scope}
    
    \begin{scope}[xshift=13.5cm]
    \draw[fill=yellow] (-2, -1.5) -- ++ (4, 0) -- ++ (0, 3) -- ++ (-4, 0) -- cycle;
    \draw[line width=1pt, fill=gray!10!white] (105:0.5cm) -- (105:1.0cm) arc(105:-105:1.0cm) -- (-105:0.5cm) arc(255:105:0.5cm);
    \draw[draw=none, fill=red!60!white] (0.355291cm, 0.925967cm) -- (0.381173cm, 1.02256cm) arc(69.56:290.44:1.1cm) -- (0.355291cm, -0.925967cm) arc(290.44:69.56:1.0cm);
    \draw[line width=1pt, color=red!60!white] (0.355291cm, 0.925967cm) -- (0.381173cm, 1.02256cm) arc(69.56:290.44:1.1cm) -- (0.355291cm, -0.925967cm);
    \draw[draw=none, fill=white] (105:0.5cm) -- (105:1cm) arc(105:255:1cm) -- (255:0.5cm) arc(-105:105:0.5cm);
    \draw[line width=1pt, color=gray] (75:0.5cm) -- (75:1cm);
    \draw[line width=1pt] (105:0.5cm) -- (105:1cm);
    \draw[line width=1pt, color=gray] (-75:0.5cm) -- (-75:1cm);
    \draw[line width=1pt] (-105:0.5cm) -- (-105:1cm);
    \draw[draw=none, fill=brown!60!white] (75:0.9cm) arc(75:-75:0.9cm) -- (-75:1.0cm) arc(-75:75:1.0cm) -- (75:0.9cm);
    \draw[line width=1pt] (0,0) circle(0.5cm);
    \draw[line width=1pt] (0,0) circle(1cm);
    \draw[line width=1pt, color=gray, style=dotted] (75:0.9cm) arc(75:-75:0.9cm);
    
    \node[text=gray] () at (90:0.75cm) {$E$};
    \node[text=gray] () at (-90:0.75cm) {$E$};
    \node[] () at (0,0) {$C$};
    \node[text=gray] () at (0.1cm:0.7cm) {$\widetilde{D}$};
    \node[] () at (180:0.75cm) {$B$};
    \node[] () at (-1.5, 0) {$A$};
    \node[below] () at (0, -1.5) {(d)};
    \end{scope}
    \end{tikzpicture}
    \caption{The second step $\Gamma_2$ illustrated. (a) The partial trace $\Tr_{CD \setminus C'D'}$ applied; partial traces are depicted with graying the traced-out subsystems. (b) The partial trace $\Tr_{C'D'}$ also applied. Note that $\Tr_{CD \setminus C'D'} \circ \Tr_{C'D'} = \Tr_{CD}$. (c) By the cyclic property of the partial trace, the reduced density matrix on $BC\widetilde{D}E$ is indistinguishable from $\rho_{BC\widetilde{D}E}$. The figure depicts this state with a grayed $C\widetilde{D}$ without any coloring. The brown shading corresponds to the remaining gates on $D \setminus \widetilde{D}$. (d) The pre-circuit $\rho_{BC\widetilde{D}E}$ is extendible from $BE$ to $BC$, so we apply its extending channel $\Phi$. Doing so, we obtain the reduced density matrix $\Tr_{DE} \left( V_{AD \setminus \widetilde{D}} \rho V_{AD \setminus \widetilde{D}}^{\dagger} \right)$.}
    \label{fig:circuit_stability_1_2}
\end{figure}
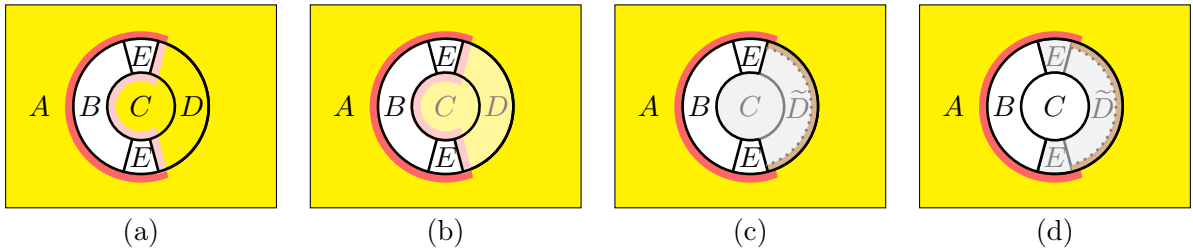

Now, we apply our second channel $\mathcal{I}_{A'} \otimes \Gamma_{2}$. (See Fig.~\ref{fig:circuit_stability_1_2} for its graphical illustration.) This channel is a composition of two operations, described below. First, we apply $\Tr_{CD\setminus C'D'}$. This partial trace, together with the partial trace in Eq.~\eqref{eq:end_of_step1}, can be combined together as $\Tr_{CD\setminus C'D'} \circ \Tr_{C'D'} = \Tr_{CD}$. Our state at this point is thus $\Tr_{CD}\left(V_{ACD} \rho V_{ACD}^{\dagger}\right)$. By Lemma~\ref{lemma:light_cone}(i), $V_{ACD} = V_{CD} V_{AD\setminus\widetilde{D}}$, where $V_{AD\setminus\widetilde{D}}$ is the past light-cone of $A$ under $V_{ACD}$ and $V_{CD}$ is supported on $CD$. Since $\operatorname{dist}(A,C)>d$, $V_{AD\setminus\widetilde{D}}$ is indeed supported on $AD\setminus\widetilde{D}$. Consequently, the state we have at this point is:
\begin{equation}\label{eq:end_of_step2_part1}
    \Tr_{CD}\left(V_{ACD} \rho V_{ACD}^{\dagger}\right) = \Tr_{CD}\left(V_{AD\setminus \widetilde{D}} \rho V_{AD\setminus \widetilde{D}}^{\dagger}\right).
\end{equation}

The remaining channel in the construction of $\mathcal{I}_{A'}\otimes\Gamma_2$ is an extending map of $\rho_{BC\widetilde{D}E}$. Notice that, because $V_{AD\setminus \widetilde{D}}$ is supported on $\Lambda \setminus BC\widetilde{D}E$, the reduced density matrix of $V_{AD\setminus \widetilde{D}} \rho V_{AD\setminus \widetilde{D}}^{\dagger}$ on $BC\widetilde{D}E$ is indistinguishable from $\rho_{BC\widetilde{D}E}$. This $\rho_{BC\widetilde{D}E}$, by our assumption, is extendible from $BE$ to $BC$. Let us then denote its extending channel as $\Phi:\mathcal{D}_{BE} \rightarrow \mathcal{D}_{BC}$, and apply it as our second operation of $\Gamma_2$. Upon applying $\Phi$ to Eq.~\eqref{eq:end_of_step2_part1}, we obtain:
\begin{equation}\label{eq:end_of_step2_part2}
\begin{aligned} 
        \mathcal{I}_{AD\setminus\widetilde{D}} \otimes \Phi \left[\Tr_{CD} \left( V_{AD\setminus \widetilde{D}} \rho V_{AD\setminus \widetilde{D}}^{\dagger} \right) \right] 
        &= \Tr_{D \setminus \widetilde{D}} \left( V_{AD\setminus\widetilde{D}} \left( \mathcal{I}_{AD\setminus\widetilde{D}} \otimes \Phi \Bigl[ \rho_{AB(D \setminus \widetilde{D})E} \Bigr] \right)V_{AD\setminus\widetilde{D}}^{\dagger} \right) \\
        &= \Tr_{D \setminus \widetilde{D}} \left( V_{AD\setminus\widetilde{D}} \rho_{ABCD \setminus \widetilde{D}} V_{AD\setminus\widetilde{D}}^{\dagger} \right) \\
        &= \Tr_{DE} \left( V_{AD\setminus\widetilde{D}} \rho V_{AD\setminus\widetilde{D}}^{\dagger} \right).
\end{aligned}
\end{equation}

To set up for the final step of the proof, let us update our schematic circuit diagram across the cross-section in Fig.~\ref{fig:circuit_stability_1_1_plc}. In our second step, we effectively cancelled out $V_{CD}$. Fig.~\ref{fig:circuit_stability_1_2_plc} shows the set of remaining gates in $V_{AD\setminus\widetilde{D}} \rho V_{AD\setminus\widetilde{D}}^{\dagger}$ over the same cross-section. Note that only the gates supported in $CD$ were affected by our second step.

\begin{figure}[h]
\centering
\begin{tikzpicture}[every path/.style={thick},scale=0.45]
    \foreach \x in {-7,...,28}
    {
    \draw[] (\x * 0.5+0.25, 0.7) -- ++ (0,2.5);
    }
    \draw[fill=yellow] (-2.5-0.45, 1.0) -- ++ (0.9, 0) -- ++ (0, 0.4) -- ++ (-0.9, 0) -- cycle;
    \draw[fill=yellow] (-1.5-0.45, 1.0) -- ++ (0.9, 0) -- ++ (0, 0.4) -- ++ (-0.9, 0) -- cycle;
    \draw[fill=yellow] (-0.5-0.45, 1.0) -- ++ (0.9, 0) -- ++ (0, 0.4) -- ++ (-0.9, 0) -- cycle;
    \draw[fill=white] (0.5-0.45, 1.0) -- ++ (0.9, 0) -- ++ (0, 0.4) -- ++ (-0.9, 0) -- cycle;
    \draw[fill=white] (1.5-0.45, 1.0) -- ++ (0.9, 0) -- ++ (0, 0.4) -- ++ (-0.9, 0) -- cycle;
    \draw[fill=white] (2.5-0.45, 1.0) -- ++ (0.9, 0) -- ++ (0, 0.4) -- ++ (-0.9, 0) -- cycle;
    \draw[fill=white] (3.5-0.45, 1.0) -- ++ (0.9, 0) -- ++ (0, 0.4) -- ++ (-0.9, 0) -- cycle;
    \draw[fill=white] (4.5-0.45, 1.0) -- ++ (0.9, 0) -- ++ (0, 0.4) -- ++ (-0.9, 0) -- cycle;
    \draw[fill=white] (5.5-0.45, 1.0) -- ++ (0.9, 0) -- ++ (0, 0.4) -- ++ (-0.9, 0) -- cycle;
    \draw[fill=white] (6.5-0.45, 1.0) -- ++ (0.9, 0) -- ++ (0, 0.4) -- ++ (-0.9, 0) -- cycle;
    \draw[fill=white] (7.5-0.45, 1.0) -- ++ (0.9, 0) -- ++ (0, 0.4) -- ++ (-0.9, 0) -- cycle;
    \draw[fill=white] (8.5-0.45, 1.0) -- ++ (0.9, 0) -- ++ (0, 0.4) -- ++ (-0.9, 0) -- cycle;
    \draw[fill=white] (9.5-0.45, 1.0) -- ++ (0.9, 0) -- ++ (0, 0.4) -- ++ (-0.9, 0) -- cycle;
    \draw[fill=yellow] (10.5-0.45, 1.0) -- ++ (0.9, 0) -- ++ (0, 0.4) -- ++ (-0.9, 0) -- cycle;
    \draw[fill=yellow] (11.5-0.45, 1.0) -- ++ (0.9, 0) -- ++ (0, 0.4) -- ++ (-0.9, 0) -- cycle;
    \draw[fill=yellow] (12.5-0.45, 1.0) -- ++ (0.9, 0) -- ++ (0, 0.4) -- ++ (-0.9, 0) -- cycle;
    \draw[fill=yellow] (13.5-0.45, 1.0) -- ++ (0.9, 0) -- ++ (0, 0.4) -- ++ (-0.9, 0) -- cycle;
    \draw[fill=yellow] (-3-0.45, 1.0) -- ++ (0.4, 0) -- ++ (0, 0.4) -- ++ (-0.4, 0);
    \draw[fill=yellow] (14+0.45, 1.0) -- ++ (-0.4, 0) -- ++ (0, 0.4) -- ++ (0.4, 0);
    \draw[fill=yellow] (-3-0.45, 1.5) -- ++ (0.9, 0) -- ++ (0, 0.4) -- ++ (-0.9, 0) -- cycle;
    \draw[fill=yellow] (-2-0.45, 1.5) -- ++ (0.9, 0) -- ++ (0, 0.4) -- ++ (-0.9, 0) -- cycle;
    \draw[fill=yellow] (-1-0.45, 1.5) -- ++ (0.9, 0) -- ++ (0, 0.4) -- ++ (-0.9, 0) -- cycle;
    \draw[fill=white] (0-0.45, 1.5) -- ++ (0.9, 0) -- ++ (0, 0.4) -- ++ (-0.9, 0) -- cycle;
    \draw[fill=white] (1-0.45, 1.5) -- ++ (0.9, 0) -- ++ (0, 0.4) -- ++ (-0.9, 0) -- cycle;
    \draw[fill=white] (2-0.45, 1.5) -- ++ (0.9, 0) -- ++ (0, 0.4) -- ++ (-0.9, 0) -- cycle;
    \draw[fill=white] (3-0.45, 1.5) -- ++ (0.9, 0) -- ++ (0, 0.4) -- ++ (-0.9, 0) -- cycle;
    \draw[fill=white] (4-0.45, 1.5) -- ++ (0.9, 0) -- ++ (0, 0.4) -- ++ (-0.9, 0) -- cycle;
    \draw[fill=white] (5-0.45, 1.5) -- ++ (0.9, 0) -- ++ (0, 0.4) -- ++ (-0.9, 0) -- cycle;
    \draw[fill=white] (6-0.45, 1.5) -- ++ (0.9, 0) -- ++ (0, 0.4) -- ++ (-0.9, 0) -- cycle;
    \draw[fill=white] (7-0.45, 1.5) -- ++ (0.9, 0) -- ++ (0, 0.4) -- ++ (-0.9, 0) -- cycle;
    \draw[fill=white] (8-0.45, 1.5) -- ++ (0.9, 0) -- ++ (0, 0.4) -- ++ (-0.9, 0) -- cycle;
    \draw[fill=white] (9-0.45, 1.5) -- ++ (0.9, 0) -- ++ (0, 0.4) -- ++ (-0.9, 0) -- cycle;
    \draw[fill=white] (10-0.45, 1.5) -- ++ (0.9, 0) -- ++ (0, 0.4) -- ++ (-0.9, 0) -- cycle;
    \draw[fill=yellow] (11-0.45, 1.5) -- ++ (0.9, 0) -- ++ (0, 0.4) -- ++ (-0.9, 0) -- cycle;
    \draw[fill=yellow] (12-0.45, 1.5) -- ++ (0.9, 0) -- ++ (0, 0.4) -- ++ (-0.9, 0) -- cycle;
    \draw[fill=yellow] (13-0.45, 1.5) -- ++ (0.9, 0) -- ++ (0, 0.4) -- ++ (-0.9, 0) -- cycle;
    \draw[fill=yellow] (14-0.45, 1.5) -- ++ (0.9, 0) -- ++ (0, 0.4) -- ++ (-0.9, 0) -- cycle;
    \draw[fill=yellow] (-2.5-0.45, 2.0) -- ++ (0.9, 0) -- ++ (0, 0.4) -- ++ (-0.9, 0) -- cycle;
    \draw[fill=yellow] (-1.5-0.45, 2.0) -- ++ (0.9, 0) -- ++ (0, 0.4) -- ++ (-0.9, 0) -- cycle;
    \draw[fill=white] (-0.5-0.45, 2.0) -- ++ (0.9, 0) -- ++ (0, 0.4) -- ++ (-0.9, 0) -- cycle;
    \draw[fill=white] (0.5-0.45, 2.0) -- ++ (0.9, 0) -- ++ (0, 0.4) -- ++ (-0.9, 0) -- cycle;
    \draw[fill=white] (1.5-0.45, 2.0) -- ++ (0.9, 0) -- ++ (0, 0.4) -- ++ (-0.9, 0) -- cycle;
    \draw[fill=white] (2.5-0.45, 2.0) -- ++ (0.9, 0) -- ++ (0, 0.4) -- ++ (-0.9, 0) -- cycle;
    \draw[fill=white] (3.5-0.45, 2.0) -- ++ (0.9, 0) -- ++ (0, 0.4) -- ++ (-0.9, 0) -- cycle;
    \draw[fill=white] (4.5-0.45, 2.0) -- ++ (0.9, 0) -- ++ (0, 0.4) -- ++ (-0.9, 0) -- cycle;
    \draw[fill=white] (5.5-0.45, 2.0) -- ++ (0.9, 0) -- ++ (0, 0.4) -- ++ (-0.9, 0) -- cycle;
    \draw[fill=white] (6.5-0.45, 2.0) -- ++ (0.9, 0) -- ++ (0, 0.4) -- ++ (-0.9, 0) -- cycle;
    \draw[fill=white] (7.5-0.45, 2.0) -- ++ (0.9, 0) -- ++ (0, 0.4) -- ++ (-0.9, 0) -- cycle;
    \draw[fill=white] (8.5-0.45, 2.0) -- ++ (0.9, 0) -- ++ (0, 0.4) -- ++ (-0.9, 0) -- cycle;
    \draw[fill=white] (9.5-0.45, 2.0) -- ++ (0.9, 0) -- ++ (0, 0.4) -- ++ (-0.9, 0) -- cycle;
    \draw[fill=white] (10.5-0.45, 2.0) -- ++ (0.9, 0) -- ++ (0, 0.4) -- ++ (-0.9, 0) -- cycle;
    \draw[fill=yellow] (11.5-0.45, 2.0) -- ++ (0.9, 0) -- ++ (0, 0.4) -- ++ (-0.9, 0) -- cycle;
    \draw[fill=yellow] (12.5-0.45, 2.0) -- ++ (0.9, 0) -- ++ (0, 0.4) -- ++ (-0.9, 0) -- cycle;
    \draw[fill=yellow] (13.5-0.45, 2.0) -- ++ (0.9, 0) -- ++ (0, 0.4) -- ++ (-0.9, 0) -- cycle;
    \draw[fill=yellow] (-3-0.45, 2.0) -- ++ (0.4, 0) -- ++ (0, 0.4) -- ++ (-0.4, 0);
    \draw[fill=yellow] (14+0.45, 2.0) -- ++ (-0.4, 0) -- ++ (0, 0.4) -- ++ (0.4, 0);
    \draw[fill=yellow] (-3-0.45, 2.5) -- ++ (0.9, 0) -- ++ (0, 0.4) -- ++ (-0.9, 0) -- cycle;
    \draw[fill=yellow] (-2-0.45, 2.5) -- ++ (0.9, 0) -- ++ (0, 0.4) -- ++ (-0.9, 0) -- cycle;
    \draw[fill=white] (-1-0.45, 2.5) -- ++ (0.9, 0) -- ++ (0, 0.4) -- ++ (-0.9, 0) -- cycle;
    \draw[fill=white] (0-0.45, 2.5) -- ++ (0.9, 0) -- ++ (0, 0.4) -- ++ (-0.9, 0) -- cycle;
    \draw[fill=white] (1-0.45, 2.5) -- ++ (0.9, 0) -- ++ (0, 0.4) -- ++ (-0.9, 0) -- cycle;
    \draw[fill=white] (2-0.45, 2.5) -- ++ (0.9, 0) -- ++ (0, 0.4) -- ++ (-0.9, 0) -- cycle;
    \draw[fill=white] (3-0.45, 2.5) -- ++ (0.9, 0) -- ++ (0, 0.4) -- ++ (-0.9, 0) -- cycle;
    \draw[fill=white] (4-0.45, 2.5) -- ++ (0.9, 0) -- ++ (0, 0.4) -- ++ (-0.9, 0) -- cycle;
    \draw[fill=white] (5-0.45, 2.5) -- ++ (0.9, 0) -- ++ (0, 0.4) -- ++ (-0.9, 0) -- cycle;
    \draw[fill=white] (6-0.45, 2.5) -- ++ (0.9, 0) -- ++ (0, 0.4) -- ++ (-0.9, 0) -- cycle;
    \draw[fill=white] (7-0.45, 2.5) -- ++ (0.9, 0) -- ++ (0, 0.4) -- ++ (-0.9, 0) -- cycle;
    \draw[fill=white] (8-0.45, 2.5) -- ++ (0.9, 0) -- ++ (0, 0.4) -- ++ (-0.9, 0) -- cycle;
    \draw[fill=white] (9-0.45, 2.5) -- ++ (0.9, 0) -- ++ (0, 0.4) -- ++ (-0.9, 0) -- cycle;
    \draw[fill=white] (10-0.45, 2.5) -- ++ (0.9, 0) -- ++ (0, 0.4) -- ++ (-0.9, 0) -- cycle;
    \draw[fill=white] (11-0.45, 2.5) -- ++ (0.9, 0) -- ++ (0, 0.4) -- ++ (-0.9, 0) -- cycle;
    \draw[fill=yellow] (12-0.45, 2.5) -- ++ (0.9, 0) -- ++ (0, 0.4) -- ++ (-0.9, 0) -- cycle;
    \draw[fill=yellow] (13-0.45, 2.5) -- ++ (0.9, 0) -- ++ (0, 0.4) -- ++ (-0.9, 0) -- cycle;
    \draw[fill=yellow] (14-0.45, 2.5) -- ++ (0.9, 0) -- ++ (0, 0.4) -- ++ (-0.9, 0) -- cycle;
    \draw [decorate, decoration = {calligraphic brace}] (0.5, 3.275) --  (2.5, 3.275);
    \node[] at (1.5, 4.25) {$B$};
    \draw [decorate, decoration = {calligraphic brace}] (2.5, 3.275) --  (8.0, 3.275);
    \node[] at (5.25, 4.25) {$C$};
    \draw [decorate, decoration = {calligraphic brace}] (8.0, 3.275) --  (12.0, 3.275);
    \node[] at (10.0, 4.25) {$D$};
    \node[] at (-1.6, 4.25) {$A$};
    \node[] at (13.25, 4.25) {$A$};
    \node[] at (-4.2, 2.0) {$\cdots$};
    \node[] at (15.2, 2.0) {$\cdots$};
    \draw[->, line width=1pt] (-5.2, 0.9) -- node[left] {\small time} (-5.2, 3.0);

\end{tikzpicture}
    \caption{Schematic illustration of $V_{AD\setminus\widetilde{D}} \rho V_{AD\setminus\widetilde{D}}^{\dagger}$, for the cross-section in Fig.~\ref{fig:circuit_stability_1_1_plc}. Compared to Fig.~\ref{fig:circuit_stability_1_1_plc}(c), only the gates of $V_{CD}$ were removed.}
    \label{fig:circuit_stability_1_2_plc}
\end{figure}
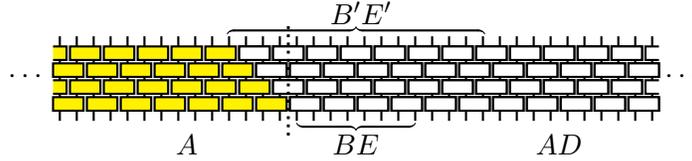

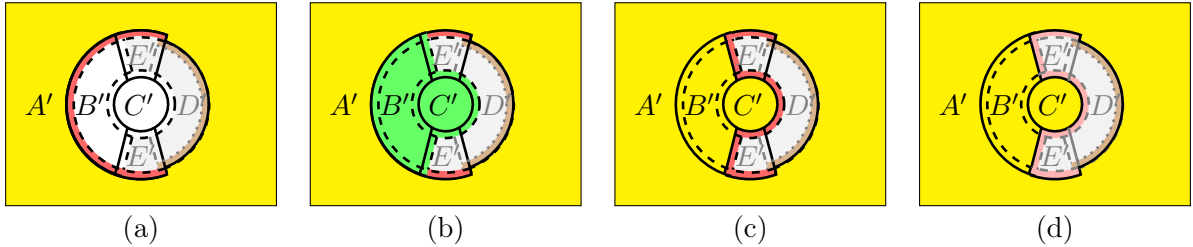
\begin{figure}[h]
    \centering
    \begin{tikzpicture}[scale=0.84]
    \begin{scope}[xshift=0cm]
    \draw[fill=yellow] (-2, -1.5) -- ++ (4, 0) -- ++ (0, 3) -- ++ (-4, 0) -- cycle;
    \draw[line width=1pt, color=gray!10!white, fill=gray!10!white] (105:0.5cm) -- (105:1.0cm) arc(105:-105:1.0cm) -- (-105:0.5cm) arc(255:105:0.5cm);
    \draw[line width=1pt] (0.355291cm, 0.925967cm) -- (0.381173cm, 1.02256cm) arc(69.56:290.44:1.1cm) -- (0.355291cm, -0.925967cm);
    \draw[draw=none, fill=white] (105:0.5cm) -- (105:1.1cm) arc(105:255:1.1cm) -- (255:0.5cm) arc(-105:105:0.5cm);
    \draw[draw=none, fill=red!60!white] (0.355291cm, 0.925967cm) -- (0.381173cm, 1.02256cm) arc(69.56:290.44:1.1cm) -- (0.355291cm, -0.925967cm) arc(290.44:69.56:1.0cm);
    \draw[line width=1pt] (0.381173cm, 1.02256cm) arc(69.56:290.44:1.1cm);
    \draw[line width=1pt] (249.56:1.0cm) -- (240:0.4cm) arc(240:480:0.4cm) -- (-0.355291cm, 0.925967cm);
    \draw[line width=1pt] (120:0.4cm) arc(120:240:0.4cm);
    \draw[line width=1pt] (105:1.1cm) arc(105:255:1.1cm);
    \draw[line width=1pt, color=gray, style=dashed] (75:0.5cm) -- (75:1cm);
    \draw[line width=1pt, style=dashed] (105:0.5cm) -- (105:1.0cm);
    \draw[line width=1pt, color=gray, style=dashed] (-75:0.5cm) -- (-75:1cm);
    \draw[line width=1pt, style=dashed] (-105:0.5cm) -- (-105:1.0cm);
    \draw[line width=1pt] (-0.355291cm, 0.925967cm) -- (-0.381173cm, 1.02256cm);
    \draw[line width=1pt] (-0.355291cm, -0.925967cm) -- (-0.381173cm, -1.02256cm);
    \draw[draw=none, fill=brown!60!white] (75:0.9cm) arc(75:-75:0.9cm) -- (-75:1.0cm) arc(-75:75:1.0cm) -- (75:0.9cm);
    \draw[line width=1pt] (60:0.4cm) -- (0.381173cm, 1.02256cm);
    \draw[line width=1pt] (-60:0.4cm) -- (0.381173cm, -1.02256cm);
    \draw[line width=1pt, style=dashed] (0,0) circle(0.5cm);
    \draw[line width=1pt, style=dashed] (0.355291cm, 0.925967cm) arc(69.56:-69.56:1.0cm);
    \draw[line width=1pt, color=gray, style=dotted] (75:0.9cm) arc(75:-75:0.9cm);
    \draw[line width=1pt, style=dashed] (0.355291cm, 0.925967cm) arc(69.56:290.44:1.0cm);
    \draw[line width=1pt] (70:1cm) arc(70:-70:1.0cm);
    
    \node[text=gray] () at (90:0.75cm) {$E'$};
    \node[text=gray] () at (-90:0.75cm) {$E'$};
    \node[] () at (-0.03,0) {$C'$};
    \node[text=gray] () at (0:0.75cm) {$D'$};
    \node[] () at (180:0.75cm) {$B'$};
    \node[] () at (-1.5, 0) {$A'$};
    \node[below] () at (0, -1.5) {(a)};
    \end{scope}

    \begin{scope}[xshift=4.5cm]
    \draw[fill=yellow] (-2, -1.5) -- ++ (4, 0) -- ++ (0, 3) -- ++ (-4, 0) -- cycle;
    \draw[line width=1pt, color=gray!10!white, fill=gray!10!white] (105:0.5cm) -- (105:1.0cm) arc(105:-105:1.0cm) -- (-105:0.5cm) arc(255:105:0.5cm);
    \draw[draw=none, fill=red!60!white] (0.355291cm, 0.925967cm) -- (0.381173cm, 1.02256cm) arc(69.56:290.44:1.1cm) -- (0.355291cm, -0.925967cm) arc(290.44:69.56:1.0cm);
    \draw[line width=1pt] (0.355291cm, 0.925967cm) -- (0.381173cm, 1.02256cm) arc(69.56:290.44:1.1cm) -- (0.355291cm, -0.925967cm);
    \draw[draw=none, color=green!60!white, fill=green!60!white] (105:0.5cm) -- (105:1.1cm) arc(105:255:1.1cm) -- (255:0.5cm) arc(-105:105:0.5cm);
    \draw[line width=1pt] (249.56:1.0cm) -- (240:0.4cm) arc(240:480:0.4cm) -- (-0.355291cm, 0.925967cm);
    \draw[line width=1pt] (120:0.4cm) arc(120:240:0.4cm);
    \draw[line width=1pt] (105:1.1cm) arc(105:255:1.1cm);
    \draw[line width=1pt, color=gray, style=dashed] (75:0.5cm) -- (75:1cm);
    \draw[line width=1pt, style=dashed] (105:0.5cm) -- (105:1.0cm);
    \draw[line width=1pt, color=gray, style=dashed] (-75:0.5cm) -- (-75:1cm);
    \draw[line width=1pt, style=dashed] (-105:0.5cm) -- (-105:1.0cm);
    \draw[line width=1pt] (-0.355291cm, 0.925967cm) -- (-0.381173cm, 1.02256cm);
    \draw[line width=1pt] (-0.355291cm, -0.925967cm) -- (-0.381173cm, -1.02256cm);
    \draw[draw=none, fill=brown!60!white] (75:0.9cm) arc(75:-75:0.9cm) -- (-75:1.0cm) arc(-75:75:1.0cm) -- (75:0.9cm);
    \draw[line width=1pt] (60:0.4cm) -- (0.381173cm, 1.02256cm);
    \draw[line width=1pt] (-60:0.4cm) -- (0.381173cm, -1.02256cm);
    \draw[line width=1pt, style=dashed] (0,0) circle(0.5cm);
    \draw[line width=1pt, style=dashed] (0.355291cm, 0.925967cm) arc(69.56:-69.56:1.0cm);
    \draw[line width=1pt, color=gray, style=dotted] (75:0.9cm) arc(75:-75:0.9cm);
    \draw[line width=1pt, style=dashed] (0.355291cm, 0.925967cm) arc(69.56:290.44:1.0cm);
    \draw[line width=1pt] (70:1cm) arc(70:-70:1.0cm);
    
    \node[text=gray] () at (90:0.75cm) {$E'$};
    \node[text=gray] () at (-90:0.75cm) {$E'$};
    \node[] () at (-0.03,0) {$C'$};
    \node[text=gray] () at (0:0.75cm) {$D'$};
    \node[] () at (180:0.75cm) {$B'$};
    \node[] () at (-1.5, 0) {$A'$};
    \node[below] () at (0, -1.5) {(b)};
    \end{scope}
    
    \begin{scope}[xshift=9cm]
    \draw[fill=yellow] (-2, -1.5) -- ++ (4, 0) -- ++ (0, 3) -- ++ (-4, 0) -- cycle;
    \draw[line width=1pt, color=gray!10!white, fill=gray!10!white] (105:0.5cm) -- (105:1.0cm) arc(105:-105:1.0cm) -- (-105:0.5cm) arc(255:105:0.5cm);
    \draw[draw=none, fill=red!60!white] (0.355291cm, 0.925967cm) -- (0.381173cm, 1.02256cm) arc(69.56:290.44:1.1cm) -- (0.355291cm, -0.925967cm) arc(290.44:69.56:1.0cm);
    \draw[line width=1pt] (0.355291cm, 0.925967cm) -- (0.381173cm, 1.02256cm) arc(69.56:290.44:1.1cm) -- (0.355291cm, -0.925967cm);
    \draw[draw=none, fill=red!60!white] (105:0.5cm) -- (105:1.1cm) arc(105:255:1.1cm) -- (255:0.5cm) arc(-105:105:0.5cm);
    \draw[draw=none, color=yellow, fill=yellow] (120:0.4cm) -- (-0.381173cm, 1.02256cm) arc(105:255:1.1cm) -- (-120:0.4cm) arc(-120:120:0.4cm);
    \draw[line width=1pt] (249.56:1.0cm) -- (240:0.4cm) arc(240:480:0.4cm) -- (-0.355291cm, 0.925967cm);
    \draw[line width=1pt] (120:0.4cm) arc(120:240:0.4cm);
    \draw[line width=1pt] (105:1.1cm) arc(105:255:1.1cm);
    \draw[line width=1pt, color=gray, style=dashed] (75:0.5cm) -- (75:1cm);
    \draw[line width=1pt, style=dashed] (105:0.5cm) -- (105:1.0cm);
    \draw[line width=1pt, color=gray, style=dashed] (-75:0.5cm) -- (-75:1cm);
    \draw[line width=1pt, style=dashed] (-105:0.5cm) -- (-105:1.0cm);
    \draw[line width=1pt] (-0.355291cm, 0.925967cm) -- (-0.381173cm, 1.02256cm);
    \draw[line width=1pt] (-0.355291cm, -0.925967cm) -- (-0.381173cm, -1.02256cm);
    \draw[draw=none, fill=brown!60!white] (75:0.9cm) arc(75:-75:0.9cm) -- (-75:1.0cm) arc(-75:75:1.0cm) -- (75:0.9cm);
    \draw[line width=1pt] (60:0.4cm) -- (0.381173cm, 1.02256cm);
    \draw[line width=1pt] (-60:0.4cm) -- (0.381173cm, -1.02256cm);
    \draw[line width=1pt, style=dashed] (0,0) circle(0.5cm);
    \draw[line width=1pt, style=dashed] (0.355291cm, 0.925967cm) arc(69.56:-69.56:1.0cm);
    \draw[line width=1pt, color=gray, style=dotted] (75:0.9cm) arc(75:-75:0.9cm);
    \draw[line width=1pt, style=dashed] (0.355291cm, 0.925967cm) arc(69.56:290.44:1.0cm);
    \draw[line width=1pt] (70:1cm) arc(70:-70:1.0cm);
    
    \node[text=gray] () at (90:0.75cm) {$E'$};
    \node[text=gray] () at (-90:0.75cm) {$E'$};
    \node[] () at (-0.03,0) {$C'$};
    \node[text=gray] () at (0:0.75cm) {$D'$};
    \node[] () at (180:0.75cm) {$B'$};
    \node[] () at (-1.5, 0) {$A'$};
    \node[below] () at (0, -1.5) {(c)};
    \end{scope}

    \begin{scope}[xshift=13.5cm]
    \draw[fill=yellow] (-2, -1.5) -- ++ (4, 0) -- ++ (0, 3) -- ++ (-4, 0) -- cycle;
    \draw[line width=1pt, color=gray!10!white, fill=gray!10!white] (105:0.5cm) -- (105:1.0cm) arc(105:-105:1.0cm) -- (-105:0.5cm) arc(255:105:0.5cm);
    \draw[draw=none, color=red!40!white, fill=red!30!white] (105:0.5cm) -- (105:1.1cm) arc(105:255:1.1cm) -- (255:0.5cm) arc(-105:105:0.5cm);
    \draw[draw=none, color=red!40!white, fill=red!30!white] (0.355291cm, 0.925967cm) -- (0.381173cm, 1.02256cm) arc(69.56:290.44:1.1cm) -- (0.355291cm, -0.925967cm) arc(290.44:69.56:1.0cm);
    \draw[line width=1pt] (0.355291cm, 0.925967cm) -- (0.381173cm, 1.02256cm) arc(69.56:290.44:1.1cm) -- (0.355291cm, -0.925967cm);
    \draw[draw=none, color=yellow, fill=yellow] (120:0.4cm) -- (-0.381173cm, 1.02256cm) arc(105:255:1.1cm) -- (-120:0.4cm) arc(-120:120:0.4cm);
    \draw[line width=1pt] (249.56:1.0cm) -- (240:0.4cm) arc(240:480:0.4cm) -- (-0.355291cm, 0.925967cm);
    \draw[line width=1pt] (120:0.4cm) arc(120:240:0.4cm);
    \draw[line width=1pt] (105:1.1cm) arc(105:255:1.1cm);
    \draw[line width=1pt, color=gray, style=dashed] (75:0.5cm) -- (75:1cm);
    \draw[line width=1pt, color=gray, style=dashed] (105:0.5cm) -- (105:1.0cm);
    \draw[line width=1pt, color=gray, style=dashed] (-75:0.5cm) -- (-75:1cm);
    \draw[line width=1pt, color=gray, style=dashed] (-105:0.5cm) -- (-105:1.0cm);
    \draw[line width=1pt] (-0.355291cm, 0.925967cm) -- (-0.381173cm, 1.02256cm);
    \draw[line width=1pt] (-0.355291cm, -0.925967cm) -- (-0.381173cm, -1.02256cm);
    \draw[draw=none, fill=brown!60!white] (75:0.9cm) arc(75:-75:0.9cm) -- (-75:1.0cm) arc(-75:75:1.0cm) -- (75:0.9cm);
    \draw[line width=1pt, color=gray] (60:0.4cm) -- (0.355291cm, 0.925967cm);
    \draw[line width=1pt, color=gray] (-60:0.4cm) -- (0.355291cm, -0.925967cm);
    \draw[line width=1pt, color=gray, style=dashed] (120:0.5cm) arc(120:-120:0.5cm);
    \draw[line width=1pt, style=dashed] (120:0.5cm) arc(120:180+60:0.5cm);
    \draw[line width=1pt, color=gray, style=dashed] (0.355291cm, 0.925967cm) arc(69.56:-69.56:1.0cm);
    \draw[line width=1pt, color=gray, style=dotted] (75:0.9cm) arc(75:-75:0.9cm);
    \draw[line width=1pt, style=dashed] (180-69.56:1.0cm) arc(180-69.56:180+69.56:1.0cm);
    \draw[line width=1pt, color=gray, style=dashed] (0.355291cm, 0.925967cm) arc(69.56:180-69.56:1.0cm);
    \draw[line width=1pt, color=gray, style=dashed] (0.355291cm, -0.925967cm) arc(-69.56:-180+69.56:1.0cm);
    \draw[line width=1pt] (70:1cm) arc(70:-70:1.0cm);
    
    \node[text=gray] () at (90:0.75cm) {$E'$};
    \node[text=gray] () at (-90:0.75cm) {$E'$};
    \node[] () at (-0.03,0) {$C'$};
    \node[text=gray] () at (0:0.75cm) {$D'$};
    \node[] () at (180:0.75cm) {$B'$};
    \node[] () at (-1.5, 0) {$A'$};
    \node[below] () at (0, -1.5) {(d)};
    \end{scope}
    \end{tikzpicture}
    \caption{The third step $\Gamma_3$ illustrated. (a) $\Tr_{DE}\left(V_{AD\setminus\widetilde{D}} \rho V_{AD\setminus\widetilde{D}}^{\dagger}\right)$ over partition $A'B'C'D'E'$. (b) The unitary $K$, depicted in green. (c) $K$ applied. (d) Partial trace $\Tr_{D'E' \setminus DE}$ applied, yielding our desired $\rho_{A'B'C'}'$.}
    \label{fig:circuit_stability_1_3}
\end{figure}

\begin{figure}[h]
\centering
\begin{tikzpicture}[every path/.style={thick},scale=0.45]
    \foreach \x in {-7,...,28}
    {
    \draw[] (\x * 0.5+0.25, 0.7) -- ++ (0,2.5);
    }
    \draw[fill=yellow] (-2.5-0.45, 1.0) -- ++ (0.9, 0) -- ++ (0, 0.4) -- ++ (-0.9, 0) -- cycle;
    \draw[fill=yellow] (-1.5-0.45, 1.0) -- ++ (0.9, 0) -- ++ (0, 0.4) -- ++ (-0.9, 0) -- cycle;
    \draw[fill=yellow] (-0.5-0.45, 1.0) -- ++ (0.9, 0) -- ++ (0, 0.4) -- ++ (-0.9, 0) -- cycle;
    \draw[fill=green!60!white] (0.5-0.45, 1.0) -- ++ (0.9, 0) -- ++ (0, 0.4) -- ++ (-0.9, 0) -- cycle;
    \draw[fill=green!60!white] (1.5-0.45, 1.0) -- ++ (0.9, 0) -- ++ (0, 0.4) -- ++ (-0.9, 0) -- cycle;
    \draw[fill=green!60!white] (2.5-0.45, 1.0) -- ++ (0.9, 0) -- ++ (0, 0.4) -- ++ (-0.9, 0) -- cycle;
    \draw[fill=green!60!white] (3.5-0.45, 1.0) -- ++ (0.9, 0) -- ++ (0, 0.4) -- ++ (-0.9, 0) -- cycle;
    \draw[fill=green!60!white] (4.5-0.45, 1.0) -- ++ (0.9, 0) -- ++ (0, 0.4) -- ++ (-0.9, 0) -- cycle;
    \draw[fill=green!60!white] (5.5-0.45, 1.0) -- ++ (0.9, 0) -- ++ (0, 0.4) -- ++ (-0.9, 0) -- cycle;
    \draw[fill=green!60!white] (6.5-0.45, 1.0) -- ++ (0.9, 0) -- ++ (0, 0.4) -- ++ (-0.9, 0) -- cycle;
    \draw[fill=green!60!white] (7.5-0.45, 1.0) -- ++ (0.9, 0) -- ++ (0, 0.4) -- ++ (-0.9, 0) -- cycle;
    \draw[fill=white] (8.5-0.45, 1.0) -- ++ (0.9, 0) -- ++ (0, 0.4) -- ++ (-0.9, 0) -- cycle;
    \draw[fill=white] (9.5-0.45, 1.0) -- ++ (0.9, 0) -- ++ (0, 0.4) -- ++ (-0.9, 0) -- cycle;
    \draw[fill=yellow] (10.5-0.45, 1.0) -- ++ (0.9, 0) -- ++ (0, 0.4) -- ++ (-0.9, 0) -- cycle;
    \draw[fill=yellow] (11.5-0.45, 1.0) -- ++ (0.9, 0) -- ++ (0, 0.4) -- ++ (-0.9, 0) -- cycle;
    \draw[fill=yellow] (12.5-0.45, 1.0) -- ++ (0.9, 0) -- ++ (0, 0.4) -- ++ (-0.9, 0) -- cycle;
    \draw[fill=yellow] (13.5-0.45, 1.0) -- ++ (0.9, 0) -- ++ (0, 0.4) -- ++ (-0.9, 0) -- cycle;
    \draw[fill=yellow] (-3-0.45, 1.0) -- ++ (0.4, 0) -- ++ (0, 0.4) -- ++ (-0.4, 0);
    \draw[fill=yellow] (14+0.45, 1.0) -- ++ (-0.4, 0) -- ++ (0, 0.4) -- ++ (0.4, 0);
    \draw[fill=yellow] (-3-0.45, 1.5) -- ++ (0.9, 0) -- ++ (0, 0.4) -- ++ (-0.9, 0) -- cycle;
    \draw[fill=yellow] (-2-0.45, 1.5) -- ++ (0.9, 0) -- ++ (0, 0.4) -- ++ (-0.9, 0) -- cycle;
    \draw[fill=yellow] (-1-0.45, 1.5) -- ++ (0.9, 0) -- ++ (0, 0.4) -- ++ (-0.9, 0) -- cycle;
    \draw[fill=green!60!white] (0-0.45, 1.5) -- ++ (0.9, 0) -- ++ (0, 0.4) -- ++ (-0.9, 0) -- cycle;
    \draw[fill=green!60!white] (1-0.45, 1.5) -- ++ (0.9, 0) -- ++ (0, 0.4) -- ++ (-0.9, 0) -- cycle;
    \draw[fill=green!60!white] (2-0.45, 1.5) -- ++ (0.9, 0) -- ++ (0, 0.4) -- ++ (-0.9, 0) -- cycle;
    \draw[fill=green!60!white] (3-0.45, 1.5) -- ++ (0.9, 0) -- ++ (0, 0.4) -- ++ (-0.9, 0) -- cycle;
    \draw[fill=green!60!white] (4-0.45, 1.5) -- ++ (0.9, 0) -- ++ (0, 0.4) -- ++ (-0.9, 0) -- cycle;
    \draw[fill=green!60!white] (5-0.45, 1.5) -- ++ (0.9, 0) -- ++ (0, 0.4) -- ++ (-0.9, 0) -- cycle;
    \draw[fill=green!60!white] (6-0.45, 1.5) -- ++ (0.9, 0) -- ++ (0, 0.4) -- ++ (-0.9, 0) -- cycle;
    \draw[fill=green!60!white] (7-0.45, 1.5) -- ++ (0.9, 0) -- ++ (0, 0.4) -- ++ (-0.9, 0) -- cycle;
    \draw[fill=white] (8-0.45, 1.5) -- ++ (0.9, 0) -- ++ (0, 0.4) -- ++ (-0.9, 0) -- cycle;
    \draw[fill=white] (9-0.45, 1.5) -- ++ (0.9, 0) -- ++ (0, 0.4) -- ++ (-0.9, 0) -- cycle;
    \draw[fill=white] (10-0.45, 1.5) -- ++ (0.9, 0) -- ++ (0, 0.4) -- ++ (-0.9, 0) -- cycle;
    \draw[fill=yellow] (11-0.45, 1.5) -- ++ (0.9, 0) -- ++ (0, 0.4) -- ++ (-0.9, 0) -- cycle;
    \draw[fill=yellow] (12-0.45, 1.5) -- ++ (0.9, 0) -- ++ (0, 0.4) -- ++ (-0.9, 0) -- cycle;
    \draw[fill=yellow] (13-0.45, 1.5) -- ++ (0.9, 0) -- ++ (0, 0.4) -- ++ (-0.9, 0) -- cycle;
    \draw[fill=yellow] (14-0.45, 1.5) -- ++ (0.9, 0) -- ++ (0, 0.4) -- ++ (-0.9, 0) -- cycle;
    \draw[fill=yellow] (-2.5-0.45, 2.0) -- ++ (0.9, 0) -- ++ (0, 0.4) -- ++ (-0.9, 0) -- cycle;
    \draw[fill=yellow] (-1.5-0.45, 2.0) -- ++ (0.9, 0) -- ++ (0, 0.4) -- ++ (-0.9, 0) -- cycle;
    \draw[fill=green!60!white] (-0.5-0.45, 2.0) -- ++ (0.9, 0) -- ++ (0, 0.4) -- ++ (-0.9, 0) -- cycle;
    \draw[fill=green!60!white] (0.5-0.45, 2.0) -- ++ (0.9, 0) -- ++ (0, 0.4) -- ++ (-0.9, 0) -- cycle;
    \draw[fill=green!60!white] (1.5-0.45, 2.0) -- ++ (0.9, 0) -- ++ (0, 0.4) -- ++ (-0.9, 0) -- cycle;
    \draw[fill=green!60!white] (2.5-0.45, 2.0) -- ++ (0.9, 0) -- ++ (0, 0.4) -- ++ (-0.9, 0) -- cycle;
    \draw[fill=green!60!white] (3.5-0.45, 2.0) -- ++ (0.9, 0) -- ++ (0, 0.4) -- ++ (-0.9, 0) -- cycle;
    \draw[fill=green!60!white] (4.5-0.45, 2.0) -- ++ (0.9, 0) -- ++ (0, 0.4) -- ++ (-0.9, 0) -- cycle;
    \draw[fill=green!60!white] (5.5-0.45, 2.0) -- ++ (0.9, 0) -- ++ (0, 0.4) -- ++ (-0.9, 0) -- cycle;
    \draw[fill=green!60!white] (6.5-0.45, 2.0) -- ++ (0.9, 0) -- ++ (0, 0.4) -- ++ (-0.9, 0) -- cycle;
    \draw[fill=white] (7.5-0.45, 2.0) -- ++ (0.9, 0) -- ++ (0, 0.4) -- ++ (-0.9, 0) -- cycle;
    \draw[fill=white] (8.5-0.45, 2.0) -- ++ (0.9, 0) -- ++ (0, 0.4) -- ++ (-0.9, 0) -- cycle;
    \draw[fill=white] (9.5-0.45, 2.0) -- ++ (0.9, 0) -- ++ (0, 0.4) -- ++ (-0.9, 0) -- cycle;
    \draw[fill=white] (10.5-0.45, 2.0) -- ++ (0.9, 0) -- ++ (0, 0.4) -- ++ (-0.9, 0) -- cycle;
    \draw[fill=yellow] (11.5-0.45, 2.0) -- ++ (0.9, 0) -- ++ (0, 0.4) -- ++ (-0.9, 0) -- cycle;
    \draw[fill=yellow] (12.5-0.45, 2.0) -- ++ (0.9, 0) -- ++ (0, 0.4) -- ++ (-0.9, 0) -- cycle;
    \draw[fill=yellow] (13.5-0.45, 2.0) -- ++ (0.9, 0) -- ++ (0, 0.4) -- ++ (-0.9, 0) -- cycle;
    \draw[fill=yellow] (-3-0.45, 2.0) -- ++ (0.4, 0) -- ++ (0, 0.4) -- ++ (-0.4, 0);
    \draw[fill=yellow] (14+0.45, 2.0) -- ++ (-0.4, 0) -- ++ (0, 0.4) -- ++ (0.4, 0);
    \draw[fill=yellow] (-3-0.45, 2.5) -- ++ (0.9, 0) -- ++ (0, 0.4) -- ++ (-0.9, 0) -- cycle;
    \draw[fill=yellow] (-2-0.45, 2.5) -- ++ (0.9, 0) -- ++ (0, 0.4) -- ++ (-0.9, 0) -- cycle;
    \draw[fill=green!60!white] (-1-0.45, 2.5) -- ++ (0.9, 0) -- ++ (0, 0.4) -- ++ (-0.9, 0) -- cycle;
    \draw[fill=green!60!white] (0-0.45, 2.5) -- ++ (0.9, 0) -- ++ (0, 0.4) -- ++ (-0.9, 0) -- cycle;
    \draw[fill=green!60!white] (1-0.45, 2.5) -- ++ (0.9, 0) -- ++ (0, 0.4) -- ++ (-0.9, 0) -- cycle;
    \draw[fill=green!60!white] (2-0.45, 2.5) -- ++ (0.9, 0) -- ++ (0, 0.4) -- ++ (-0.9, 0) -- cycle;
    \draw[fill=green!60!white] (3-0.45, 2.5) -- ++ (0.9, 0) -- ++ (0, 0.4) -- ++ (-0.9, 0) -- cycle;
    \draw[fill=green!60!white] (4-0.45, 2.5) -- ++ (0.9, 0) -- ++ (0, 0.4) -- ++ (-0.9, 0) -- cycle;
    \draw[fill=green!60!white] (5-0.45, 2.5) -- ++ (0.9, 0) -- ++ (0, 0.4) -- ++ (-0.9, 0) -- cycle;
    \draw[fill=green!60!white] (6-0.45, 2.5) -- ++ (0.9, 0) -- ++ (0, 0.4) -- ++ (-0.9, 0) -- cycle;
    \draw[fill=white] (7-0.45, 2.5) -- ++ (0.9, 0) -- ++ (0, 0.4) -- ++ (-0.9, 0) -- cycle;
    \draw[fill=white] (8-0.45, 2.5) -- ++ (0.9, 0) -- ++ (0, 0.4) -- ++ (-0.9, 0) -- cycle;
    \draw[fill=white] (9-0.45, 2.5) -- ++ (0.9, 0) -- ++ (0, 0.4) -- ++ (-0.9, 0) -- cycle;
    \draw[fill=white] (10-0.45, 2.5) -- ++ (0.9, 0) -- ++ (0, 0.4) -- ++ (-0.9, 0) -- cycle;
    \draw[fill=white] (11-0.45, 2.5) -- ++ (0.9, 0) -- ++ (0, 0.4) -- ++ (-0.9, 0) -- cycle;
    \draw[fill=yellow] (12-0.45, 2.5) -- ++ (0.9, 0) -- ++ (0, 0.4) -- ++ (-0.9, 0) -- cycle;
    \draw[fill=yellow] (13-0.45, 2.5) -- ++ (0.9, 0) -- ++ (0, 0.4) -- ++ (-0.9, 0) -- cycle;
    \draw[fill=yellow] (14-0.45, 2.5) -- ++ (0.9, 0) -- ++ (0, 0.4) -- ++ (-0.9, 0) -- cycle;
    \draw [decorate, decoration = {calligraphic brace}] (0.5, 3.275) --  (2.5, 3.275);
    \node[] at (1.5, 4.25) {$B$};
    \draw [decorate, decoration = {calligraphic brace}] (2.5, 3.275) --  (8.0, 3.275);
    \node[] at (5.25, 4.25) {$C$};
    \draw [decorate, decoration = {calligraphic brace}] (8.0, 3.275) --  (12.0, 3.275);
    \node[] at (10.0, 4.25) {$D$};
    \draw [decorate, decoration = {calligraphic brace,mirror}] (-1.5, 0.625) --  (6.0, 0.625);
    \node[] at (2.25, -0.35) {$B'C'$};
    \node[] at (-1.6, 4.25) {$A$};
    \node[] at (13.25, 4.25) {$A$};
    \node[] at (-4.2, 2.0) {$\cdots$};
    \node[] at (15.2, 2.0) {$\cdots$};
    \draw[->, line width=1pt] (-5.2, 0.9) -- node[left] {\small time} (-5.2, 3.0);

\end{tikzpicture}
    \caption{Schematic illustration of the third step, for the cross-section in Fig.~\ref{fig:circuit_stability_1_1_plc}. The yellow gates are $V_{AD\setminus\widetilde{D}}$ and the green gates are $K$.}
    \label{fig:circuit_stability_1_3_plc}
\end{figure}

\sloppy
Finally, we apply $\mathcal{I}_{A'} \otimes \Gamma_3$ to obtain the post-circuit reduced density matrix $\rho_{A'B'C'}'$. (See Figs.~\ref{fig:circuit_stability_1_3} and~\ref{fig:circuit_stability_1_3_plc} for graphical illustrations.) From the first two steps, $U = F_{BE} V_{CD} V_{AD\setminus\widetilde{D}}$. Let $K$ be the past light-cone of $A'B'C'$ under $F_{BE}V_{CD}$. By Lemma~\ref{lemma:light_cone}(ii),
\begin{equation}\label{eq:K_restores}
    \rho_{A'B'C'}' = \Tr_{D'E'}\left( K V_{AD\setminus\widetilde{D}} \rho V_{AD\setminus\widetilde{D}}^{\dagger} K^{\dagger} \right).
\end{equation}
Note that $K$ acts trivially on $A'$, $D$, and $E$. First, neither $F_{BE}$ nor $V_{CD}$ acts on $A' = A\setminus BE(d)$. Therefore, every gate of $K$ is connected to $B'C'$ by a sequence of at most $d$ gates, and $K$ is supported on $B'C'(d)$. However, $B'C'(d)$ is disjoint from $E$, because $B'C'$ is disjoint from $E' = E(d)$. It is also disjoint from $D$, because $C'(d) \subseteq C$ and $\operatorname{dist}(D,B') \geq \operatorname{dist}(D,B) - d > d$. Thus $K$ is supported on $(A\setminus A')BC$, on which our channel acts after the second step. We choose $\Gamma_3(\cdot) := \Tr_{D'E'\setminus DE}\left(K (\cdot) K^{\dagger}\right)$ and obtain:
\begin{equation}\label{eq:end_of_step3}
\begin{aligned}
        \mathcal{I}_{A'} \otimes (\Gamma_3 \circ \Gamma_2 \circ \Gamma_1)(\rho_{A'B'E'}') &= \Tr_{D'E' \setminus DE} \Bigl( K \Tr_{DE} \left( V_{AD\setminus\widetilde{D}} \rho V_{AD\setminus\widetilde{D}}^{\dagger}  \right) K^{\dagger} \Bigr) \\
        &= \Tr_{D'E'} \left( K V_{AD\setminus\widetilde{D}} \rho V_{AD\setminus\widetilde{D}}^{\dagger} K^{\dagger} \right) \\
        &= \rho_{A'B'C'}'.
\end{aligned}
\end{equation}
\end{proof}

We now extend our line of reasoning to prove Theorem~\ref{theorem:circuit_stability_3}, the analogue of Theorem~\ref{theorem:circuit_stability_1} for the partition in Fig.~\ref{fig:extendible_state_examples}(c)-(d). Suppose that $\rho$ a state on the partition in Fig.~\ref{fig:circuit_stability_partitions_3}(a)-(b), which is equivalent to the partition in Fig.~\ref{fig:extendible_state_examples}(c)-(d). If $\rho_{BC\widetilde{D}E}$ is extendible from $BE$ to $BC$, then $\rho' := U \rho U^{\dagger}$ is extendible from $B'E'$ to $B'C'$, where $A'B'C'D'E'$ is defined in Fig.~\ref{fig:circuit_stability_partitions_3}(c).

As we can see, Theorems~\ref{theorem:circuit_stability_1} and \ref{theorem:circuit_stability_3} are virtually identical -- their only distinction is which partition of Fig.~\ref{fig:extendible_state_examples} they pertain to. Thus, we proceed with an essentially identical proof for Theorem~\ref{theorem:circuit_stability_3}, constructing a post-circuit extending channel out of parts of $U$, partial traces, and the pre-circuit extending channel. For brevity's sake, we illustrate the proof of Theorem~\ref{theorem:circuit_stability_3} through Fig.~\ref{fig:circuit_stability_3} in Appendix~\ref{appendix:circuit_stability_3}, with reasoning analogous to the proof of the previous Theorem~\ref{theorem:circuit_stability_1}.

\begin{figure}[h]
\centering
\begin{tikzpicture}[line width=1pt]
    \begin{scope}[xshift=0.0cm]
    \draw[] (-2, -1.5) -- ++ (4, 0) -- ++ (0, 3) -- ++ (-4, 0) -- cycle;
    \draw[] (0,0) circle (0.5cm);
    \draw[] (0,0) circle (1cm);
    \draw[] (30:0.5cm) -- (30:1cm);
    \draw[] (150:0.5cm) -- (150:1cm);
    \draw[] (-30:0.5cm) -- (-30:1cm);
    \draw[] (-150:0.5cm) -- (-150:1cm);
    \draw[] (60:0.5cm) -- (60:1cm);
    \draw[] (120:0.5cm) -- (120:1cm);
    \draw[] (-60:0.5cm) -- (-60:1cm);
    \draw[] (-120:0.5cm) -- (-120:1cm);
    
    \node[] () at (90:0.75cm) {\small $D$};
    \node[] () at (0:0.75cm) {\small $B$};
    \node[] () at (0,0) {\small $C$};
    \node[] () at (180:0.75cm) {\small $B$};
    \node[] () at (270:0.75cm) {\small $D$};
    \node[] () at (45:0.75cm) {\small $E$};
    \node[] () at (-45:0.75cm) {\small $E$};
    \node[] () at (135:0.75cm) {\small $E$};
    \node[] () at (-135:0.75cm) {\small $E$};
    \node[] () at (-1.5, 0) {\small $A$};
    \node[below] () at (0, -1.5) {(a)};
    \end{scope}
    \begin{scope}[xshift=4.5cm]
    \draw[] (-2, -1.5) -- ++ (4, 0) -- ++ (0, 3) -- ++ (-4, 0) -- cycle;
    \draw[] (0,0) circle (0.5cm);
    \draw[] (0,0) circle (1cm);
    \draw[] (30:0.5cm) -- (30:1cm);
    \draw[] (150:0.5cm) -- (150:1cm);
    \draw[] (-30:0.5cm) -- (-30:1cm);
    \draw[] (-150:0.5cm) -- (-150:1cm);
    \draw[] (60:0.5cm) -- (60:1cm);
    \draw[] (120:0.5cm) -- (120:1cm);
    \draw[] (-60:0.5cm) -- (-60:1cm);
    \draw[] (-120:0.5cm) -- (-120:1cm);
    \draw[style=dotted] (60:0.9cm) arc(60:120:0.9cm);
    \draw[style=dotted] (-60:0.9cm) arc(-60:-120:0.9cm);
    
    \node[] () at (90:0.7cm) {\small $\widetilde{D}$};
    \node[] () at (0:0.75cm) {\small $B$};
    \node[] () at (0,0) {\small $C$};
    \node[] () at (180:0.75cm) {\small $B$};
    \node[] () at (270:0.7cm) {\small $\widetilde{D}$};
    \node[] () at (45:0.75cm) {\small $E$};
    \node[] () at (-45:0.75cm) {\small $E$};
    \node[] () at (135:0.75cm) {\small $E$};
    \node[] () at (-135:0.75cm) {\small $E$};
    \node[] () at (-1.5, 0) {\small $A$};
    \node[below] () at (0, -1.5) {(b)};
    \end{scope}
    \begin{scope}[xshift=9.0cm]
    \draw[] (-2, -1.5) -- ++ (4, 0) -- ++ (0, 3) -- ++ (-4, 0) -- cycle;
    \draw[style=dashed] (0,0) circle (0.5cm);
    \draw[] (0,0) circle (0.4cm);
    \draw[style=dashed] (0,0) circle (1cm);
    \draw[] (18.69:0.4cm) -- (24.29:1.1cm) arc(24.29:65.71:1.1cm) -- (71.31:0.4cm);
    \draw[] (-18.69:0.4cm) -- (-24.29:1.1cm) arc(-24.29:-65.71:1.1cm) -- (-71.31:0.4cm);
    \draw[] (180-18.69:0.4cm) -- (180-24.29:1.1cm) arc(180-24.29:180-65.71:1.1cm) -- (180-71.31:0.4cm);
    \draw[] (180+18.69:0.4cm) -- (180+24.29:1.1cm) arc(180+24.29:180+65.71:1.1cm) -- (180+71.31:0.4cm);
    \draw[] (24.29:1.1cm) arc(24.29:-24.29:1.1cm);
    \draw[] (180-24.29:1.1cm) arc(180-24.29:180+24.29:1.1cm);
    \draw[] (67:1.0cm) arc(67:180-67:1.0cm);
    \draw[] (-67:1.0cm) arc(-67:-180+67:1.0cm);
    \draw[style=dashed] (30:0.5cm) -- (30:1cm);
    \draw[style=dashed] (150:0.5cm) -- (150:1cm);
    \draw[style=dashed] (-30:0.5cm) -- (-30:1cm);
    \draw[style=dashed] (-150:0.5cm) -- (-150:1cm);
    \draw[style=dashed] (60:0.5cm) -- (60:1cm);
    \draw[style=dashed] (120:0.5cm) -- (120:1cm);
    \draw[style=dashed] (-60:0.5cm) -- (-60:1cm);
    \draw[style=dashed] (-120:0.5cm) -- (-120:1cm);
    
    \node[] () at (90:0.75cm) {\small $D'$};
    \node[] () at (0:0.75cm) {\small $B'$};
    \node[] () at (0,0) {\small $C'$};
    \node[] () at (180:0.75cm) {\small $B'$};
    \node[] () at (270:0.75cm) {\small $D'$};
    \node[] () at (45:0.75cm) {\small $E'$};
    \node[] () at (-45:0.75cm) {\small $E'$};
    \node[] () at (135:0.75cm) {\small $E'$};
    \node[] () at (-135:0.75cm) {\small $E'$};
    \node[] () at (-1.5, 0) {\small $A'$};
    \node[below] () at (0, -1.5) {(c)};
    \end{scope}
\end{tikzpicture}
\caption{(a) The pre-circuit partition $ABCDE$, equivalent to the one in Fig.~\ref{fig:extendible_state_examples}(d). $A$ is a purifying system of $BCDE$. (b) The same partition $ABCDE$, with subsystem $\widetilde{D} \subset D$ shown. $\widetilde{D}$ consists of all sites in $D$ that are at distance more than $d$ from $A$. (c) The topologically equivalent post-circuit partition $A'B'C'D'E'$. The individual subsystems are: $B' := BE(d) \setminus E(d)$, $C' := C(-d)$, $D' := BCDE \setminus BE(d)C(-d)$, $E' := E(d)$, and $A'$ is a purifying system of $B'C'D'E'$. (Recall notation in Section~\ref{subsec:notation}.) For comparison, the partition $ABCDE$ is overlaid with dashed lines.}
\label{fig:circuit_stability_partitions_3}
\end{figure}

\begin{theorem}\label{theorem:circuit_stability_3}
    Let us partition a system $\Lambda$ as in Fig.~\ref{fig:circuit_stability_partitions_3}(a)-(b) for $\Lambda = ABCDE$, and as in Fig.~\ref{fig:circuit_stability_partitions_3}(c) for $\Lambda = A'B'C'D'E'$. Let $\rho$ be a state on $\Lambda$, $U$ be a geometrically local depth-$d$ circuit supported on $\Lambda$, and $\rho' := U \rho U^{\dagger}$ be the post-circuit state. Suppose $\operatorname{dist}(A, C) > d$ and $\operatorname{dist}(B, D) > 2d$. If $\rho_{BC\widetilde{D}E}$ is locally extendible from $BE$ to $BC$, then $\rho_{B'C'D'E'}'$ is locally extendible from $B'E'$ to $B'C'$.
\end{theorem}

\begin{proof}
The proof is essentially identical to the proof of Theorem~\ref{theorem:circuit_stability_1}; the only distinction is that the partition is defined in Fig.~\ref{fig:circuit_stability_partitions_3}. We therefore simply illustrate the proof through Fig.~\ref{fig:circuit_stability_3}, which is analogous to Figs.~\ref{fig:circuit_stability_1_1}, \ref{fig:circuit_stability_1_2}, and \ref{fig:circuit_stability_1_3}, and present an otherwise identical reasoning.
\end{proof}

Here, we make a few remarks on our theorems. First, notice that all of our relevant partitions $BCDE$ and $B'C'D'E'$ are all topological balls of some constant radius. If we denote $r$ as the radius of $BCDE$, the radius of $B'C'D'E'$ is another constant $r + \mathcal{O}(d)$. Consequently by our theorems, if the pre-circuit state on any local ball of constant radius is extendible, then the post-circuit state on the (enlarged) ball of constant radius is also extendible.

Second, the post-circuit extending channel $\Gamma$ has depth dependent solely on the depth of the pre-circuit extending channel $\Phi$. In our proofs, $\Gamma$ was constructed as the composition of three types of channels: two parts of $U$ ($F_{BE}^{\dagger}$ and $K$), partial traces, and the extending channel $\Phi$ of the pre-circuit state $\rho$. Since $U$ is of constant depth $d$, these parts are also of depth at most $d$. Since the partial traces are over subsystems of constant size, the partial traces are of constant depth. Then, if the depth of $\Phi$ is also constant, the resulting depth of $\Gamma$ is a (larger) constant as well.

Combining these two remarks, we reach an important conclusion -- if there is a constant-depth extending channel $\Phi_{BE \rightarrow BC}$ for any local ball of constant radius before $U$ is applied, there is also a constant-depth extending channel $\Gamma_{B'E' \rightarrow B'C'}$ for a larger concentric ball of constant radius after $U$ is applied. In other words, given some local ball in a system, local extendibility is stable under any geometrically local depth-$d$ circuit $U$, with some enlarging of the ball radius. Of course, such extending channels can be obtained directly from the reduced density matrices on the local balls [Section~\ref{subsec:checking_extendibility}]. We have proven in this Section that any such extendible reduced density matrix remains extendible under a geometrically local constant-depth circuit.

The above remarks readily extend to spatial dimensions higher than two. For concreteness, we had specified Theorems~\ref{theorem:circuit_stability_1}-\ref{theorem:circuit_stability_3} in two spatial dimensions, so our partitions were 2-dimensional balls of some constant radius. It is not difficult to generalize the arguments of Theorems~\ref{theorem:circuit_stability_1}-\ref{theorem:circuit_stability_3} to $d$-dimensional lattices -- we simply set our partitions $BCDE$ and $B'C'D'E'$ as topological $d$-dimensional balls of some constant radius. All of our existing arguments follow through for such $d$-dimensional balls, without loss of generality.

Of course, our ultimate task is to learn circuits on systems of sizes much larger than these constant radii. It still remains to understand how to compose together extending channels of localized constant-radius balls, and prepare the global state. This task is the topic of Sections~\ref{sec:extendibility_quantum_phase} and~\ref{sec:learning}.

\section{Extendibility and quantum phases of matter}\label{sec:extendibility_quantum_phase}

In Section~\ref{sec:locally_extendible_states}, we discussed properties of locally extendible states. In this Section, we discuss the relevance of such states to quantum phases of matter. Mathematically, two ground states of local Hamiltonians are in the same phase if there is an adiabatic path between the two ground states with a uniform lower bound on the energy gap. In the literature, sometimes a more restrictive (but a more convenient) definition is used, based on the constant-depth circuit~\cite{Chen2010}. In this definition, two states are said to be in the same phase if there is a geometrically local constant-depth circuit that converts one state to the other, up to an usage of extra ancillary qubits. The latter is what we will use throughout the following discussion. 

Broadly speaking, quantum phases of matter are often divided into three categories. The first and the simplest category is short-range entangled states. These are states that can be obtained by applying a constant-depth quantum circuit to a product state~\cite{Chen2010}. That is, a state $|\psi\rangle$ is short-range entangled if $|\psi\rangle = U|0\ldots 0\rangle$ for some constant-depth quantum circuit $U$. The second category, which subsumes the first one, is called invertible states. Formally, a state $|\psi\rangle$ is invertible if $|\psi\rangle \otimes |\varphi\rangle= U|0\ldots 0\rangle$ for some constant-depth circuit $U$ and a state $|\varphi\rangle$. Although invertible states may seem similar to short-range entangled states, the former is expected to be a larger class than the latter.\footnote{An example of a state that is invertible but not short-range entangled is Kitaev's $E_8$ state~\cite{kitaev2006anyons}.} 

The last category is known as long-range entangled states. These are states that cannot be created from a product state by a constant-depth circuit, and as such, require an alternative formulation. We will consider a class of long-range entangled states that satisfy a certain homogeneity condition~\cite{DraftOther}. These are states that satisfy a certain set of local extendibility conditions, which we discuss in Section~\ref{subsec:lre}. This class includes the toric code~\cite{Kitaev2003}, string-net model~\cite{Levin2005}, and more generally, the states obtained by applying a constant-depth circuit to those states. 

For all the states in these categories, their state preparation circuits can be obtained by using the extendibility of their local reduced density matrices. We describe these details in Section~\ref{subsec:invertible}, ~\ref{subsec:lre}, and~\ref{subsec:lre_circuit}.

\subsection{Short-range entangled states}
\label{subsec:invertible}

In this Section, we explain how to construct a state-preparation circuit for short-range entangled states, and more generally, invertible states. Without loss of generality, we assume that these states are obtained by applying a constant-depth circuit to a product state. For concreteness, we will focus on the case of two dimensions. Generalization to higher dimensions is straightforward, and we briefly mention the key ideas towards the end of this Section.

The task at hand is the following. Given a set of reduced density matrices over $\mathcal{O}(1)$-sized balls, how can we infer the state preparation circuit? The main observation in solving this problem is that the reduced density matrices of these states are locally extendible, in the sense we make precise below. Without loss of generality, we consider a state $|\psi\rangle$ such that
\begin{equation}
    |\psi\rangle \otimes |\varphi\rangle = U|0\ldots 0\rangle,
\end{equation}
where $U$ is a geometrically local constant-depth circuit on a finite-dimensional lattice. The state $|\varphi\rangle$ is a state on an ancillary system, which is generally not available. As extendibility is stable under $U$ [Section~\ref{subsec:stability_circuit}], the states $|\psi\rangle \otimes |\varphi\rangle$ are locally extendible with respect to the subsystems shown in Fig.~\ref{fig:sre_extendibility}(a), provided that the thicknesses of the subsystems are large compared to the circuit depth. Moreover, because the fidelity of recovery is multiplicative~\cite{Berta2016}, if $\sigma_{BCDE} \otimes \tau_{BCDE}$ is locally extendible from $BE$ to $BC$, $\sigma_{BCDE}$ is locally extendible from $BE$ to $BC$ (and so is $\tau_{BCDE}$). Therefore, the reduced density matrices of the state $|\psi\rangle$ over the subsystems in Fig.~\ref{fig:sre_extendibility}(a) are locally extendible. 

The extendibility of Fig.~\ref{fig:sre_extendibility}(b) can be also shown straightforwardly. We sketch a short argument below, and defer additional details to Appendix~\ref{appendix:stability_simple}. Prior to applying the circuit, it trivially follows that $I(A:C)_{\rho}=0$. That is, the correlation between $A$ and $C$ is zero. Because a depth-$d$ circuit cannot create any correlation between any two regions that are separated by a distance of at least $2d+1$, if the thickness of $B$ is larger than $2d+1$, the correlation between $A$ and $C$ remains zero. Then the Petz map from $B$ to $BC$ realizes the extension. Theorem~\ref{theorem:circuit_stability_2} in Appendix~\ref{appendix:stability_simple} uses a different argument but arrives at the same conclusion.

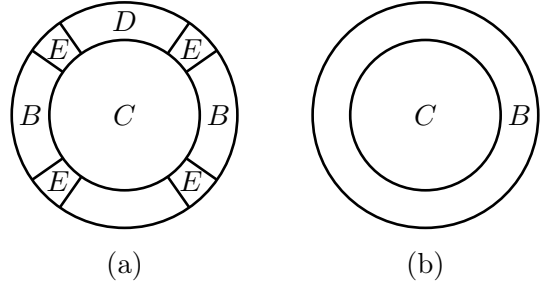
\begin{figure}[h]
    \centering
    \begin{tikzpicture}[line width=1pt]
        \draw[] (0,0) circle (1cm);
        \draw[] (0,0) circle (1.5cm);
        \draw[] (35:1cm) -- (35:1.5cm);
        \draw[] (55:1cm) -- (55:1.5cm);
        
        \draw[] (125:1cm) -- (125:1.5cm);
        \draw[] (145:1cm) -- (145:1.5cm);
        
        \draw[] (215:1cm) -- (215:1.5cm);
        \draw[] (235:1cm) -- (235:1.5cm);
        
        \draw[] (305:1cm) -- (305:1.5cm);
        \draw[] (325:1cm) -- (325:1.5cm);

        \node[] () at (0,0) {$C$};
        \node[] () at (1.25,0) {$B$};
        \node[] () at (-1.25,0) {$B$};
        \node[] () at (0,1.25) {$D$};
        \node[] () at (45:1.25cm) {$E$};
        \node[] () at (135:1.25cm) {$E$};
        \node[] () at (225:1.25cm) {$E$};
        \node[] () at (315:1.25cm) {$E$};

        \node[] () at (0, -2cm) {(a)};

        \begin{scope}[xshift=4cm]
        \draw[] (0,0) circle (1cm);
        \draw[] (0,0) circle (1.5cm);

        \node[] () at (0,0) {$C$};
        \node[] () at (1.25,0) {$B$};

        \node[] () at (0, -2cm) {(b)};
            
        \end{scope}
    \end{tikzpicture}
    \caption{For short-range entangled states, any density matrix over the shown regions (a) and (b) are locally extendible from $BE$ to $BC$. In (b) $E$ is an empty set.}
    \label{fig:sre_extendibility}
\end{figure}

Now we aim to gradually build up the global state by applying a sequence of quantum channels that can be inferred from local reduced density matrices. At a high level, this procedure is described in Fig.~\ref{fig:sre_build_up}. First, we apply a channel on disjoint regions so that on each region we obtain a reduced density matrix of $|\psi\rangle$ over that region [Fig.~\ref{fig:sre_build_up}(a)]. We refer to this channel as $\Phi_0$. By construction, this is a tensor product over the channels acting on each region. Next, we apply the extending map associated with Fig.~\ref{fig:sre_extendibility}(a) so that we connect the blue regions in Fig.~\ref{fig:sre_build_up}(a). This is again a tensor product of channels acting on disjoint bounded regions. We refer to this channel as $\Phi_1$. At this point, we obtain a reduced density matrix of $|\psi\rangle$ over the union of the red and the blue region in Fig.~\ref{fig:sre_build_up}(b). The yellow region corresponds to the removed region $E$ in Fig.~\ref{fig:sre_extendibility}(a). The complement of the blue and the red region is a union of disjoint (topological) balls. Finally, using the extending map associated with Fig.~\ref{fig:sre_extendibility}(b), we obtain the global state $|\psi\rangle$. This is again a tensor product of channels on bounded regions. This channel is referred to as $\Phi_2$. (See Fig.~\ref{fig:sre_stability} for a detailed view of each local connection, and how they correspond to each partition in Fig.~\ref{fig:sre_extendibility}). 

The composition of these extending maps define a constant-depth channel that converts a given product state to the state $|\psi\rangle$. By the Stinespring dilation theorem, we can convert each local channel into a unitary and an ancillary system. Thus we obtain a constant-depth circuit that converts a given state to $|\psi\rangle$ and an ancillary state. 

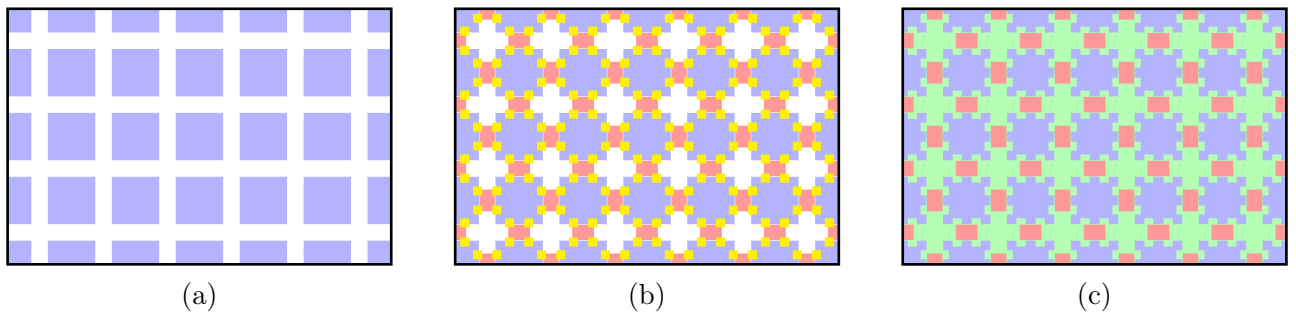
\begin{figure}[h]
    \centering
    \begin{tikzpicture}[line width=1pt, scale=0.74]
    \node[] () at (0, -2.5) {(a)};
    \begin{scope}
    \clip (-3, -2) -- ++ (6, 0) -- ++ (0, 4) -- ++ (-6, 0) -- cycle;
    \foreach \x in {-4, ..., 4}
    {
    \foreach \y in {-2, ..., 2}
    {
        \filldraw[blue!30!white] (\x-0.35, \y-0.35) -- ++ (0.7, 0) -- ++ (0, 0.7) -- ++ (-0.7, 0) -- cycle;
    }
    }
    \end{scope}
    \draw[] (-3, -2) -- ++ (6, 0) -- ++ (0, 4) -- ++ (-6, 0) -- cycle;

    \begin{scope}[xshift=7cm]
    \node[] () at (0, -2.5) {(b)};
    \begin{scope}
    \clip (-3, -2) -- ++ (6, 0) -- ++ (0, 4) -- ++ (-6, 0) -- cycle;
     \foreach \x in {-4, ..., 4}
    {
    \foreach \y in {-2, ..., 2}
    {
        \filldraw[blue!30!white] (\x-0.35, \y-0.35) -- ++ (0.7, 0) -- ++ (0, 0.7) -- ++ (-0.7, 0) -- cycle;
    }
    }
    \foreach \x in {-4, ..., 4}
    {
    \foreach \y in {-2, ..., 2}
    {
        \filldraw[red!40!white] (\x-0.15, \y+0.4) -- ++ (0.3, 0) -- ++ (0, 0.2) -- ++ (-0.3, 0) -- cycle;
    }
    }

    \foreach \x in {-4, ..., 4}
    {
    \foreach \y in {-2, ..., 2}
    {
        \filldraw[yellow] (\x-0.2, \y+0.4) -- ++ (0.1,0) -- ++ (0, -0.1) -- ++(-0.1, 0) -- cycle;
        \filldraw[yellow] (\x+0.1, \y+0.4) -- ++ (0.1,0) -- ++ (0, -0.1) -- ++(-0.1, 0) -- cycle;
        \filldraw[yellow] (\x-0.2, \y+0.7) -- ++ (0.1,0) -- ++ (0, -0.1) -- ++(-0.1, 0) -- cycle;
        \filldraw[yellow] (\x+0.1, \y+0.7) -- ++ (0.1,0) -- ++ (0, -0.1) -- ++(-0.1, 0) -- cycle;
    }
    }
    
    \foreach \x in {-4, ..., 4}
    {
    \foreach \y in {-2, ..., 2}
    {
        \filldraw[red!40!white] (\x-0.6, \y-0.15) -- ++ (0.2, 0) -- ++ (0, 0.3) -- ++ (-0.2, 0) -- cycle;
    }
    }

    \foreach \x in {-4, ..., 4}
    {
    \foreach \y in {-2, ..., 2}
    {
        \filldraw[yellow] (\x-0.6, \y-0.2) -- ++ (-0.1, 0) -- ++ (0, 0.1) -- ++ (0.1, 0) -- cycle;
        \filldraw[yellow] (\x-0.6, \y+0.1) -- ++ (-0.1, 0) -- ++ (0, 0.1) -- ++ (0.1, 0) -- cycle;
        \filldraw[yellow] (\x-0.3, \y-0.2) -- ++ (-0.1, 0) -- ++ (0, 0.1) -- ++ (0.1, 0) -- cycle;
        \filldraw[yellow] (\x-0.3, \y+0.1) -- ++ (-0.1, 0) -- ++ (0, 0.1) -- ++ (0.1, 0) -- cycle;
    }
    }
    \end{scope}
    \draw[] (-3, -2) -- ++ (6, 0) -- ++ (0, 4) -- ++ (-6, 0) -- cycle;
    \end{scope}
    
    \begin{scope}[xshift=14cm]
    
    \node[] () at (0, -2.5) {(c)};
    \filldraw[green!30!white] (-3, -2) -- ++ (6, 0) -- ++ (0, 4) -- ++ (-6, 0) -- cycle;
    \begin{scope}
    \clip (-3, -2) -- ++ (6, 0) -- ++ (0, 4) -- ++ (-6, 0) -- cycle;
     \foreach \x in {-4, ..., 4}
    {
    \foreach \y in {-2, ..., 2}
    {
        \filldraw[blue!30!white] (\x-0.35, \y-0.35) -- ++ (0.7, 0) -- ++ (0, 0.7) -- ++ (-0.7, 0) -- cycle;
    }
    }
    \foreach \x in {-4, ..., 4}
    {
    \foreach \y in {-2, ..., 2}
    {
        \filldraw[green!30!white] (\x-0.2, \y+0.4) -- ++ (0.1,0) -- ++ (0, -0.1) -- ++(-0.1, 0) -- cycle;
        \filldraw[green!30!white] (\x+0.1, \y+0.4) -- ++ (0.1,0) -- ++ (0, -0.1) -- ++(-0.1, 0) -- cycle;
        \filldraw[green!30!white] (\x-0.2, \y+0.7) -- ++ (0.1,0) -- ++ (0, -0.1) -- ++(-0.1, 0) -- cycle;
        \filldraw[green!30!white] (\x+0.1, \y+0.7) -- ++ (0.1,0) -- ++ (0, -0.1) -- ++(-0.1, 0) -- cycle;
    }
    }
    \foreach \x in {-4, ..., 4}
    {
    \foreach \y in {-2, ..., 2}
    {
        \filldraw[red!40!white] (\x-0.15, \y+0.4) -- ++ (0.3, 0) -- ++ (0, 0.2) -- ++ (-0.3, 0) -- cycle;
    }
    }
    \foreach \x in {-4, ..., 4}
    {
    \foreach \y in {-2, ..., 2}
    {
        \filldraw[green!30!white] (\x-0.6, \y-0.2) -- ++ (-0.1, 0) -- ++ (0, 0.1) -- ++ (0.1, 0) -- cycle;
        \filldraw[green!30!white] (\x-0.6, \y+0.1) -- ++ (-0.1, 0) -- ++ (0, 0.1) -- ++ (0.1, 0) -- cycle;
        \filldraw[green!30!white] (\x-0.3, \y-0.2) -- ++ (-0.1, 0) -- ++ (0, 0.1) -- ++ (0.1, 0) -- cycle;
        \filldraw[green!30!white] (\x-0.3, \y+0.1) -- ++ (-0.1, 0) -- ++ (0, 0.1) -- ++ (0.1, 0) -- cycle;
    }
    }
    \foreach \x in {-4, ..., 4}
    {
    \foreach \y in {-2, ..., 2}
    {
        \filldraw[red!40!white] (\x-0.6, \y-0.15) -- ++ (0.2, 0) -- ++ (0, 0.3) -- ++ (-0.2, 0) -- cycle;
    }
    }
    \end{scope}
    \draw[] (-3, -2) -- ++ (6, 0) -- ++ (0, 4) -- ++ (-6, 0) -- cycle;
    \end{scope}
    \end{tikzpicture}
    \caption{A schematic description on how to prepare a short-range entangled state $|\psi\rangle$ in 2D. }
    \label{fig:sre_build_up}
\end{figure}

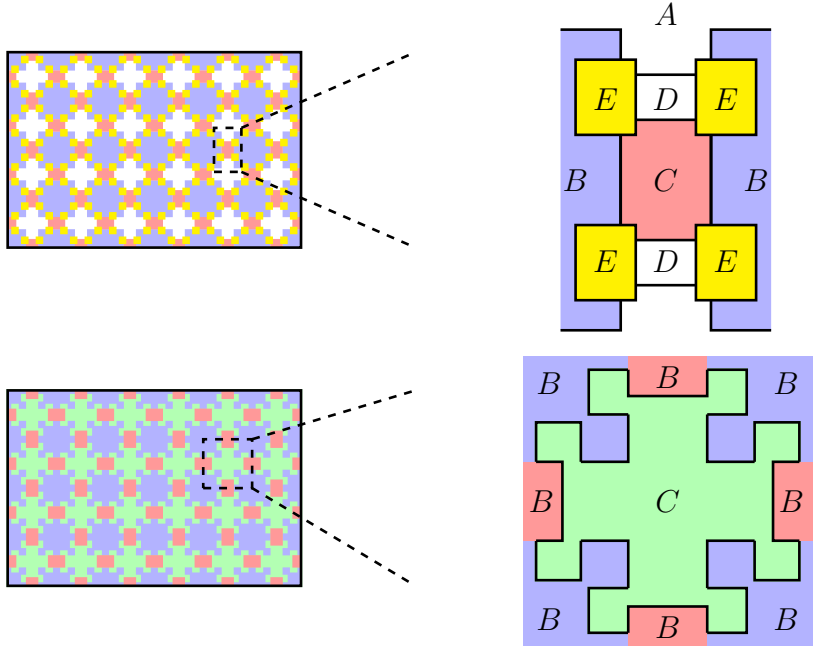
\begin{figure}[h]
\centering
\begin{tikzpicture}[line width=1pt]
    \begin{scope}[xshift=0.0cm, scale=0.65]
        \begin{scope}
        \clip (-3, -2) -- ++ (6, 0) -- ++ (0, 4) -- ++ (-6, 0) -- cycle;
         \foreach \x in {-4, ..., 4}
        {
        \foreach \y in {-2, ..., 2}
        {
            \filldraw[blue!30!white] (\x-0.35, \y-0.35) -- ++ (0.7, 0) -- ++ (0, 0.7) -- ++ (-0.7, 0) -- cycle;
        }
        }
        \foreach \x in {-4, ..., 4}
        {
        \foreach \y in {-2, ..., 2}
        {
            \filldraw[red!40!white] (\x-0.15, \y+0.4) -- ++ (0.3, 0) -- ++ (0, 0.2) -- ++ (-0.3, 0) -- cycle;
        }
        }
    
        \foreach \x in {-4, ..., 4}
        {
        \foreach \y in {-2, ..., 2}
        {
            \filldraw[yellow] (\x-0.2, \y+0.4) -- ++ (0.1,0) -- ++ (0, -0.1) -- ++(-0.1, 0) -- cycle;
            \filldraw[yellow] (\x+0.1, \y+0.4) -- ++ (0.1,0) -- ++ (0, -0.1) -- ++(-0.1, 0) -- cycle;
            \filldraw[yellow] (\x-0.2, \y+0.7) -- ++ (0.1,0) -- ++ (0, -0.1) -- ++(-0.1, 0) -- cycle;
            \filldraw[yellow] (\x+0.1, \y+0.7) -- ++ (0.1,0) -- ++ (0, -0.1) -- ++(-0.1, 0) -- cycle;
        }
        }
        
        \foreach \x in {-4, ..., 4}
        {
        \foreach \y in {-2, ..., 2}
        {
            \filldraw[red!40!white] (\x-0.6, \y-0.15) -- ++ (0.2, 0) -- ++ (0, 0.3) -- ++ (-0.2, 0) -- cycle;
        }
        }
    
        \foreach \x in {-4, ..., 4}
        {
        \foreach \y in {-2, ..., 2}
        {
            \filldraw[yellow] (\x-0.6, \y-0.2) -- ++ (-0.1, 0) -- ++ (0, 0.1) -- ++ (0.1, 0) -- cycle;
            \filldraw[yellow] (\x-0.6, \y+0.1) -- ++ (-0.1, 0) -- ++ (0, 0.1) -- ++ (0.1, 0) -- cycle;
            \filldraw[yellow] (\x-0.3, \y-0.2) -- ++ (-0.1, 0) -- ++ (0, 0.1) -- ++ (0.1, 0) -- cycle;
            \filldraw[yellow] (\x-0.3, \y+0.1) -- ++ (-0.1, 0) -- ++ (0, 0.1) -- ++ (0.1, 0) -- cycle;
        }
        }
        \end{scope}
    \draw[] (-3, -2) -- ++ (6, 0) -- ++ (0, 4) -- ++ (-6, 0) -- cycle;
    \draw[style=dashed] (1.5-0.28, 0-0.45) -- ++ (0.56, 0) -- ++ (0, 0.9) -- ++ (-0.56, 0) -- cycle;
    \draw[style=dashed] (1.5-0.28+0.56, 0-0.45+0.9) -- (5.34, 2.0);
    \draw[style=dashed] (1.5-0.28+0.56, 0-0.45) -- (5.34, -2.0);
    \end{scope}
    \begin{scope}[xshift=7.0cm, scale=0.8]
    \draw[fill=blue!30!white] (-2.0, -3.0) -- ++ (1.0, 0) -- ++ (0, 5.0) -- ++ (-1.0, 0); 
    \draw[fill=blue!30!white] (1.5, -3.0) -- ++ (-1.0, 0) -- ++ (0, 5.0) -- ++ (1.0, 0); 
    \draw[fill=red!40!white] (-1.0, -1.5) -- ++ (1.5, 0) -- ++ (0, 2.0) -- ++ (-1.5, 0) -- cycle;
    \draw[] (-1.0, -2.25) -- ++ (1.5, 0) -- ++ (0, 3.5) -- ++ (-1.5, 0) -- cycle;
    \draw[fill=yellow] (-1.75, -2.5) -- ++ (1.0, 0) -- ++ (0, 1.25) -- ++ (-1.0, 0) -- cycle;
    \draw[fill=yellow] (1.25, -2.5) -- ++ (-1.0, 0) -- ++ (0, 1.25) -- ++ (1.0, 0) -- cycle;
    \draw[fill=yellow] (1.25, 1.5) -- ++ (-1.0, 0) -- ++ (0, -1.25) -- ++ (1.0, 0) -- cycle;
    \draw[fill=yellow] (-1.75, 1.5) -- ++ (1.0, 0) -- ++ (0, -1.25) -- ++ (-1.0, 0) -- cycle;

    \node[] () at (-1.75, -0.5) {\large $B$};
    \node[] () at (1.25, -0.5) {\large $B$};
    \node[] () at (-0.25, -0.5) {\large $C$};
    \node[] () at (-0.25, -1.85) {\large $D$};
    \node[] () at (-0.25, 0.85) {\large $D$};
    \node[] () at (-1.25, -1.85) {\large $E$};
    \node[] () at (0.75, -1.85) {\large $E$};
    \node[] () at (0.75, 0.85) {\large $E$};
    \node[] () at (-1.25, 0.85) {\large $E$};
    \node[] () at (-0.25, 2.25) {\large $A$};
    \end{scope}
    \begin{scope}[xshift=0.0cm, yshift=-4.5cm, scale=0.65]
        \begin{scope}
        \filldraw[green!30!white] (-3, -2) -- ++ (6, 0) -- ++ (0, 4) -- ++ (-6, 0) -- cycle;
        \clip (-3, -2) -- ++ (6, 0) -- ++ (0, 4) -- ++ (-6, 0) -- cycle;
         \foreach \x in {-4, ..., 4}
        {
        \foreach \y in {-2, ..., 2}
        {
            \filldraw[blue!30!white] (\x-0.35, \y-0.35) -- ++ (0.7, 0) -- ++ (0, 0.7) -- ++ (-0.7, 0) -- cycle;
        }
        }
        \foreach \x in {-4, ..., 4}
        {
        \foreach \y in {-2, ..., 2}
        {
            \filldraw[green!30!white] (\x-0.2, \y+0.4) -- ++ (0.1,0) -- ++ (0, -0.1) -- ++(-0.1, 0) -- cycle;
            \filldraw[green!30!white] (\x+0.1, \y+0.4) -- ++ (0.1,0) -- ++ (0, -0.1) -- ++(-0.1, 0) -- cycle;
            \filldraw[green!30!white] (\x-0.2, \y+0.7) -- ++ (0.1,0) -- ++ (0, -0.1) -- ++(-0.1, 0) -- cycle;
            \filldraw[green!30!white] (\x+0.1, \y+0.7) -- ++ (0.1,0) -- ++ (0, -0.1) -- ++(-0.1, 0) -- cycle;
        }
        }
        \foreach \x in {-4, ..., 4}
        {
        \foreach \y in {-2, ..., 2}
        {
            \filldraw[red!40!white] (\x-0.15, \y+0.4) -- ++ (0.3, 0) -- ++ (0, 0.2) -- ++ (-0.3, 0) -- cycle;
        }
        }
        \foreach \x in {-4, ..., 4}
        {
        \foreach \y in {-2, ..., 2}
        {
            \filldraw[green!30!white] (\x-0.6, \y-0.2) -- ++ (-0.1, 0) -- ++ (0, 0.1) -- ++ (0.1, 0) -- cycle;
            \filldraw[green!30!white] (\x-0.6, \y+0.1) -- ++ (-0.1, 0) -- ++ (0, 0.1) -- ++ (0.1, 0) -- cycle;
            \filldraw[green!30!white] (\x-0.3, \y-0.2) -- ++ (-0.1, 0) -- ++ (0, 0.1) -- ++ (0.1, 0) -- cycle;
            \filldraw[green!30!white] (\x-0.3, \y+0.1) -- ++ (-0.1, 0) -- ++ (0, 0.1) -- ++ (0.1, 0) -- cycle;
        }
        }
        \foreach \x in {-4, ..., 4}
        {
        \foreach \y in {-2, ..., 2}
        {
            \filldraw[red!40!white] (\x-0.6, \y-0.15) -- ++ (0.2, 0) -- ++ (0, 0.3) -- ++ (-0.2, 0) -- cycle;
        }
        }
        \end{scope}
    \draw[] (-3, -2) -- ++ (6, 0) -- ++ (0, 4) -- ++ (-6, 0) -- cycle;
    \draw[style=dashed] (1.35-0.35, 0.35-0.35) -- ++ (1.0, 0) -- ++ (0, 1.0) -- ++ (-1.0, 0) -- cycle;
    \draw[style=dashed] (1.35-0.35+1.0, 0.35-0.35+1.0) -- (5.34, 2.0);
    \draw[style=dashed] (1.35-0.35+1.0, 0.35-0.35) -- (5.34, -2.0);
    \end{scope}
    \begin{scope}[xshift=7.0cm, yshift=-4.5cm,  scale=0.7]
    \draw[draw=none, fill=green!30!white] (-3.0, -3.0) -- (2.5, -3.0) -- (2.5, 2.5) -- (-3.0, 2.5);
    \draw[draw=none, fill=blue!30!white] (-1.0, -3.0) -- ++ (0, 2.0) -- ++ (-2.0, 0) -- ++ (0, -2.0);
    \draw[draw=none, fill=blue!30!white] (-1.0, 2.5) -- ++ (0, -2.0) -- ++ (-2.0, 0) -- ++ (0, 2.0);
    \draw[draw=none, fill=blue!30!white] (0.5, -3.0) -- ++ (0, 2.0) -- ++ (2.0, 0) -- ++ (0, -2.0);
    \draw[draw=none, fill=blue!30!white] (0.5, 2.5) -- ++ (0, -2.0) -- ++ (2.0, 0) -- ++ (0, 2.0);
    
    \draw[color=green!30!white, fill=green!30!white] (-1.75, -2.75) -- ++ (0.75, 0) -- ++ (0, 0.85) -- ++ (-0.75, 0) -- cycle;
    \draw[] (-1.75+0.75, -2.75) -- ++ (-0.75, 0) -- ++ (0, 0.85) -- ++ (0.75, 0);
    \draw[color=green!30!white, fill=green!30!white] (-2.75, -1.75) -- ++ (0.85, 0) -- ++ (0, 0.75) -- ++ (-0.85, 0) -- cycle;
    \draw[] (-2.75, -1.75+0.75) -- ++ (0, -0.75) -- ++ (0.85, 0) -- ++ (0, 0.75);
    \draw[color=green!30!white, fill=green!30!white] (-1.75, 1.4) -- ++ (0.75, 0) -- ++ (0, 0.85) -- ++ (-0.75, 0) -- cycle;
    \draw[] (-1.75+0.75, 1.4) -- ++ (-0.75, 0) -- ++ (0, 0.85) -- ++ (0.75, 0);
    \draw[color=green!30!white, fill=green!30!white] (-2.75, 0.5) -- ++ (0.85, 0) -- ++ (0, 0.75) -- ++ (-0.85, 0) -- cycle;
    \draw[] (-2.75, 0.5) -- ++ (0, 0.75) -- ++ (0.85, 0) -- ++ (0, -0.75);
    \draw[color=green!30!white, fill=green!30!white] (0.5, -2.75) -- ++ (0.75, 0) -- ++ (0, 0.85) -- ++ (-0.75, 0) -- cycle;
    \draw[] (0.5, -2.75) -- ++ (0.75, 0) -- ++ (0, 0.85) -- ++ (-0.75, 0);
    \draw[color=green!30!white, fill=green!30!white] (1.4, -1.75) -- ++ (0.85, 0) -- ++ (0, 0.75) -- ++ (-0.85, 0) -- cycle;
    \draw[] (1.4, -1.75+0.75) -- ++ (0, -0.75) -- ++ (0.85, 0)  -- ++ (0, 0.75);
    \draw[color=green!30!white, fill=green!30!white] (1.4, 0.5) -- ++ (0.85, 0) -- ++ (0, 0.75) -- ++ (-0.85, 0) -- cycle;
    \draw[] (1.4, 0.5) -- ++ (0, 0.75) -- ++ (0.85, 0) -- ++ (0, -0.75);
    \draw[color=green!30!white, fill=green!30!white] (0.5, 1.4) -- ++ (0.75, 0) -- ++ (0, 0.85) -- ++ (-0.75, 0) -- cycle;
    \draw[] (0.5, 1.4) -- ++ (0.75, 0) -- ++ (0, 0.85) -- ++ (-0.75, 0);

    \draw[draw=none, fill=red!40!white] (-3.0, -1.0) -- (-2.25, -1.0) -- ++ (0, 1.5) -- ++ (-0.75, 0);
    \draw[] (-2.75, -1.0) -- ++ (0.5, 0) -- ++ (0, 1.5) -- ++ (-0.5, 0);
    \draw[draw=none, fill=red!40!white] (-1.0, -3.0) -- ++ (0, 0.75) -- ++ (1.5, 0) -- ++ (0, -0.75);
    \draw[] (-1.0, -2.75) -- ++ (0, 0.5) -- ++ (1.5, 0) -- ++ (0, -0.5);
    \draw[draw=none, fill=red!40!white] (-1.0, 2.5) -- ++ (0, -0.75) -- ++ (1.5, 0) -- ++ (0, 0.75);
    \draw[] (-1.0, 2.25) -- ++ (0, -0.5) -- ++ (1.5, 0) -- ++ (0, 0.5);
    \draw[draw=none, fill=red!40!white] (2.5, -1.0) -- ++ (-0.75, 0.0) -- ++ (0, 1.5) -- ++ (0.75, 0);
    \draw[] (2.25, -1.0) -- ++ (-0.5, 0.0) -- ++ (0, 1.5) -- ++ (0.5, 0);

    \draw[] (-1.0, -2.75+0.85) -- ++ (0, 0.9) -- ++ (-0.9, 0);
    \draw[] (-1.0, 1.4) -- ++ (0, -0.9) -- ++ (-0.9, 0);
    \draw[] (0.5, -2.75+0.85) -- ++ (0, 0.9) -- ++ (0.9, 0);
    \draw[] (0.5, 1.4) -- ++ (0, -0.9) -- ++ (0.9, 0);

    \node[] () at (-2.5, -2.5) {\large $B$};
    \node[] () at (-2.5, 2.0) {\large $B$};
    \node[] () at (2.0, 2.0) {\large $B$};
    \node[] () at (2.0, -2.5) {\large $B$};
    \node[] () at (-2.65, -0.25) {\large $B$};
    \node[] () at (-0.25, 2.1) {\large $B$};
    \node[] () at (2.1, -0.25) {\large $B$};
    \node[] () at (-0.25, -2.65) {\large $B$};
    \node[] () at (-0.25, -0.25) {\large $C$};
    \end{scope}
\end{tikzpicture}
\caption{A close-up view of a local connection in Fig.~\ref{fig:sre_build_up}. Notice that the partitions are equivalent to those in Fig.~\ref{fig:sre_extendibility} (thus equivalent to Fig.~\ref{fig:circuit_stability_partitions_3}).}
\label{fig:sre_stability}
\end{figure}

A generalization of this approach to higher dimensions is possible, and we sketch the main idea here. Note that Fig.~\ref{fig:sre_build_up} can be viewed as a sequence of channels that creates the fattened $2$-, $1$-, and $0$-cells, in that order. In $D$ dimensions, we can use the partition of Ref.~\cite{Huang2024} (Section 5.5 of its arXiv version), in which the $k$-cells of a cubic lattice are fattened to thickness $(D-k)R$, and create the fattened $D$-, $(D-1)$-, $\ldots$, $0$-cells in that order. At each step, the erased region $E$ lies near the lower-dimensional cells that are yet to be created, and we choose it to be small compared to $C$. By applying an additional partial trace, we then obtain the reduced density matrix over a slightly thinner version of the region in the original partition, so the next step can proceed in the same manner. At the last step, the $0$-cells are disjoint balls, and we can use Fig.~\ref{fig:sre_extendibility}(b), for which $E$ is empty. Choosing $R$ to be a sufficiently large constant multiple of $d$, the existence of each extending map follows from Section~\ref{subsec:stability_circuit}. A precise description of the subsystems for general $D$ requires a careful specification of the regions involved, which is beyond the scope of this paper. Thus we conclude that any short-range entangled states (and more generally, invertible states) can be obtained by a finite-depth quantum channel, with each local channel being the extending map associated with the local reduced density matrix of the given state.

\begin{remark}
The depth of the obtained circuit is $\mathcal{O}(dD)$, because the depth of each extending map is at most $\mathcal{O}(D)$.
\end{remark}

\subsection{Long-range entangled states: setup and motivation}
\label{subsec:lre}

In this Section and Section~\ref{subsec:lre_circuit}, we explain how to construct a state-preparation circuit for a certain class of long-range entangled states. We first begin by describing the setup and the underlying motivation. How to construct the circuit is discussed in Section~\ref{subsec:lre_circuit}.

As before, we will focus on the case of two dimensions, with the goal of constructing the circuit from a set of local reduced density matrices on $\mathcal{O}(1)$-sized disks. Analogous setups for higher dimensions can be found in Ref.~\cite{DraftOther,shi2024remotedetectabilityentanglementbootstrap}. The class of states we consider are quantum many-body states on a two-dimensional disk. This state will be assumed to satisfy the local extendibility conditions for the set of subsystems shown in Fig.~\ref{fig:lre_conditions}(b)-(e). 

\begin{figure}[h]
    \centering
    \begin{tikzpicture}[line width=1pt, scale=0.8]
        \begin{scope}[yshift= 2.5cm]
        \draw[dashed] (0,0) -- ++ (12, 0) -- ++ (0, 3) -- ++ (-12, 0) -- cycle;
        \node[below] () at (6,-0.25) {(a)};
        \begin{scope}[xshift= 2cm, yshift=1cm, scale=0.3]
        \draw[] (0,0) circle (1cm);
        \draw[] (0,0) circle (1.5cm);

        \draw[] (80:1cm) -- (80:1.5cm);
        \draw[] (100:1cm) -- (100:1.5cm);
        \draw[] (260:1cm) -- (260:1.5cm);
        \draw[] (280:1cm) -- (280:1.5cm);
        \end{scope}

        \begin{scope}[xshift= 4cm, yshift=1.5cm, scale=0.3]
        \draw[] (0,0) circle (1cm);
        \draw[] (0,0) circle (1.5cm);
        \end{scope}

        \begin{scope}[xshift=7cm, scale=0.3]
        \clip (-2, 0) -- ++ (4,0) -- ++ (0, 2) -- ++ (-4, 0)--cycle;
        \draw[] (0,0) circle (1cm);
        \draw[] (0,0) circle (1.5cm);
        \draw[] (80:1cm) -- (80:1.5cm);
        \draw[] (100:1cm) -- (100:1.5cm);
        \end{scope}

        \begin{scope}[xshift=10cm, scale=0.3]
        \clip (-2, 0) -- ++ (4,0) -- ++ (0, 2) -- ++ (-4, 0)--cycle;
        \draw[] (0,0) circle (1cm);
        \draw[] (0,0) circle (1.5cm);
        \end{scope}

        \end{scope}
        \draw[] (0,0) circle (1cm);
        \draw[] (0,0) circle (1.5cm);

        \draw[] (80:1cm) -- (80:1.5cm);
        \draw[] (100:1cm) -- (100:1.5cm);
        \draw[] (260:1cm) -- (260:1.5cm);
        \draw[] (280:1cm) -- (280:1.5cm);

        \node[] () at (0,0) {$C$};
        \node[] () at (1.25,0) {$D$};
        \node[] () at (-1.25,0) {$B$};
        \node[] () at (90:1.25cm) {$E$};
        \node[] () at (270:1.25cm) {$E$};

        \node[] () at (0, -2cm) {(b)};

        \begin{scope}[xshift=4cm]
        \draw[] (0,0) circle (1cm);
        \draw[] (0,0) circle (1.5cm);

        \node[] () at (0,0) {$C$};
        \node[] () at (1.25,0) {$B$};

        \node[] () at (0, -2cm) {(c)};
            
        \end{scope}
        \begin{scope}[xshift=8cm,yshift=-0.75cm]
        \begin{scope}
        \clip (-2, 0) -- ++ (4,0) -- ++ (0, 2) -- ++ (-4, 0)--cycle;
        \draw[] (0,0) circle (1cm);
        \draw[] (0,0) circle (1.5cm);
        \draw[] (80:1cm) -- (80:1.5cm);
        \draw[] (100:1cm) -- (100:1.5cm);

        \node[] () at (0,0.25) {$C$};
        \node[] () at (135:1.25cm) {$B$};
        \node[] () at (45:1.25cm) {$D$};
        \node[] () at (90:1.25cm) {$E$};
        \end{scope}
        \draw[dashed] (-1.75,0) -- ++ (3.5, 0);

        \node[] () at (0, -1.25cm) {(d)};
            
        \end{scope}
        \begin{scope}[xshift=12cm,yshift=-0.75cm]
        \begin{scope}
        \clip (-2, 0) -- ++ (4,0) -- ++ (0, 2) -- ++ (-4, 0)--cycle;
        \draw[] (0,0) circle (1cm);
        \draw[] (0,0) circle (1.5cm);

        \node[] () at (0,0.25) {$C$};
        \node[] () at (0,1.25) {$B$};
        \end{scope}
        \draw[dashed] (-1.75,0) -- ++ (3.5, 0);
        \node[] () at (0, -1.25cm) {(e)};
            
        \end{scope}
    \end{tikzpicture}
    \caption{(a) We consider a many-body state on a disk that satisfies a set of local extendibility conditions. Its reduced density matrices are locally extendible from $BE$ to $BC$ for the shown subsystems in the bulk (b)-(c) and at the boundary (d)-(e). In (d)-(e), the dashed line represents the physical boundary.}
    \label{fig:lre_conditions}
\end{figure}
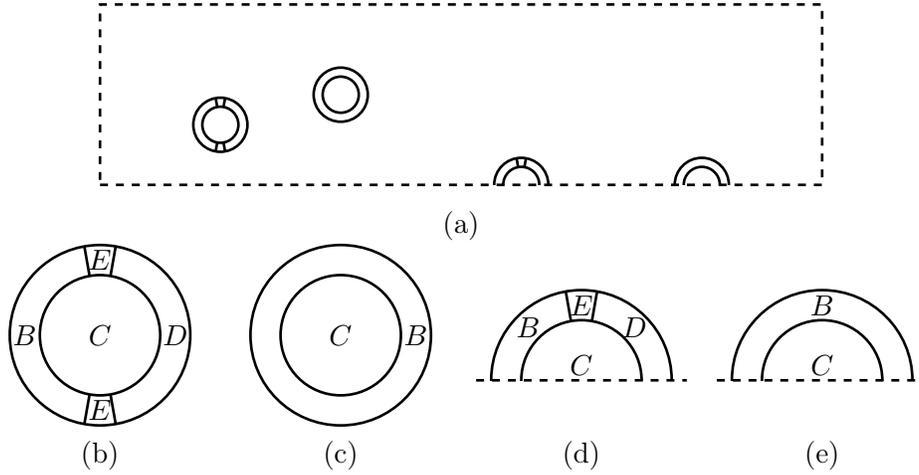

The reader may wonder about the motivation for imposing these conditions. In brief, this is to include a large class of long-range entangled states that are known in the literature. Physically, the states we consider should be viewed as ground states of some gapped local Hamiltonian on a disk. A set of conditions that include such states were recently proposed in Ref.~\cite{shi2020fusion,Shi2021}, summarized as follows:
\begin{equation}
    \begin{aligned}
        S(\rho_{BC}) - S(\rho_B) + S(\rho_{CD}) - S(\rho_D)&=0 & \text{for Fig.~\ref{fig:lre_conditions_old}(b) and (d)},\\
        S(\rho_{BC}) + S(\rho_C) - S(\rho_B)&=0 & \text{for Fig.~\ref{fig:lre_conditions_old}(c) and (e)},
    \end{aligned}
    \label{eq:eb_axioms}
\end{equation}
where $\rho$ is the underlying global state. One example that satisfies Eq.~\eqref{eq:eb_axioms} is the surface code~\cite{bravyi1998quantum} with a uniform boundary condition, e.g., a boundary condition that is smooth everywhere. More generally, Eq.~\eqref{eq:eb_axioms} are satisfied by the string-net model~\cite{Levin2005} with a uniform boundary condition~\cite{kitaev2012models}. Thus these conditions encapsulate a large class of models that are known in the literature.

\begin{figure}[h]
    \centering
    \begin{tikzpicture}[line width=1pt, scale=0.8]
        \begin{scope}[yshift= 2.5cm]
        \draw[dashed] (0,0) -- ++ (12, 0) -- ++ (0, 3) -- ++ (-12, 0) -- cycle;
        \node[below] () at (6,-0.25) {(a)};
        \begin{scope}[xshift= 2cm, yshift=1cm, scale=0.3]
        \draw[] (0,0) circle (1cm);
        \draw[] (0,0) circle (1.5cm);

        \draw[] (90:1cm) -- (90:1.5cm);
        \draw[] (270:1cm) -- (270:1.5cm);
        \end{scope}

        \begin{scope}[xshift= 4cm, yshift=1.5cm, scale=0.3]
        \draw[] (0,0) circle (1cm);
        \draw[] (0,0) circle (1.5cm);
        \end{scope}

        \begin{scope}[xshift=7cm, scale=0.3]
        \clip (-2, 0) -- ++ (4,0) -- ++ (0, 2) -- ++ (-4, 0)--cycle;
        \draw[] (0,0) circle (1cm);
        \draw[] (0,0) circle (1.5cm);
        \draw[] (90:1cm) -- (90:1.5cm);
        \end{scope}

        \begin{scope}[xshift=10cm, scale=0.3]
        \clip (-2, 0) -- ++ (4,0) -- ++ (0, 2) -- ++ (-4, 0)--cycle;
        \draw[] (0,0) circle (1cm);
        \draw[] (0,0) circle (1.5cm);
        \end{scope}

        \end{scope}
        \draw[] (0,0) circle (1cm);
        \draw[] (0,0) circle (1.5cm);

        \draw[] (90:1cm) -- (90:1.5cm);
        \draw[] (270:1cm) -- (270:1.5cm);

        \node[] () at (0,0) {$C$};
        \node[] () at (1.25,0) {$D$};
        \node[] () at (-1.25,0) {$B$};

        \node[] () at (0, -2cm) {(b)};

        \begin{scope}[xshift=4cm]
        \draw[] (0,0) circle (1cm);
        \draw[] (0,0) circle (1.5cm);

        \node[] () at (0,0) {$C$};
        \node[] () at (1.25,0) {$B$};

        \node[] () at (0, -2cm) {(c)};
            
        \end{scope}
        \begin{scope}[xshift=8cm,yshift=-0.75cm]
        \begin{scope}
        \clip (-2, 0) -- ++ (4,0) -- ++ (0, 2) -- ++ (-4, 0)--cycle;
        \draw[] (0,0) circle (1cm);
        \draw[] (0,0) circle (1.5cm);
        \draw[] (90:1cm) -- (90:1.5cm);

        \node[] () at (0,0.25) {$C$};
        \node[] () at (135:1.25cm) {$B$};
        \node[] () at (45:1.25cm) {$D$};
        \end{scope}
        \draw[dashed] (-1.75,0) -- ++ (3.5, 0);

        \node[] () at (0, -1.25cm) {(d)};
            
        \end{scope}
        \begin{scope}[xshift=12cm,yshift=-0.75cm]
        \begin{scope}
        \clip (-2, 0) -- ++ (4,0) -- ++ (0, 2) -- ++ (-4, 0)--cycle;
        \draw[] (0,0) circle (1cm);
        \draw[] (0,0) circle (1.5cm);

        \node[] () at (0,0.25) {$C$};
        \node[] () at (0,1.25) {$B$};
        \end{scope}
        \draw[dashed] (-1.75,0) -- ++ (3.5, 0);
        \node[] () at (0, -1.25cm) {(e)};
            
        \end{scope}
    \end{tikzpicture}
    \caption{(a) The conditions initially proposed in Ref.~\cite{shi2020fusion,Shi2021} for describing long-range entangled states in two dimensions. For (b) and (d), the condition is $S(\rho_{BC}) - S(\rho_B) + S(\rho_{CD}) - S(\rho_D)$=0, where $\rho$ is the ground state and $S(\rho)$ is the von Neumann entropy of $\rho$. For (c) and (e), $S(\rho_{BC}) + S(\rho_C) - S(\rho_B)=0$.}
    \label{fig:lre_conditions_old}
\end{figure}
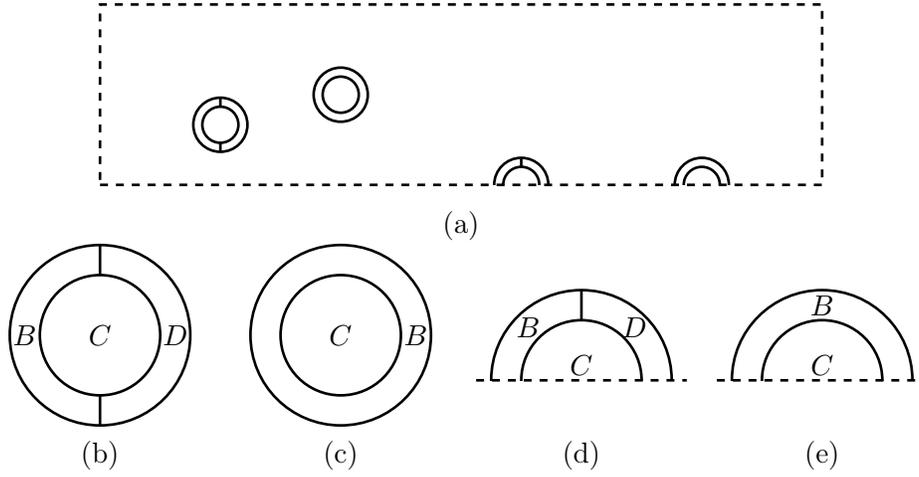

While Eq.~\eqref{eq:eb_axioms} may appear unrelated to the conditions in Fig.~\ref{fig:lre_conditions}, there is in fact an intimate relation. By introducing a purifying system (denoted as $A$), Eq.~\eqref{eq:eb_axioms} can be converted to the following equations:
\begin{equation}
    \begin{aligned}
        S(\rho_{AB}) - S(\rho_B) + S(\rho_{BC}) - S(\rho_{ABC})&=0 & \text{for Fig.~\ref{fig:lre_conditions_old}(a)-(d)}.
    \end{aligned}
    \label{eq:eb_axioms_purified}
\end{equation}
An important consequence of these assumptions is that there is an extending map $\Phi_{B\to BC}: \mathcal{D}_B \to \mathcal{D}_{BC}$ (often referred to as the Petz recovery map~\cite{Petz1988,Petz2003}) such that
\begin{equation}
    \rho_{ABC} = \mathcal{I}_A\otimes \Phi_{B\to BC}(\rho_{AB}).
\end{equation}
That is, Eq.~\eqref{eq:eb_axioms} implies that the density matrices $\rho_{BCD}$ for the subsystems in Fig.~\ref{fig:lre_conditions_old}(b) and (d) are extendible from $B$ to $BC$. Similarly, the density matrices $\rho_{BC}$ in Fig.~\ref{fig:lre_conditions_old}(c) and (e) are extendible from $B$ to $BC$.

However, an undesirable feature of Eq.~\eqref{eq:eb_axioms} is that they can break under application of a constant-depth circuit. This is known as the ``spurious'' contribution to the topological entanglement entropy~\cite{Kitaev2006,Levin2006,Zou2016,Williamson2019,Kato2020,kim2023universal}. On the other hand, the extendibility conditions in Fig.~\ref{fig:lre_conditions} remain invariant under constant-depth circuits, as long as the thicknesses of the subsystems are sufficiently large compared to the depth of the circuit. 

Thus the extendibility conditions in Fig.~\ref{fig:lre_conditions} are applicable to a broader class of long-range entangled states. It holds for a broad class of known exactly solvable models such as the toric code~\cite{Kitaev2003} and the string-net model~\cite{Levin2005} with a uniform boundary condition, and more generally, even the states obtained by applying a generic constant-depth circuit to such states. Its broad applicability and the robustness against constant-depth circuit is what makes those conditions useful. 

That said, not all gapped ground states satisfy our conditions. These conditions cannot be satisfied everywhere if the underlying state lies in a topologically ordered ground state subspace that includes another orthogonal state. For instance, if we consider the surface code consisting of both smooth and rough boundary condition, its ground state subspace is spanned by two locally indistinguishable states~\cite{bravyi1998quantum}. Any circuit construction approach based on local reduced density matrices (like ours) must necessarily fail for such states. Otherwise, one would end up with a conclusion that two orthogonal states can be created from identical circuits, which is a contradiction. Extending our method to such states is an important open problem, which we discuss further in Section~\ref{sec:discussion}.

\subsection{Long-range entangled states: circuit construction}
\label{subsec:lre_circuit}

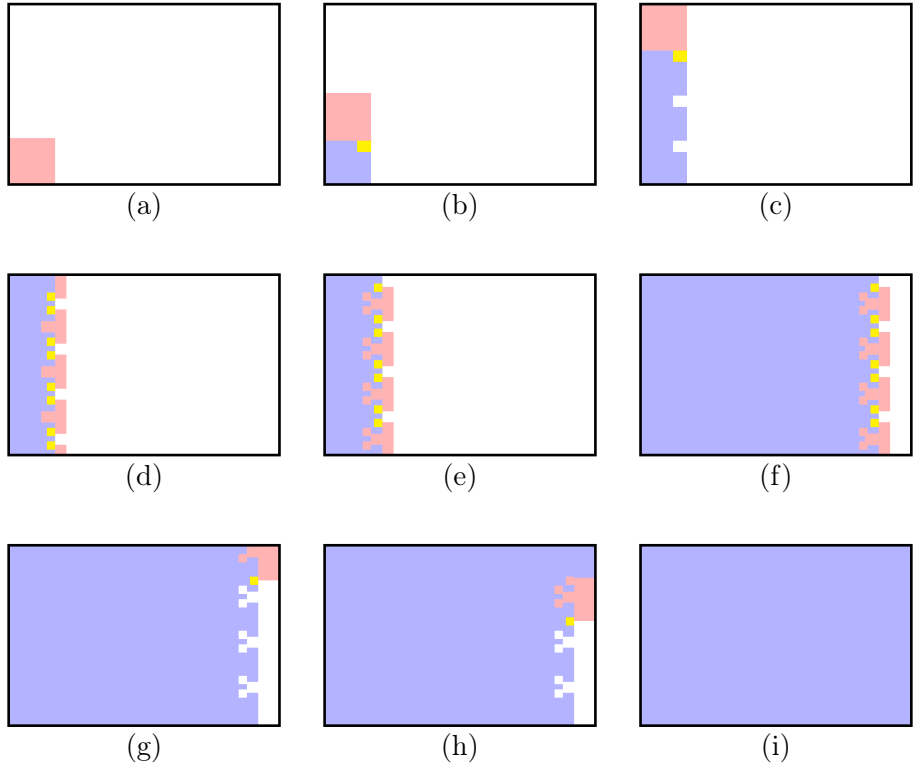
\begin{figure}[H]
    \centering
    \begin{tikzpicture}[line width=1pt, scale=0.6]
    \node[] () at (0, -2.5) {(a)};
    \begin{scope}
    \clip (-3, -2) -- ++ (6, 0) -- ++ (0, 4) -- ++ (-6, 0) -- cycle;
    \filldraw[red!30!white] (-3, -2) -- ++ (1, 0) -- ++ (0, 1) -- ++ (-1, 0) -- cycle;
    \end{scope}
    \draw[] (-3, -2) -- ++ (6, 0) -- ++ (0, 4) -- ++ (-6, 0) -- cycle;

    \begin{scope}[xshift=7cm]
    \node[] () at (0, -2.5) {(b)};
    \begin{scope}
    \clip (-3, -2) -- ++ (6, 0) -- ++ (0, 4) -- ++ (-6, 0) -- cycle;
    \filldraw[blue!30!white] (-3, -2) -- ++ (1, 0) -- ++ (0, 1) -- ++ (-1, 0) -- cycle;
    \filldraw[yellow] (-2.25, -1.25) -- ++ (0.25, 0) -- ++ (0, 0.25) -- ++ (-0.25, 0) -- cycle;

    \filldraw[red!30!white] (-3, -2+1) -- ++ (1, 0) -- ++ (0, 1) -- ++ (-1, 0) -- cycle;
    \end{scope}
    \draw[] (-3, -2) -- ++ (6, 0) -- ++ (0, 4) -- ++ (-6, 0) -- cycle;
        
    \end{scope}

    \begin{scope}[xshift=14cm]
    \node[] () at (0, -2.5) {(c)};
    \begin{scope}
    \clip (-3, -2) -- ++ (6, 0) -- ++ (0, 4) -- ++ (-6, 0) -- cycle;
    \filldraw[blue!30!white] (-3, -2) -- ++ (1, 0) -- ++ (0, 1) -- ++ (-1, 0) -- cycle;
    \filldraw[white] (-2.25, -1.25) -- ++ (0.25, 0) -- ++ (0, 0.25) -- ++ (-0.25, 0) -- cycle;
    \filldraw[blue!30!white] (-3, -2+1) -- ++ (1, 0) -- ++ (0, 1) -- ++ (-1, 0) -- cycle;
    \filldraw[white] (-2.25, -0.25) -- ++ (0.25, 0) -- ++ (0, 0.25) -- ++ (-0.25, 0) -- cycle;
    \filldraw[blue!30!white] (-3, -2+2) -- ++ (1, 0) -- ++ (0, 1) -- ++ (-1, 0) -- cycle;
    \filldraw[yellow] (-2.25, 0.75) -- ++ (0.25, 0) -- ++ (0, 0.25) -- ++ (-0.25, 0) -- cycle;
    \filldraw[red!30!white] (-3, -2+3) -- ++ (1, 0) -- ++ (0, 1) -- ++ (-1, 0) -- cycle;
    \end{scope}
    \draw[] (-3, -2) -- ++ (6, 0) -- ++ (0, 4) -- ++ (-6, 0) -- cycle;
        
    \end{scope}

    \begin{scope}[yshift=-6cm]
    \node[] () at (0, -2.5) {(d)};
    \begin{scope}
    \clip (-3, -2) -- ++ (6, 0) -- ++ (0, 4) -- ++ (-6, 0) -- cycle;
    \filldraw[red!30!white] (-2.5, -2.5) -- ++ (0.75,0) -- ++ (0,0.7) -- ++ (-1, 0) -- cycle;
    \filldraw[red!30!white] (-2.5, -1.5) -- ++ (0.75,0) -- ++ (0,0.7) -- ++ (-1, 0) -- cycle;
    \filldraw[red!30!white] (-2.5, -0.5) -- ++ (0.75,0) -- ++ (0,0.7) -- ++ (-1, 0) -- cycle;
    \filldraw[red!30!white] (-2.5, 0.5) -- ++ (0.75,0) -- ++ (0,0.7) -- ++ (-1, 0) -- cycle;
    \filldraw[red!30!white] (-2.5, 1.5) -- ++ (0.75,0) -- ++ (0,0.7) -- ++ (-1, 0) -- cycle;
    
    \filldraw[blue!30!white] (-3, -2) -- ++ (1, 0) -- ++ (0, 1) -- ++ (-1, 0) -- cycle;
    \filldraw[red!30!white] (-2.25, -1.25) -- ++ (0.25, 0) -- ++ (0, 0.25) -- ++ (-0.25, 0) -- cycle;
    \filldraw[blue!30!white] (-3, -2+1) -- ++ (1, 0) -- ++ (0, 1) -- ++ (-1, 0) -- cycle;
    \filldraw[red!30!white] (-2.25, -0.25) -- ++ (0.25, 0) -- ++ (0, 0.25) -- ++ (-0.25, 0) -- cycle;
    \filldraw[blue!30!white] (-3, -2+2) -- ++ (1, 0) -- ++ (0, 1) -- ++ (-1, 0) -- cycle;
    \filldraw[red!30!white] (-2.25, 0.75) -- ++ (0.25, 0) -- ++ (0, 0.25) -- ++ (-0.25, 0) -- cycle;
    \filldraw[blue!30!white] (-3, -2+3) -- ++ (1, 0) -- ++ (0, 1) -- ++ (-1, 0) -- cycle;

    \filldraw[yellow] (-2.125, -1.85) -- ++ (0.125, 0) -- ++ (0, 0.125) -- ++ (-0.125, 0) -- cycle;
    \filldraw[yellow] (-2.125, -1.55) -- ++ (0.125, 0) -- ++ (0, 0.125) -- ++ (-0.125, 0) -- cycle;

    \filldraw[yellow] (-2.125, -0.85) -- ++ (0.125, 0) -- ++ (0, 0.125) -- ++ (-0.125, 0) -- cycle;
    \filldraw[yellow] (-2.125, -0.55) -- ++ (0.125, 0) -- ++ (0, 0.125) -- ++ (-0.125, 0) -- cycle;

    \filldraw[yellow] (-2.125, 0.15) -- ++ (0.125, 0) -- ++ (0, 0.125) -- ++ (-0.125, 0) -- cycle;
    \filldraw[yellow] (-2.125, 0.45) -- ++ (0.125, 0) -- ++ (0, 0.125) -- ++ (-0.125, 0) -- cycle;

    \filldraw[yellow] (-2.125, 1.15) -- ++ (0.125, 0) -- ++ (0, 0.125) -- ++ (-0.125, 0) -- cycle;
    \filldraw[yellow] (-2.125, 1.45) -- ++ (0.125, 0) -- ++ (0, 0.125) -- ++ (-0.125, 0) -- cycle;
    \end{scope}
    \draw[] (-3, -2) -- ++ (6, 0) -- ++ (0, 4) -- ++ (-6, 0) -- cycle;
        
    \end{scope}

    \begin{scope}[yshift=-6cm, xshift=7cm]
    \node[] () at (0, -2.5) {(e)};
    \begin{scope}
    \clip (-3, -2) -- ++ (6, 0) -- ++ (0, 4) -- ++ (-6, 0) -- cycle;

    \filldraw[red!30!white] (-2, -2) -- ++ (0.5,0) -- ++ (0,0.7) -- ++ (-1, 0) -- cycle;
    \filldraw[red!30!white] (-2, -1) -- ++ (0.5,0) -- ++ (0,0.7) -- ++ (-1, 0) -- cycle;
    \filldraw[red!30!white] (-2, -0) -- ++ (0.5,0) -- ++ (0,0.7) -- ++ (-1, 0) -- cycle;
    \filldraw[red!30!white] (-2, 1) -- ++ (0.5,0) -- ++ (0,0.7) -- ++ (-1, 0) -- cycle;
    \filldraw[red!30!white] (-2, 2) -- ++ (0.5,0) -- ++ (0,0.7) -- ++ (-1, 0) -- cycle;
    
    \filldraw[blue!30!white] (-2.5, -2.5) -- ++ (0.75,0) -- ++ (0,0.7) -- ++ (-1, 0) -- cycle;
    \filldraw[blue!30!white] (-2.5, -1.5) -- ++ (0.75,0) -- ++ (0,0.7) -- ++ (-1, 0) -- cycle;
    \filldraw[blue!30!white] (-2.5, -0.5) -- ++ (0.75,0) -- ++ (0,0.7) -- ++ (-1, 0) -- cycle;
    \filldraw[blue!30!white] (-2.5, 0.5) -- ++ (0.75,0) -- ++ (0,0.7) -- ++ (-1, 0) -- cycle;
    \filldraw[blue!30!white] (-2.5, 1.5) -- ++ (0.75,0) -- ++ (0,0.7) -- ++ (-1, 0) -- cycle;
    
    \filldraw[blue!30!white] (-3, -2) -- ++ (1, 0) -- ++ (0, 1) -- ++ (-1, 0) -- cycle;
    \filldraw[blue!30!white] (-2.25, -1.25) -- ++ (0.25, 0) -- ++ (0, 0.25) -- ++ (-0.25, 0) -- cycle;
    \filldraw[blue!30!white] (-3, -2+1) -- ++ (1, 0) -- ++ (0, 1) -- ++ (-1, 0) -- cycle;
    \filldraw[blue!30!white] (-2.25, -0.25) -- ++ (0.25, 0) -- ++ (0, 0.25) -- ++ (-0.25, 0) -- cycle;
    \filldraw[blue!30!white] (-3, -2+2) -- ++ (1, 0) -- ++ (0, 1) -- ++ (-1, 0) -- cycle;
    \filldraw[blue!30!white] (-2.25, 0.75) -- ++ (0.25, 0) -- ++ (0, 0.25) -- ++ (-0.25, 0) -- cycle;
    \filldraw[blue!30!white] (-3, -2+3) -- ++ (1, 0) -- ++ (0, 1) -- ++ (-1, 0) -- cycle;

    \filldraw[red!30!white] (-2.125, -1.85) -- ++ (0.125, 0) -- ++ (0, 0.125) -- ++ (-0.125, 0) -- cycle;
    \filldraw[red!30!white] (-2.125, -1.55) -- ++ (0.125, 0) -- ++ (0, 0.125) -- ++ (-0.125, 0) -- cycle;

    \filldraw[red!30!white] (-2.125, -0.85) -- ++ (0.125, 0) -- ++ (0, 0.125) -- ++ (-0.125, 0) -- cycle;
    \filldraw[red!30!white] (-2.125, -0.55) -- ++ (0.125, 0) -- ++ (0, 0.125) -- ++ (-0.125, 0) -- cycle;

    \filldraw[red!30!white] (-2.125, 0.15) -- ++ (0.125, 0) -- ++ (0, 0.125) -- ++ (-0.125, 0) -- cycle;
    \filldraw[red!30!white] (-2.125, 0.45) -- ++ (0.125, 0) -- ++ (0, 0.125) -- ++ (-0.125, 0) -- cycle;

    \filldraw[red!30!white] (-2.125, 1.15) -- ++ (0.125, 0) -- ++ (0, 0.125) -- ++ (-0.125, 0) -- cycle;
    \filldraw[red!30!white] (-2.125, 1.45) -- ++ (0.125, 0) -- ++ (0, 0.125) -- ++ (-0.125, 0) -- cycle;

    \filldraw[yellow] (-1.875, -1.35) -- ++ (0.125, 0) -- ++ (0, 0.125) -- ++ (-0.125, 0) -- cycle;
    \filldraw[yellow] (-1.875, -1.05) -- ++ (0.125, 0) -- ++ (0, 0.125) -- ++ (-0.125, 0) -- cycle;

    \filldraw[yellow] (-1.875, -0.35) -- ++ (0.125, 0) -- ++ (0, 0.125) -- ++ (-0.125, 0) -- cycle;
    \filldraw[yellow] (-1.875, -0.05) -- ++ (0.125, 0) -- ++ (0, 0.125) -- ++ (-0.125, 0) -- cycle;

    \filldraw[yellow] (-1.875, 0.65) -- ++ (0.125, 0) -- ++ (0, 0.125) -- ++ (-0.125, 0) -- cycle;
    \filldraw[yellow] (-1.875, 0.95) -- ++ (0.125, 0) -- ++ (0, 0.125) -- ++ (-0.125, 0) -- cycle;

    \filldraw[yellow] (-1.875, 1.65) -- ++ (0.125, 0) -- ++ (0, 0.125) -- ++ (-0.125, 0) -- cycle;
    \end{scope}
    \draw[] (-3, -2) -- ++ (6, 0) -- ++ (0, 4) -- ++ (-6, 0) -- cycle;
        
    \end{scope}

    \begin{scope}[yshift=-6cm, xshift=14cm]
    \node[] () at (0, -2.5) {(f)};
    \begin{scope}
    \clip (-3, -2) -- ++ (6, 0) -- ++ (0, 4) -- ++ (-6, 0) -- cycle;

    \filldraw[red!30!white] (2, -2) -- ++ (0.5,0) -- ++ (0,0.7) -- ++ (-1, 0) -- cycle;
    \filldraw[red!30!white] (2, -1) -- ++ (0.5,0) -- ++ (0,0.7) -- ++ (-1, 0) -- cycle;
    \filldraw[red!30!white] (2, -0) -- ++ (0.5,0) -- ++ (0,0.7) -- ++ (-1, 0) -- cycle;
    \filldraw[red!30!white] (2, 1) -- ++ (0.5,0) -- ++ (0,0.7) -- ++ (-1, 0) -- cycle;
    \filldraw[red!30!white] (2, 2) -- ++ (0.5,0) -- ++ (0,0.7) -- ++ (-1, 0) -- cycle;

    \filldraw[blue!30!white] (-3, -2) --++ (5.25, 0) -- ++ (0, 4) -- ++ (-5.25, 0)--cycle;

    \filldraw[red!30!white] (2, -1.75) -- ++ (0.25, 0) -- ++ (0, 0.2) -- ++ (-0.25, 0) -- cycle;
    \filldraw[red!30!white] (2, -0.75) -- ++ (0.25, 0) -- ++ (0, 0.2) -- ++ (-0.25, 0) -- cycle;
    \filldraw[red!30!white] (2, 0.25) -- ++ (0.25, 0) -- ++ (0, 0.2) -- ++ (-0.25, 0) -- cycle;
    \filldraw[red!30!white] (2, 1.25) -- ++ (0.25, 0) -- ++ (0, 0.2) -- ++ (-0.25, 0) -- cycle;
    
    \filldraw[red!30!white] (1.875, -1.85) -- ++ (0.125, 0) -- ++ (0, 0.125) -- ++ (-0.125, 0) -- cycle;
    \filldraw[red!30!white] (1.875, -1.55) -- ++ (0.125, 0) -- ++ (0, 0.125) -- ++ (-0.125, 0) -- cycle;

    \filldraw[red!30!white] (1.875, -0.85) -- ++ (0.125, 0) -- ++ (0, 0.125) -- ++ (-0.125, 0) -- cycle;
    \filldraw[red!30!white] (1.875, -0.55) -- ++ (0.125, 0) -- ++ (0, 0.125) -- ++ (-0.125, 0) -- cycle;

    \filldraw[red!30!white] (1.875, 0.15) -- ++ (0.125, 0) -- ++ (0, 0.125) -- ++ (-0.125, 0) -- cycle;
    \filldraw[red!30!white] (1.875, 0.45) -- ++ (0.125, 0) -- ++ (0, 0.125) -- ++ (-0.125, 0) -- cycle;

    \filldraw[red!30!white] (1.875, 1.15) -- ++ (0.125, 0) -- ++ (0, 0.125) -- ++ (-0.125, 0) -- cycle;
    \filldraw[red!30!white] (1.875, 1.45) -- ++ (0.125, 0) -- ++ (0, 0.125) -- ++ (-0.125, 0) -- cycle;

    \filldraw[yellow] (2.125, -1.35) -- ++ (0.125, 0) -- ++ (0, 0.125) -- ++ (-0.125, 0) -- cycle;
    \filldraw[yellow] (2.125, -1.05) -- ++ (0.125, 0) -- ++ (0, 0.125) -- ++ (-0.125, 0) -- cycle;

    \filldraw[yellow] (2.125, -0.35) -- ++ (0.125, 0) -- ++ (0, 0.125) -- ++ (-0.125, 0) -- cycle;
    \filldraw[yellow] (2.125, -0.05) -- ++ (0.125, 0) -- ++ (0, 0.125) -- ++ (-0.125, 0) -- cycle;

    \filldraw[yellow] (2.125, 0.65) -- ++ (0.125, 0) -- ++ (0, 0.125) -- ++ (-0.125, 0) -- cycle;
    \filldraw[yellow] (2.125, 0.95) -- ++ (0.125, 0) -- ++ (0, 0.125) -- ++ (-0.125, 0) -- cycle;

    \filldraw[yellow] (2.125, 1.65) -- ++ (0.125, 0) -- ++ (0, 0.125) -- ++ (-0.125, 0) -- cycle;
    \end{scope}
    \draw[] (-3, -2) -- ++ (6, 0) -- ++ (0, 4) -- ++ (-6, 0) -- cycle;
        
    \end{scope}

    \begin{scope}[yshift=-12cm]
    \node[] () at (0, -2.5) {(g)};
    \begin{scope}
    \clip (-3, -2) -- ++ (6, 0) -- ++ (0, 4) -- ++ (-6, 0) -- cycle;
    \filldraw[red!30!white] (3, 2) -- ++ (-1, 0) -- ++ (0, -0.75) -- ++ (1,0) -- cycle;
    
    \filldraw[blue!30!white] (2, -2) -- ++ (0.5,0) -- ++ (0,0.7) -- ++ (-1, 0) -- cycle;
    \filldraw[blue!30!white] (2, -1) -- ++ (0.5,0) -- ++ (0,0.7) -- ++ (-1, 0) -- cycle;
    \filldraw[blue!30!white] (2, -0) -- ++ (0.5,0) -- ++ (0,0.7) -- ++ (-1, 0) -- cycle;
    \filldraw[blue!30!white] (2, 1) -- ++ (0.5,0) -- ++ (0,0.7) -- ++ (-1, 0) -- cycle;
    \filldraw[blue!30!white] (2, 2) -- ++ (0.5,0) -- ++ (0,0.7) -- ++ (-1, 0) -- cycle;

    \filldraw[blue!30!white] (-3, -2) --++ (5.25, 0) -- ++ (0, 4) -- ++ (-5.25, 0)--cycle;

    \filldraw[white] (2.125, -1.35) -- ++ (0.125, 0) -- ++ (0, 0.125) -- ++ (-0.125, 0) -- cycle;
    \filldraw[white] (2.125, -1.05) -- ++ (0.125, 0) -- ++ (0, 0.125) -- ++ (-0.125, 0) -- cycle;

    \filldraw[white] (2.125, -0.35) -- ++ (0.125, 0) -- ++ (0, 0.125) -- ++ (-0.125, 0) -- cycle;
    \filldraw[white] (2.125, -0.05) -- ++ (0.125, 0) -- ++ (0, 0.125) -- ++ (-0.125, 0) -- cycle;

    \filldraw[white] (2.125, 0.65) -- ++ (0.125, 0) -- ++ (0, 0.125) -- ++ (-0.125, 0) -- cycle;
    \filldraw[white] (2.125, 0.95) -- ++ (0.125, 0) -- ++ (0, 0.125) -- ++ (-0.125, 0) -- cycle;
    \filldraw[yellow] (2.375, 1.15) -- ++ (0.125, 0) -- ++ (0, 0.125) -- ++ (-0.125, 0) -- cycle;

    \filldraw[red!30!white] (2.125, 1.65) -- ++ (0.125, 0) -- ++ (0, 0.125) -- ++ (-0.125, 0) -- cycle;
    \end{scope}
    \draw[] (-3, -2) -- ++ (6, 0) -- ++ (0, 4) -- ++ (-6, 0) -- cycle;
        
    \end{scope}
    \begin{scope}[yshift=-12cm, xshift=7cm]
    \node[] () at (0, -2.5) {(h)};
    \begin{scope}
    \clip (-3, -2) -- ++ (6, 0) -- ++ (0, 4) -- ++ (-6, 0) -- cycle;
    \filldraw[blue!30!white] (3, 2) -- ++ (-1, 0) -- ++ (0, -0.75) -- ++ (1,0) -- cycle;
    \filldraw[red!30!white] (3, 1.25) -- ++ (-2, 0) -- ++ (0, -0.9) -- ++ (2,0) -- cycle;
    
    \filldraw[blue!30!white] (2, -2) -- ++ (0.5,0) -- ++ (0,0.7) -- ++ (-1, 0) -- cycle;
    \filldraw[blue!30!white] (2, -1) -- ++ (0.5,0) -- ++ (0,0.7) -- ++ (-1, 0) -- cycle;
    \filldraw[blue!30!white] (2, -0) -- ++ (0.5,0) -- ++ (0,0.7) -- ++ (-1, 0) -- cycle;
    \filldraw[blue!30!white] (2, 1) -- ++ (0.5,0) -- ++ (0,0.7) -- ++ (-1, 0) -- cycle;
    \filldraw[blue!30!white] (2, 2) -- ++ (0.5,0) -- ++ (0,0.7) -- ++ (-1, 0) -- cycle;

    \filldraw[blue!30!white] (-3, -2) --++ (5.25, 0) -- ++ (0, 4) -- ++ (-5.25, 0)--cycle;

    \filldraw[white] (2.125, -1.35) -- ++ (0.125, 0) -- ++ (0, 0.125) -- ++ (-0.125, 0) -- cycle;
    \filldraw[white] (2.125, -1.05) -- ++ (0.125, 0) -- ++ (0, 0.125) -- ++ (-0.125, 0) -- cycle;

    \filldraw[white] (2.125, -0.35) -- ++ (0.125, 0) -- ++ (0, 0.125) -- ++ (-0.125, 0) -- cycle;
    \filldraw[white] (2.125, -0.05) -- ++ (0.125, 0) -- ++ (0, 0.125) -- ++ (-0.125, 0) -- cycle;

    \filldraw[red!30!white] (2.125, 0.65) -- ++ (0.125, 0) -- ++ (0, 0.125) -- ++ (-0.125, 0) -- cycle;
    \filldraw[red!30!white] (2.125, 0.95) -- ++ (0.125, 0) -- ++ (0, 0.125) -- ++ (-0.125, 0) -- cycle;
    \filldraw[red!30!white] (2.375, 1.15) -- ++ (0.125, 0) -- ++ (0, 0.125) -- ++ (-0.125, 0) -- cycle;

    \filldraw[yellow] (2.375, 0.25) -- ++ (0.125, 0) -- ++ (0, 0.125) -- ++ (-0.125, 0) -- cycle;

    \end{scope}
    \draw[] (-3, -2) -- ++ (6, 0) -- ++ (0, 4) -- ++ (-6, 0) -- cycle;
        
    \end{scope}

    \begin{scope}[yshift=-12cm, xshift=14cm]
    \node[] () at (0, -2.5) {(i)};
    \draw[fill=blue!30!white] (-3, -2) -- ++ (6, 0) -- ++ (0, 4) -- ++ (-6, 0) -- cycle;
    \end{scope}
    \end{tikzpicture}
    \caption{Description of the channel that builds up a 2D long-range entangled global state. In (a), we obtain the reduced density matrix on the red region. In (b)-(i), each figure represents an extending map. The red and yellow regions correspond to $C$ and $E$ of the extending map, respectively. The blue region is (together with the yellow region) the support of the reduced density matrix in the previous step, which can be viewed as a subset of the purifying system.}
    \label{fig:lre_buildup}
\end{figure}

\begin{figure}[h]
    \centering
    \begin{tikzpicture}[line width=1pt]
        \begin{scope}[xshift=0.0cm, yshift=0.0cm, scale=0.65]
            \begin{scope}
            \clip (-3, -2) -- ++ (6, 0) -- ++ (0, 4) -- ++ (-6, 0) -- cycle;
            \filldraw[blue!30!white] (-3, -2) -- ++ (1, 0) -- ++ (0, 1) -- ++ (-1, 0) -- cycle;
            \filldraw[yellow] (-2.25, -1.25) -- ++ (0.25, 0) -- ++ (0, 0.25) -- ++ (-0.25, 0) -- cycle;
            \filldraw[red!30!white] (-3, -2+1) -- ++ (1, 0) -- ++ (0, 1) -- ++ (-1, 0) -- cycle;
            \end{scope}
        \draw[] (-3, -2) -- ++ (6, 0) -- ++ (0, 4) -- ++ (-6, 0) -- cycle;
        \draw[style=dashed] (-3, -2) -- ++ (1.2, 0) -- ++ (0, 2.4) -- ++ (-1.2, 0) -- cycle;
        \draw[style=dashed] (-1.8, -2) -- (5.34, -2.4);
        \draw[style=dashed] (-1.8, 0.4) -- (5.34, 1.8);
        \end{scope}
        \begin{scope}[xshift=7.0cm, yshift=0.0cm, scale=0.7]
            \draw[fill=blue!30!white] (-3, -2.5) -- ++ (2, 0) -- ++ (0, 2) -- ++ (-2, 0);
            \draw[fill=red!30!white] (-3, -0.5) -- ++ (2, 0) -- ++ (0, 2) -- ++ (-2, 0);
            \draw[fill=yellow] (-3+1.1, -1.25) -- ++ (0.9, 0) -- ++ (0, 0.75) -- ++ (-0.9, 0) -- cycle;
            \draw[] (-0.25-0.75, -0.8) -- ++ (0.75, 0) -- ++ (0, 2.8) -- (-0.25-2.75, 2.0);
            \draw[line width=1.5pt] (-0.25, -2.5) -- ++ (-2.75, 0) -- ++ (0, 5.0);

            \node[] () at (-3+1.6, -0.9) {\small $E$};
            \node[] () at (-2, -1.75) {$B$};
            \node[] () at (-2, 0.5) {$C$};
            \node[] () at (-0.6, 0.5) {$D$};
            \node[] () at (-0.45, -1.85) {$A$};
        \end{scope}
        \begin{scope}[xshift=0.0cm, yshift=-4.5cm, scale=0.65]
            \begin{scope}
                \clip (-3, -2) -- ++ (6, 0) -- ++ (0, 4) -- ++ (-6, 0) -- cycle;
            
                \filldraw[red!30!white] (-2, -2) -- ++ (0.5,0) -- ++ (0,0.7) -- ++ (-1, 0) -- cycle;
                \filldraw[red!30!white] (-2, -1) -- ++ (0.5,0) -- ++ (0,0.7) -- ++ (-1, 0) -- cycle;
                \filldraw[red!30!white] (-2, -0) -- ++ (0.5,0) -- ++ (0,0.7) -- ++ (-1, 0) -- cycle;
                \filldraw[red!30!white] (-2, 1) -- ++ (0.5,0) -- ++ (0,0.7) -- ++ (-1, 0) -- cycle;
                \filldraw[red!30!white] (-2, 2) -- ++ (0.5,0) -- ++ (0,0.7) -- ++ (-1, 0) -- cycle;
                
                \filldraw[blue!30!white] (-2.5, -2.5) -- ++ (0.75,0) -- ++ (0,0.7) -- ++ (-1, 0) -- cycle;
                \filldraw[blue!30!white] (-2.5, -1.5) -- ++ (0.75,0) -- ++ (0,0.7) -- ++ (-1, 0) -- cycle;
                \filldraw[blue!30!white] (-2.5, -0.5) -- ++ (0.75,0) -- ++ (0,0.7) -- ++ (-1, 0) -- cycle;
                \filldraw[blue!30!white] (-2.5, 0.5) -- ++ (0.75,0) -- ++ (0,0.7) -- ++ (-1, 0) -- cycle;
                \filldraw[blue!30!white] (-2.5, 1.5) -- ++ (0.75,0) -- ++ (0,0.7) -- ++ (-1, 0) -- cycle;
                
                \filldraw[blue!30!white] (-3, -2) -- ++ (1, 0) -- ++ (0, 1) -- ++ (-1, 0) -- cycle;
                \filldraw[blue!30!white] (-2.25, -1.25) -- ++ (0.25, 0) -- ++ (0, 0.25) -- ++ (-0.25, 0) -- cycle;
                \filldraw[blue!30!white] (-3, -2+1) -- ++ (1, 0) -- ++ (0, 1) -- ++ (-1, 0) -- cycle;
                \filldraw[blue!30!white] (-2.25, -0.25) -- ++ (0.25, 0) -- ++ (0, 0.25) -- ++ (-0.25, 0) -- cycle;
                \filldraw[blue!30!white] (-3, -2+2) -- ++ (1, 0) -- ++ (0, 1) -- ++ (-1, 0) -- cycle;
                \filldraw[blue!30!white] (-2.25, 0.75) -- ++ (0.25, 0) -- ++ (0, 0.25) -- ++ (-0.25, 0) -- cycle;
                \filldraw[blue!30!white] (-3, -2+3) -- ++ (1, 0) -- ++ (0, 1) -- ++ (-1, 0) -- cycle;
            
                \filldraw[red!30!white] (-2.125, -1.85) -- ++ (0.125, 0) -- ++ (0, 0.125) -- ++ (-0.125, 0) -- cycle;
                \filldraw[red!30!white] (-2.125, -1.55) -- ++ (0.125, 0) -- ++ (0, 0.125) -- ++ (-0.125, 0) -- cycle;
            
                \filldraw[red!30!white] (-2.125, -0.85) -- ++ (0.125, 0) -- ++ (0, 0.125) -- ++ (-0.125, 0) -- cycle;
                \filldraw[red!30!white] (-2.125, -0.55) -- ++ (0.125, 0) -- ++ (0, 0.125) -- ++ (-0.125, 0) -- cycle;
            
                \filldraw[red!30!white] (-2.125, 0.15) -- ++ (0.125, 0) -- ++ (0, 0.125) -- ++ (-0.125, 0) -- cycle;
                \filldraw[red!30!white] (-2.125, 0.45) -- ++ (0.125, 0) -- ++ (0, 0.125) -- ++ (-0.125, 0) -- cycle;
            
                \filldraw[red!30!white] (-2.125, 1.15) -- ++ (0.125, 0) -- ++ (0, 0.125) -- ++ (-0.125, 0) -- cycle;
                \filldraw[red!30!white] (-2.125, 1.45) -- ++ (0.125, 0) -- ++ (0, 0.125) -- ++ (-0.125, 0) -- cycle;

                \filldraw[yellow] (-1.875, -1.35) -- ++ (0.125, 0) -- ++ (0, 0.125) -- ++ (-0.125, 0) -- cycle;
                \filldraw[yellow] (-1.875, -1.05) -- ++ (0.125, 0) -- ++ (0, 0.125) -- ++ (-0.125, 0) -- cycle;
            
                \filldraw[yellow] (-1.875, -0.35) -- ++ (0.125, 0) -- ++ (0, 0.125) -- ++ (-0.125, 0) -- cycle;
                \filldraw[yellow] (-1.875, -0.05) -- ++ (0.125, 0) -- ++ (0, 0.125) -- ++ (-0.125, 0) -- cycle;
            
                \filldraw[yellow] (-1.875, 0.65) -- ++ (0.125, 0) -- ++ (0, 0.125) -- ++ (-0.125, 0) -- cycle;
                \filldraw[yellow] (-1.875, 0.95) -- ++ (0.125, 0) -- ++ (0, 0.125) -- ++ (-0.125, 0) -- cycle;
            
                \filldraw[yellow] (-1.875, 1.65) -- ++ (0.125, 0) -- ++ (0, 0.125) -- ++ (-0.125, 0) -- cycle;
            \end{scope}
        \draw[] (-3, -2) -- ++ (6, 0) -- ++ (0, 4) -- ++ (-6, 0) -- cycle;
        \draw[style=dashed] (-3+0.75, -2+1-0.15) -- ++ (1.0, 0) -- ++ (0, 1.0) -- ++ (-1.0, 0) -- cycle;
        \draw[style=dashed] (-2+0.75, -2+1-0.15) -- (5.34, -2.4);
        \draw[style=dashed] (-2+0.75, -2+1-0.15+1) -- (5.34, 1.8);
        \end{scope}
        \begin{scope}[xshift=7.0cm, yshift=-4.5cm, scale=0.7]
            \draw[draw=none, fill=blue!30!white] (-3+1.75, -1.25-1.0-0.5) -- (-3+1.75, 1.25-0.5+1.0+0.5) -- (-3+1.75-2.0, 1.25-0.5+1.0+0.5) -- (-3+1.75-2.0, -1.25-1.0-0.5) -- cycle;
            \draw[fill=red!30!white] (-3, -1.25) -- ++ (0.75, 0) -- ++ (0, 0.5) -- ++ (1, 0) -- ++ (0, -1) -- ++ (1, 0) -- ++ (0, 3) -- ++ (-1, 0) -- ++ (0, -1) -- ++ (-1, 0) -- ++ (0, 0.5) -- ++ (-0.75, 0) -- ++ (0, -0.75) -- ++ (0.75, 0) -- ++ (0, -0.5) -- ++ (-0.75, 0) -- cycle;
            \draw[fill=yellow] (-3+1.75, -1.25) -- ++ (0, -1.0) -- ++ (-0.75, 0) -- ++ (0, 1.0) -- cycle;
            \draw[] (-3+1.75, -1.25-1.0) -- ++ (0, -0.5);
            \draw[fill=yellow] (-3+1.75, 1.25-0.5) -- ++ (0, 1.0) -- ++ (-0.75, 0) -- ++ (0, -1.0) -- cycle;
            \draw[] (-3+1.75, 1.25-0.5+1.0) -- ++ (0, 0.5);
            \draw[] (-3+1.75, -1.25-0.75) -- ++ (1.7, 0) -- ++ (0, 3.5) -- ++ (-1.7, 0);

            \node[] () at (-2.8, -2.0) {$B$};
            \node[] () at (-2.8, 1.55) {$B$};
            \node[] () at (-1.1, -0.25) {$C$};
            \node[] () at (0.1, -0.25) {$D$};
            \node[] () at (-1.6, 1.25) {$E$};
            \node[] () at (-1.6, -1.75) {$E$};
            \node[] () at (-0.6, -2.5) {$A$};
        \end{scope}
    \end{tikzpicture}
    \caption{A close-up view of steps in Fig.~\ref{fig:lre_buildup}. Notice that the partitions are equivalent to those in Fig.~\ref{fig:lre_conditions} (thus equivalent to Fig.~\ref{fig:circuit_stability_partitions_1}).}
    \label{fig:lre_stability}
\end{figure}
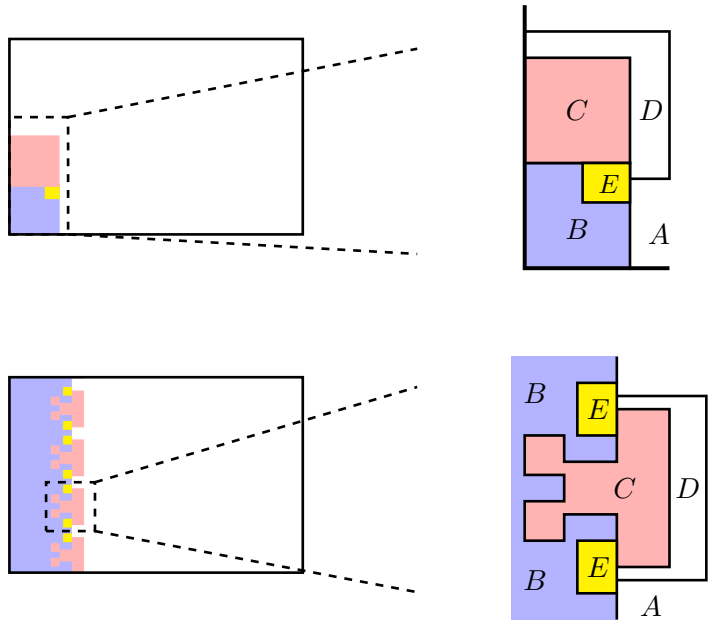

As in Section~\ref{subsec:invertible}, our construction is the composition of a series of extending maps. The whole process is described by Fig.~\ref{fig:lre_buildup} (and a close-up view is given by Fig.~\ref{fig:lre_stability}). In (a)-(c), we are applying a sequence of extending maps associated with Fig.~\ref{fig:lre_conditions}(d). In (d)-(f), we are applying a sequence of extending maps associated with Fig.~\ref{fig:lre_conditions}(a). The steps in (g)-(i) involves the subsystems in Fig.~\ref{fig:lre_conditions} (d)-(e). Each such extending map can be  purified, resulting in a unitary state-preparation circuit.

For a $\ell\times \ell$ system, the depth of the whole process is $O(\ell)$, which is in fact optimal~\cite{haah2016invariant}. Therefore, the state preparation circuit obtained this way has an asymptotically optimal depth. In fact, a more fine-grained statement can be made. Many known long-range entangled states can be prepared by a \emph{sequential circuit}~\cite{kim2017holographic,Chen2024}, which has a depth of at most $c\ell$ for some constant $c$ that can depend on the details of the circuit. Our method allows us to learn a circuit that is not just optimal in $\ell$, but also in $c$ (up to some universal constant). We expand upon this point further below.

Without loss of generality, suppose we have a long-range entangled states that satisfy the constraints in Fig.~\ref{fig:lre_conditions} over $\mathcal{O}(1)$-size disks. Following Fig.~\ref{fig:lre_buildup}, we see that there is a depth-$\mathcal{O}(\ell)$ circuit that prepares the global state. If we apply a depth-$d$ circuit to this state, by the discussion in Section~\ref{subsec:stability_circuit}, Fig.~\ref{fig:lre_conditions} holds only for sufficiently large disks -- the ones whose radius is $\Theta(d)$. Because each such extending map acts on a $e^{\mathcal{O}(d^2)}$-dimensional Hilbert space, its circuit depth can be trivially bounded by $e^{\mathcal{O}(d^2)}$. Following the construction of the state preparation circuit depicted in Fig.~\ref{fig:lre_buildup}, we would obtain a circuit whose depth is $\mathcal{O}(e^{\mathcal{O}(d^2)} \ell)$. 

This naive upper bound of $\mathcal{O}(e^{\mathcal{O}(d^2)} \ell)$ can be upgraded to $\mathcal{O}(d\ell)$. This is because the post-circuit extending map is a composition of some constant-depth circuits and an extending map of the state prior to applying the circuit, over a ball of radius $\Theta(d)$ [Section~\ref{subsec:stability_circuit}]. The latter can be further decomposed into $\mathcal{O}(d)$ layers of commuting channels, using a construction similar to Fig.~\ref{fig:lre_buildup} but only applied to a $\mathcal{O}(d)$-size ball. Altogether, the depth of the extending map becomes $\mathcal{O}(d)$. Therefore, the overall state preparation circuit has a depth of $\mathcal{O}(d\ell)$.

\section{Circuit learning algorithms}\label{sec:learning}

In this Section, we describe our circuit learning algorithms, leveraging the circuit constructions in Section~\ref{sec:extendibility_quantum_phase}. The high-level overview of our algorithms can be summarized as follows [Fig.~\ref{fig:algorithm_overview}]. First, we use local state tomography to learn reduced density matrices of balls whose radii we initially set to $r=1$.  (\textbf{Note}: For our purpose, the approximation error of the tomography should be estimated in terms of the Bures distance, not the more standard trace distance. For this reason, our result relies on Ref.~\cite{Flammia2024quantumchisquared}.) Second, from these reduced density matrices, we learn the extending maps [Section~\ref{subsec:learning_extending_map}]. In particular, as a byproduct of this analysis, we estimate the degree to which each reduced density matrix is extendible [Definition~\ref{definition:locally_extendible_approximate}]. If the combined error of the extending maps is smaller than the desired approximation error [Proposition~\ref{prop:error_bound}], the algorithm outputs a circuit composed of the extending maps. Otherwise, the radius $r$ is incremented and the process repeats.

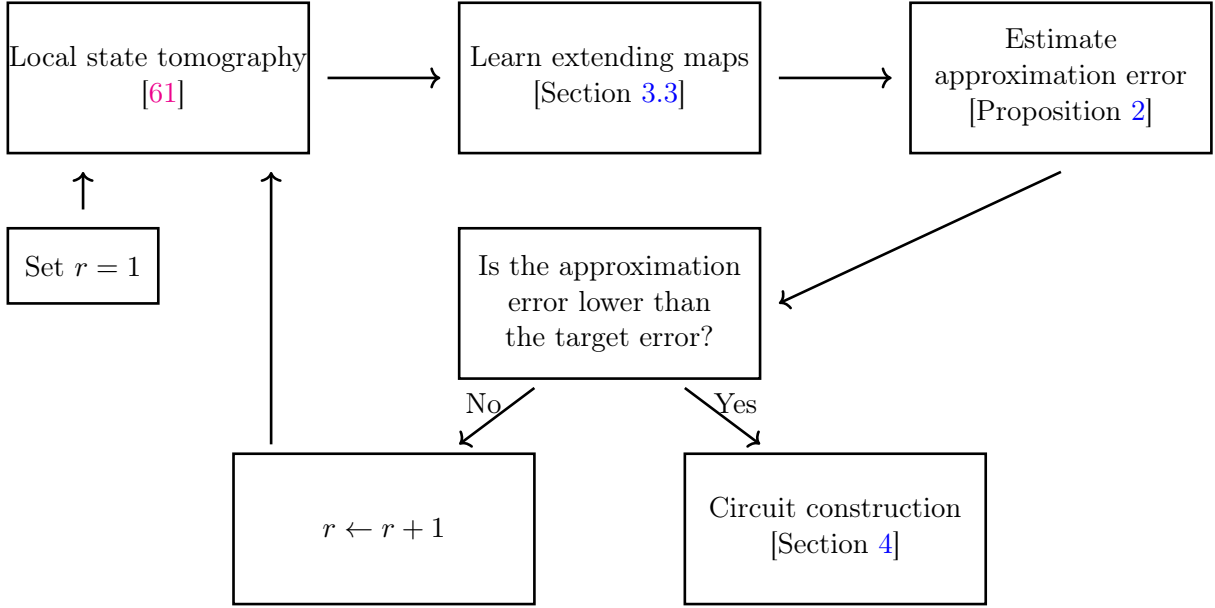
\begin{figure}[h]
    \centering
    \begin{tikzpicture}[line width=1pt, scale=0.92, transform shape]
        \draw[] (-2, -1) -- ++ (4, 0) -- ++ (0, 2) -- ++ (-4, 0) -- cycle;
        \node[align=center] () at (0,0) {Local state tomography\\~\cite{Flammia2024quantumchisquared}};
        \begin{scope}[xshift= 6cm]
            \draw[] (-2, -1) -- ++ (4, 0) -- ++ (0, 2) -- ++ (-4, 0) -- cycle;
        \node[align=center] () at (0,0) {Learn extending maps \\ {[Section~\ref{subsec:learning_extending_map}]}};
        \end{scope}
        \begin{scope}[xshift= 12cm]
            \draw[] (-2, -1) -- ++ (4, 0) -- ++ (0, 2) -- ++ (-4, 0) -- cycle;
        \node[align=center] () at (0,0) {Estimate \\ approximation error \\ {[Proposition~\ref{prop:error_bound}]}};
        \end{scope}

        \draw[->] (12, -1.25) -- (8.25, -3);

        \begin{scope}[xshift= 6cm, yshift=-3cm]
            \draw[] (-2, -1) -- ++ (4, 0) -- ++ (0, 2) -- ++ (-4, 0) -- cycle;
        \node[align=center] () at (0,0) {Is the approximation\\ error lower than\\ the target error?};
        \end{scope}

        \draw[->] (5, -4.125) -- (4, -4.875);
        \draw[->] (7, -4.125) -- (8, -4.875);

        \node[] () at (4.325, -4.325) {No};
        \node[] () at (7.675, -4.325) {Yes};

        \begin{scope}[xshift= 3cm, yshift=-6cm]
            \draw[] (-2, -1) -- ++ (4, 0) -- ++ (0, 2) -- ++ (-4, 0) -- cycle;
        \node[align=center] () at (0,0) {$r\leftarrow r+1$};
        \end{scope}
        \begin{scope}[xshift= 9cm, yshift=-6cm]
            \draw[] (-2, -1) -- ++ (4, 0) -- ++ (0, 2) -- ++ (-4, 0) -- cycle;
        \node[align=center] () at (0,0) {Circuit construction \\ {[Section~\ref{sec:extendibility_quantum_phase}]}};
        \end{scope}

        \begin{scope}[yshift=-3cm]
            \draw[] (-2, 0) -- ++ (2, 0) -- ++ (0, 1) -- ++ (-2, 0) -- cycle;
        \node[align=center] () at (-1,0.5) {Set $r=1$};
        \end{scope}

        \draw[->]  (-1, 1.25-3) -- (-1, -1.25);

        \draw[->] (2.25, 0) -- ++ (1.5,0);
        \draw[->] (8.25, 0) -- ++ (1.5,0);

        \draw[->] (1.5, -4.875) -- (1.5, -1.25);
    \end{tikzpicture}
    \caption{Overview of the circuit learning algorithm. }
    \label{fig:algorithm_overview}
\end{figure}

Because each individual extending map used in the circuits can be inferred from reduced density matrices, if we know the reduced density matrices exactly, we can obtain the the exact circuit. However, in practice the reduced density matrix can be obtained up to some finite precision. Thus the aim of this Section is to quantify those finite-precision effects. 

To facilitate this analysis, we shall use the following convention. Throughout this Section, we shall view the circuits described in Section~\ref{sec:extendibility_quantum_phase} as a composition of the individual extending maps. Recall that each extending map was associated with a region $BCDE$. Because these regions will differ in different iterations, we shall label them as $B_kC_kD_kE_k$, where $k\in \{0, \ldots, N \}$. A purification of these regions shall be labeled as $A_kF_k$, where $A_k$ is chosen so that $A_kB_kE_k$ is the support of the density matrix constructed so far and $F_k$ is an abstract purifying system. From our setup, it is evident that $A_kB_kC_k = A_{k+1}B_{k+1}E_{k+1}$. We shall set $A_0B_0E_0 = \varnothing$, and $A_{N}B_NC_N$ to be the entire system. 

Without loss of generality, suppose we have learned the reduced density matrix $\rho_{B_kC_kD_kE_k}$ up to a precision of $\delta_k$ (in Bures distance). Using the methods in Section~\ref{subsec:learning_extending_map}, it is possible to obtain $2\delta_k$-extending maps, with a complexity at most polynomial in $1/\delta_k$ and $e^{\mathcal{O}(d^2)}$. Thus, our main goal is to bound the trace distance between the global state $\rho$ and the state obtained from a composition of these extending maps. The following proposition establishes this bound.

\begin{proposition}
\label{prop:error_bound}
    Let $\rho$ be a density matrix and $\Phi_{k}: \mathcal{D}_{B_kE_k} \to \mathcal{D}_{B_kC_k}$ be an $\epsilon_k$-extending map for $\rho_{B_kC_kD_kE_k}$.
    \begin{equation}
        \|\rho - \Phi_N\circ \ldots \circ \Phi_0(1) \|_1\leq \sum_{k=0}^{N-1} \epsilon_k.
    \end{equation}
\end{proposition}
\begin{proof}
Let $\tilde{\rho}_{A_kB_kC_k}$ be a state prepared from a composition of $\Phi_0, \ldots, \Phi_k$. Define $\mu_k:= \mathcal{B}(\rho_{A_kB_kC_k}, \tilde{\rho}_{A_kB_kC_k})$. Note the following inequality:
\begin{equation}
\begin{aligned}
\mu_{k+1} &= \bures{\rho_{A_{k+1}B_{k+1}C_{k+1}}}{ \mathcal{I}_{A_{k+1}}\otimes \Phi_{k+1}(\tilde{\rho}_{A_{k+1}B_{k+1}E_{k+1}})} \\
    &\leq 
    \bures{\rho_{A_{k+1}B_{k+1}C_{k+1}}}{ \mathcal{I}_{A_{k+1}}\otimes \Phi_{k+1}(\rho_{A_{k+1}B_{k+1}E_{k+1}})} + \mu_k\\
    &\leq \bures{\rho_{A_{k+1}B_{k+1}C_{k+1}F_{k+1}}}{ \mathcal{I}_{A_{k+1}}\otimes \Phi_{k+1}(\rho_{A_{k+1}B_{k+1}E_{k+1}F_{k+1}})} + \mu_k\\
    &\leq \epsilon_{k+1} + \mu_k.
\end{aligned}
\end{equation}
In the second line, we used the triangle inequality. In the third line, we used the fact that Bures distance is non-increasing under a partial trace. The last line follows from the definition of the extending map. Repeatedly applying this argument for all $k$, and then using the relation between Bures distance and trace distance [Eq.~\eqref{eq:fuchs}], the claim follows.
\end{proof}

In the setups discussed in Section~\ref{sec:extendibility_quantum_phase}, by assumption $\rho_{B_kC_kD_kE_k}$ is extendible from $B_kE_k$ to $B_kC_k$ for all $k$. Moreover, their extending maps can be learned using the method in Lemma~\ref{lemma:extend_from_samples_fast} or Lemma~\ref{lemma:extend_from_samples_slow}. Therefore, by learning these extending maps with small enough error, one can obtain a circuit that prepares the given state. This leads to our main results of this paper [Theorem~\ref{thm:main_invertible} and~\ref{thm:main_lre}].

\mainone*
\begin{proof}
From the discussion in Section~\ref{subsec:invertible}, it follows that the given state can be obtained from $k$ layers of extending maps, each of which can be obtained from a depth-$d$ circuit. To achieve the desired error, it suffices to obtain the extending maps with an error of $\mathcal{O}(\epsilon/n)$ with success probability at least $1-\mathcal{O}(1/n)$ [Proposition~\ref{prop:error_bound}]. 

The extending maps can be obtained from copies of reduced density matrices, using the methods of Lemma~\ref{lemma:extend_from_samples_fast} or~\ref{lemma:extend_from_samples_slow}. Either way, because the density matrices are supported on Hilbert spaces of dimensions of at most $e^{\mathcal{O}(d^k)}$, they can be obtained from $\widetilde{\mathcal{O}}(n^2 \log (n)/\epsilon^2  \cdot e^{\mathcal{O}(d^k)})$ single-copy measurements, using the state tomography method in Ref.~\cite{Flammia2024quantumchisquared}. Because the set of these density matrices can be partitioned into $\Theta(k)$ subsets, each of which consisting of density matrices of disjoint supports, it suffices to have $\widetilde{\mathcal{O}}(n^2 \log (n)/\epsilon^2 \cdot ke^{\mathcal{O}(d^k)})$ copies of the state. 

Depending on whether one uses  Lemma~\ref{lemma:extend_from_samples_fast} or~\ref{lemma:extend_from_samples_slow}, one obtains a different upper bound on the circuit depth of the learned circuit. If one were to use Lemma~\ref{lemma:extend_from_samples_fast}, the depth of each extending map is at most $e^{\mathcal{O}(d^k)}$. If one uses Lemma~\ref{lemma:extend_from_samples_slow}, the depth is at most $\mathcal{O}(d)$, using the standard method of compiling a general unitary into one- and two-qubit gates~\cite{nielsen2010quantum}. The classical running time  follows straightforwardly from  Lemma~\ref{lemma:extend_from_samples_fast} and~\ref{lemma:extend_from_samples_slow}.
\end{proof}

Nearly the same proof applies to long-range entangled states using the circuit described in Section~\ref{subsec:lre_circuit}. 
\begin{theorem}
\label{thm:main_lre}
Consider an unknown $n$-qubit long-range entangled state on a disk with a uniform boundary condition [Section~\ref{subsec:lre}], with the following promise: the state is obtained by applying a depth-$d$ circuit to a state in which the conditions in Fig.~\ref{fig:lre_conditions} are satisfied for $\mathcal{O}(1)$-size disks everywhere. 

There is an algorithm that learns a circuit that prepares the same state up to a trace distance of $\epsilon$ with high probability from $\widetilde{O}(n^2 \log (n)/\epsilon^2 \cdot e^{\mathcal{O}(d^2)})$ samples of the state, with the following complexity:
    \begin{itemize}
        \item One can obtain a circuit of depth $\mathcal{O}(e^{\mathcal{O}(d^2)}\sqrt{n})$ using $\widetilde{\mathcal{O}}(n^3/\epsilon^2)$ classical running time, or
        \item One can obtain a circuit of depth $\mathcal{O}(d\sqrt{n})$ using $\mathcal{O}((n/\epsilon)^{\mathcal{O}(d^2)})$ classical running time.
    \end{itemize}
\end{theorem}
\noindent
Though we focused on the 2D long-range entangled state for concreteness, the same proof technique works in higher dimensions under an appropriately generalized notion of long-range entanglement~\cite{DraftOther}. 

\begin{remark}
Thanks to Proposition~\ref{prop:error_bound}, the circuits obtained from Theorem~\ref{thm:main_invertible} and~\ref{thm:main_lre} are outputs a circuit that prepares the given state correctly, even if the parameter $d$ is not known a priori. 
\end{remark}

\section{Discussion}
\label{sec:discussion}

We have proposed polynomial-time classical algorithms for learning the state preparation circuit from many copies of a given state. Our algorithms are applicable to both invertible states in any number of dimensions and long-range entangled states (specifically, the states obtained from the string-net model~\cite{Levin2005} by applying a constant-depth circuit) in two dimensions. 

Our algorithm has a simple structure, involving tomography of local reduced density matrices followed by a simple classical post-processing [Fig.~\ref{fig:algorithm_overview}]. We thus expect it to be useful for a broad range of applications. In the near term, we may be able to use these algorithms to learn the state preparation circuit directly from many-body quantum states prepared by quantum simulators~\cite{Zhang2017,Ebadi2021,Scholl2021}, using classical shadows~\cite{Huang2020,Hu2023,Tran2023}. Moreover, our method can be used as a flexible black-box method to convert efficient descriptions~\cite{carleo2017solving,Zaletel2020,Vidal2008,becca2017quantum} of quantum states to quantum circuits.

Looking ahead, there are several natural generalizations of our work that warrants a further study. First, while we have considered only constant-depth quantum circuits applied to product states or topologically ordered states~\cite{Levin2005}, in the studies of many-body quantum system, it is more natural to consider a quasi-local unitary~\cite{bachmann2012automorphic}. Such an extension should be possible by using the Lieb-Robinson bound~\cite{hastings2010localityquantumsystems}. A recent work showed that the quasi-local unitary can be approximated by a constant-depth circuit, with each gate acting on poly-logarithmic number of qubits in the system size~\cite{coser2019classification}. Thus our argument should apply after coarse-graining the sites appropriately. Second, going to dimensions three or higher, we remark that there are gapped phases of matter that do not necessarily satisfy the conditions that were put forward in Ref.~\cite{DraftOther}, such as fractons~\cite{Haah2011}. Restricting attention to higher-dimensional ``liquid'' topological phases  that satisfy the higher-dimensional conditions of Ref.~\cite{DraftOther}, we expect our techniques extend straightforwardly.  However, beyond these liquid phases, the existence of a circuit-learning method remains an open problem. Third, we remark that the form of the circuit we used for long-range entangled states [Fig.~\ref{fig:lre_buildup}] is in fact a sequential circuit~\cite{Chen2024}. Thus a natural question is whether the state preparation circuit of  any state prepared by a sequential circuit can be also efficiently learned. Lastly, we remark that the overall structure of our state preparation circuit for invertible states [Fig.~\ref{fig:sre_build_up}] has many similarities with the circuit used for Gibbs state preparation~\cite{brandao2019finite}. Thus it is natural to ask whether our method can be extended to the problem of learning state preparation circuit for Gibbs states. We also note that adaptive circuits, which use mid-circuit measurements and feedforward, can prepare certain states with a lower depth than unitary circuits~\cite{Piroli2021,Tantivasadakarn2023,Smith2024}. Our analysis does not apply to such circuits, since it relies on the light-cone structure of unitary circuits.

At a more fundamental level, it would be desirable to have a better information-theoretic understanding of our notion of extendibility, which was originally put forward in Ref.~\cite{DraftOther}. This is a relaxation of the more well-studied quantum Markov chain condition, which has played a fundamental role in quantum information theory~\cite{Petz1988,Fawzi2015}. However, unlike the quantum Markov chain condition, our notion of extendibility remains largely unexplored. For instance, to the best of our knowledge, there is no simple entropic characterization of extendible states, even though there is one known for quantum Markov chains~\cite{Petz1988}. Moreover, it is not even clear if there should be a generalization of Fawzi-Renner inequality~\cite{Fawzi2015}. These are the questions that we leave for future work.

\vspace{0.5cm}

\noindent 
\textbf{Related work}
Recently, Liu and Landau also came up with an efficient algorithm for learning state preparation circuits for states prepared by constant-depth circuits, using a different method~\cite{landau2024learningquantumstatesprepared}.

\section*{Author contributions}
The authors used Claude Opus 5.5 (Anthropic) for editing the manuscript and correcting errors. All the ideas in this work are due to the authors, who take full responsibility for its content.

\section*{Acknowledgement}
We thank Yunchao Liu and Zeph Landau for helpful discussions, as well as coordination of our submissions. We thank Alexei Kitaev for the early discussions that  ultimately led to this work. We thank David P\'erez-Garc\'ia, Amanda Young, and Soonwon Choi for helpful discussions. IK and HK acknowledge support from NSF under award number PHY-2337931. IK also acknowledges support from the Alfred P. Sloan Foundation through a Sloan Research Fellowship. DR acknowledges support by the Simons Foundation under grant 376205.

\bibliographystyle{quantum}
\bibliography{ref}

\appendix

\section{Stability of local extendibility: a simple case}
\label{appendix:stability_simple}

In this Section, we prove Theorem~\ref{theorem:circuit_stability_2}, a simple case of Theorem~\ref{theorem:circuit_stability_1} for the partition in Fig.~\ref{fig:sre_extendibility}(b). Theorem~\ref{theorem:circuit_stability_2} is simply a concrete version of the discussion in Section~\ref{subsec:invertible} concerning the partition $\Lambda = ABC$ of Fig.~\ref{fig:sre_extendibility}(b), equivalent to the partition in Fig.~\ref{fig:circuit_stability_partitions_2}. Our goal here is to prove that if the pre-circuit reduced density matrix on $BC$ is extendible from $B$ to $BC$, then the post-circuit reduced density matrix on $B'C'$ is extendible from $B'$ to $B'C'$. 

Notice that the partition in Fig.~\ref{fig:circuit_stability_partitions_2} is just a variant of Fig.~\ref{fig:circuit_stability_partitions_1} in which: (i) $E$ and $E'$ are empty sets, and (ii) $BD$ and $B'D'$ are merged into singular $B$ and $B'$ respectively. Consequently, the proof of Theorem~\ref{theorem:circuit_stability_2} is very similar to the proof of Theorem~\ref{theorem:circuit_stability_1}.

\begin{figure}[h]
    \centering
    \begin{tikzpicture}[line width=1pt, scale=1.0]
        \begin{scope}[xshift=0cm]
        \draw[] (-2, -1.5) -- ++ (4, 0) -- ++ (0, 3) -- ++ (-4, 0) -- cycle;
        \draw[] (0,0) circle (0.5cm);
        \draw[] (0,0) circle (1.0cm);

        \node[] () at (0,0) {$C$};
        \node[] () at (0.75cm,0) {$B$};
        \node[] () at (-1.5, 0) {$A$};
        \node[] () at (0, -2.0cm) {(a)};
        \end{scope}
        \begin{scope}[xshift=4.5cm]
        \draw[] (-2, -1.5) -- ++ (4, 0) -- ++ (0, 3) -- ++ (-4, 0) -- cycle;
        \draw[] (0,0) circle (0.4cm);
        \draw[style=dashed] (0,0) circle (0.5cm);
        \draw[] (0,0) circle (1.1cm);
        \draw[style=dashed] (0,0) circle (1.0cm);

        \node[] () at (0,0) {$C'$};
        \node[] () at (0.725cm,0) {$B'$};
        \node[] () at (-1.5, 0) {$A'$};
        \node[] () at (0, -2.0cm) {(b)};
        \end{scope}
    \end{tikzpicture}
    \caption{(a) The pre-circuit partition $ABC$. $A$ is a purifying system of $BC$. \\ (b) The topologically equivalent post-circuit partition $A'B'C'$. The individual subsystems are: $B' := B(d)$, $C' := C(-d)$, and $A'$ is a purifying system of $B'C'$. (Recall notation in Section~\ref{subsec:notation}.) For comparison, the partition $ABCDE$ is overlaid with dashed lines.}
    \label{fig:circuit_stability_partitions_2}
\end{figure}
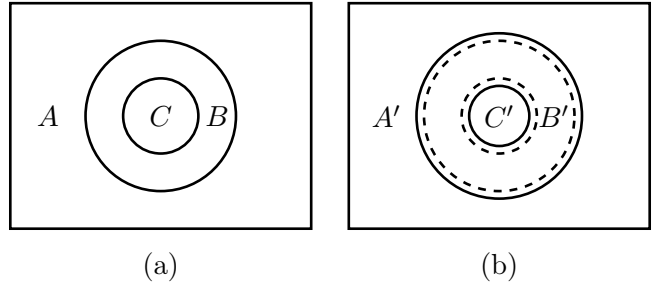

\begin{theorem}\label{theorem:circuit_stability_2}
    Let us partition a system $\Lambda$ as in Fig.~\ref{fig:circuit_stability_partitions_2}(a)-(b) for $\Lambda = ABC$, and as in Fig.~\ref{fig:circuit_stability_partitions_2}(b) for $\Lambda = A'B'C'$. Let $\rho$ be a state on $\Lambda$, $U$ be a geometrically local depth-$d$ circuit supported on $\Lambda$, and $\rho' := U \rho U^{\dagger}$ be the post-circuit state. Suppose $\operatorname{dist}(A, C) > d$. If $\rho_{BC}$ is locally extendible from $B$ to $BC$, then $\rho_{B'C'}'$ is locally extendible from $B'$ to $B'C'$.
\end{theorem}

\begin{proof}
The proof is essentially identical to the proof of Theorem~\ref{theorem:circuit_stability_1} -- in three steps, we construct a channel $\Gamma:\mathcal{D}_{B'}\rightarrow \mathcal{D}_{B'C'}$ that satisfies:
\begin{equation}\label{eq:circuit_stability_2}
    \rho' = \mathcal{I}_{A'} \otimes \Gamma(\rho_{A'B'}').
\end{equation}
\noindent The only distinction of this theorem from Theorem~\ref{theorem:circuit_stability_1} is that $E$ and $E'$ are empty sets, and that $BD$ and $B'D'$ have been merged into singular $B$ and $B'$ respectively. Therefore, our proof is merely a simpler version of the proof of Theorem~\ref{theorem:circuit_stability_1}. Corresponding graphic illustrations are provided in Figs.~\ref{fig:circuit_stability_2_1}-\ref{fig:circuit_stability_2_3}.

\begin{figure}[h]
    \centering
    \begin{tikzpicture}[line width=1pt, scale=0.9]
        \begin{scope}[xshift=0cm]
        \draw[fill=yellow] (-2, -1.5) -- ++ (4, 0) -- ++ (0, 3) -- ++ (-4, 0) -- cycle;
        \draw[] (0,0) circle (0.4cm);
        \draw[style=dashed] (0,0) circle (0.5cm);
        \draw[] (0,0) circle (1.1cm);
        \draw[style=dashed] (0,0) circle (1.0cm);

        \node[] () at (0,0) {$C'$};
        \node[] () at (0.725cm,0) {$B'$};
        \node[] () at (-1.5, 0) {$A'$};
        \node[] () at (0, -2.0cm) {(a)};
        \end{scope}
        \begin{scope}[xshift=4.5cm]
        \draw[fill=yellow] (-2, -1.5) -- ++ (4, 0) -- ++ (0, 3) -- ++ (-4, 0) -- cycle;
        \draw[fill=green] (0,0) circle (1.1cm);
        \draw[style=dashed] (0,0) circle (1.0cm);
        \draw[style=dashed] (0,0) circle (0.5cm);
        \draw[fill=yellow] (0,0) circle (0.4cm);

        \node[] () at (0,0) {$C'$};
        \node[] () at (0.725cm,0) {$B'$};
        \node[] () at (-1.5, 0) {$A'$};
        \node[] () at (0, -2.0cm) {(b)};
        \end{scope}
        \begin{scope}[xshift=9.0cm]
        \draw[fill=yellow] (-2, -1.5) -- ++ (4, 0) -- ++ (0, 3) -- ++ (-4, 0) -- cycle;
        \draw[draw=none, fill=red!60!white] (0,0) circle (1.1cm);
        \draw[fill=white] (0,0) circle (1.0cm);
        \draw[fill=red!60!white] (0,0) circle (0.5cm);
        \draw[draw=none, fill=yellow] (0,0) circle (0.4cm);

        \node[] () at (0,0) {$C$};
        \node[] () at (0.75cm,0) {$B$};
        \node[] () at (-1.5, 0) {$A$};
        \node[] () at (0, -2.0cm) {(c)};
        \end{scope}
    \end{tikzpicture}
    \caption{The first step $\Gamma_1$ for Theorem~\ref{theorem:circuit_stability_2} illustrated. (a) $\rho := U \rho U^{\dagger}$ over $A'B'C' = ABC$. The yellow shading depicted application of the gates in $U$. (b) The unitary $F_{B}$, depicted in green. Note $B'$ as its support. (c) $F_{B}^{\dagger}$ applied, resulting in state $V_{AC} \rho V_{AC}^{\dagger}$. The red shading indicates the gates supported on $B'$ not cancelled out by $F_{B}^{\dagger}$.}
    \label{fig:circuit_stability_2_1}
\end{figure}
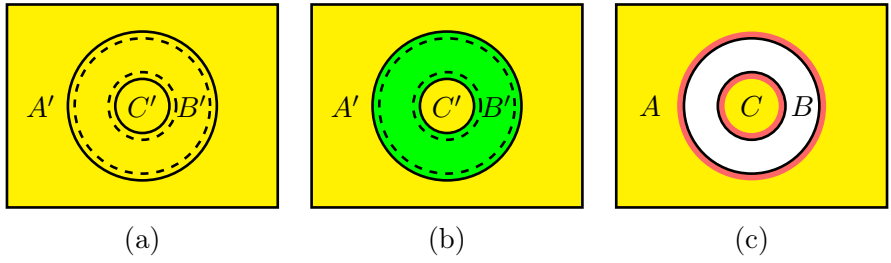

As our first step, we apply $F_B^{\dagger}$. By Lemma~\ref{lemma:light_cone}(iii), $U = F_B V_{AC}$, where $F_B$ is supported on $B(d) = B'$ and $V_{AC}$ is supported on $AC$. Since $F_{B}$ is supported in the complement of $A'C'$, our first step amounts to:
\begin{equation}\label{eq:end_of_step_1-2}
\begin{aligned}
        \mathcal{I}_{A'} \otimes \Gamma_1 \left( \rho_{A'B'}' \right) &= F_{B}^{\dagger} \rho_{A'B'}' F_{B} \\
            &= \Tr_{C'} \left( V_{AC} \rho V_{AC}^{\dagger} \right).
\end{aligned}
\end{equation}

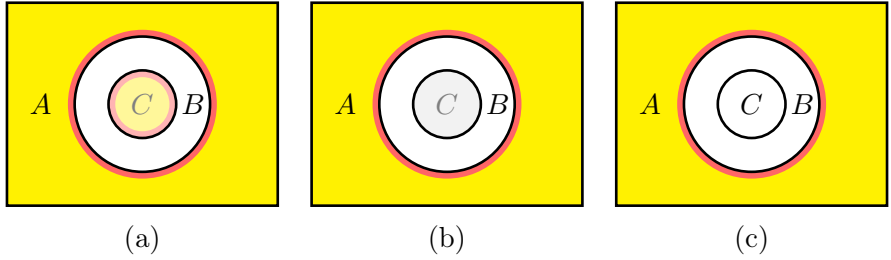
\begin{figure}[h]
    \centering
    \begin{tikzpicture}[line width=1pt, scale=0.9]
        \begin{scope}[xshift=0cm]
        \draw[fill=yellow] (-2, -1.5) -- ++ (4, 0) -- ++ (0, 3) -- ++ (-4, 0) -- cycle;
        \draw[draw=none, fill=red!60!white] (0,0) circle (1.1cm);
        \draw[fill=white] (0,0) circle (1.0cm);
        \draw[fill=red!30!white] (0,0) circle (0.5cm);
        \draw[draw=none, fill=yellow!50!white] (0,0) circle (0.4cm);

        \node[text=gray] () at (0,0) {$C$};
        \node[] () at (0.75cm,0) {$B$};
        \node[] () at (-1.5, 0) {$A$};
        \node[] () at (0, -2.0cm) {(a)};
        \end{scope}
        \begin{scope}[xshift=4.5cm]
        \draw[fill=yellow] (-2, -1.5) -- ++ (4, 0) -- ++ (0, 3) -- ++ (-4, 0) -- cycle;
        \draw[draw=none, fill=red!60!white] (0,0) circle (1.1cm);
        \draw[fill=white] (0,0) circle (1.0cm);
        \draw[fill=gray!10!white] (0,0) circle (0.5cm);
        \draw[draw=none, fill=gray!10!white] (0,0) circle (0.4cm);

        \node[text=gray] () at (0,0) {$C$};
        \node[] () at (0.75cm,0) {$B$};
        \node[] () at (-1.5, 0) {$A$};
        \node[] () at (0, -2.0cm) {(b)};
        \end{scope}
        \begin{scope}[xshift=9.0cm]
        \draw[fill=yellow] (-2, -1.5) -- ++ (4, 0) -- ++ (0, 3) -- ++ (-4, 0) -- cycle;
        \draw[draw=none, fill=red!60!white] (0,0) circle (1.1cm);
        \draw[fill=white] (0,0) circle (1.0cm);
        \draw[fill=white] (0,0) circle (0.5cm);

        \node[] () at (0,0) {$C$};
        \node[] () at (0.75cm,0) {$B$};
        \node[] () at (-1.5, 0) {$A$};
        \node[] () at (0, -2.0cm) {(c)};
        \end{scope}
    \end{tikzpicture}
    \caption{The second step $\Gamma_2$ for Theorem~\ref{theorem:circuit_stability_2} illustrated. (a) The partial trace $\Tr_{C \setminus C'} \circ \Tr_{C'} = \Tr_{C}$ applied; partial traces are depicted with graying the traced-out subsystems. (b) By the cyclic property of the partial trace, the reduced density matrix is indistinguishable from $\rho_{BC}$. (c) The pre-circuit $\rho_{BC}$ is extendible from $B$ to $BC$, so we apply its extending channel $\Phi$. Doing so, we obtain the state $V_{A} \rho V_{A}^{\dagger}$.}
    \label{fig:circuit_stability_2_2}
\end{figure}

Our second step consists of two operations. First, we apply partial trace $\Tr_{C \setminus C'}$, which combines with the partial trace in Eq.~\eqref{eq:end_of_step_1-2} into $\Tr_{C}$. By Lemma~\ref{lemma:light_cone}(i), $V_{AC} = V_C V_A$, where $V_A$ is the past light-cone of $A$ under $V_{AC}$ and $V_C$ is supported on $C$. Since $\operatorname{dist}(A,C) > d$, $V_A$ is supported on $A$. Under $\Tr_{C}$, the reduced density matrix on $AB$ thus becomes $V_A \rho_{AB} V_A^{\dagger}$. Thus, as our second step, we apply the extending channel $\Phi:\mathcal{D}_{B} \rightarrow \mathcal{D}_{BC}$ of $\rho_{BC}$:
\begin{equation}\label{eq:end_of_step_2-2}
\begin{aligned}
        \mathcal{I}_{A'} \otimes \left( \Gamma_2 \circ \Gamma_1 \right) \left( \rho_{A'B'}' \right) 
            &= \mathcal{I}_{A} \otimes \Phi \Bigl[ \Tr_{C \setminus C'} \left( \Tr_{C'} \left( V_{AC} \rho V_{AC}^{\dagger} \right) \right) \Bigr] \\
            &= V_{A} \Phi \left(\rho_{AB}\right) V_{A}^{\dagger} \\
            &= V_{A} \rho V_{A}^{\dagger},
\end{aligned}
\end{equation}
\noindent where $V_{A}$ is a unitary supported on $A$.

\begin{figure}[h]
    \centering
    \begin{tikzpicture}[line width=1pt, scale=0.9]
        \begin{scope}[xshift=0cm]
        \draw[fill=yellow] (-2, -1.5) -- ++ (4, 0) -- ++ (0, 3) -- ++ (-4, 0) -- cycle;
        \draw[fill=red!60!white] (0,0) circle (1.1cm);
        \draw[style=dashed, fill=white] (0,0) circle (1.0cm);
        \draw[style=dashed, fill=white] (0,0) circle (0.5cm);
        \draw[] (0,0) circle (0.4cm);

        \node[] () at (0,0) {$C'$};
        \node[] () at (0.725cm,0) {$B'$};
        \node[] () at (-1.5, 0) {$A'$};
        \node[] () at (0, -2.0cm) {(a)};
        \end{scope}
        \begin{scope}[xshift=4.5cm]
        \draw[fill=yellow] (-2, -1.5) -- ++ (4, 0) -- ++ (0, 3) -- ++ (-4, 0) -- cycle;
        \draw[fill=green] (0,0) circle (1.1cm);
        \draw[style=dashed] (0,0) circle (1.0cm);
        \draw[style=dashed] (0,0) circle (0.5cm);
        \draw[] (0,0) circle (0.4cm);

        \node[] () at (0,0) {$C'$};
        \node[] () at (0.725cm,0) {$B'$};
        \node[] () at (-1.5, 0) {$A'$};
        \node[] () at (0, -2.0cm) {(b)};
        \end{scope}
        \begin{scope}[xshift=9.0cm]
        \draw[fill=yellow] (-2, -1.5) -- ++ (4, 0) -- ++ (0, 3) -- ++ (-4, 0) -- cycle;
        \draw[] (0,0) circle (0.4cm);
        \draw[style=dashed] (0,0) circle (0.5cm);
        \draw[] (0,0) circle (1.1cm);
        \draw[style=dashed] (0,0) circle (1.0cm);

        \node[] () at (0,0) {$C'$};
        \node[] () at (0.725cm,0) {$B'$};
        \node[] () at (-1.5, 0) {$A'$};
        \node[] () at (0, -2.0cm) {(c)};
        \end{scope}
    \end{tikzpicture}
    \caption{The third step $\Gamma_3$ for Theorem~\ref{theorem:circuit_stability_2} illustrated. (a) $V_{A} \rho V_{A}^{\dagger}$ over partition $A'B'C'$. (b) The unitary $F_B V_C$, depicted in green. (c) $F_B V_C$ applied. This is exactly the post-circuit state $\rho'$.}
    \label{fig:circuit_stability_2_3}
\end{figure}
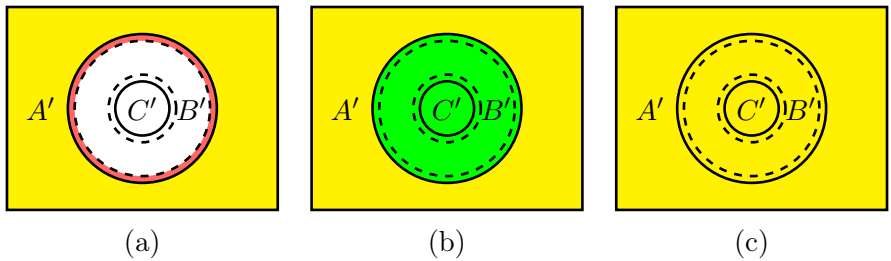

For our final step, notice that $U = F_B V_C V_A$, where $F_B V_C$ is supported on $B'C'$. Therefore, we simply apply $F_B V_C$ to our state and obtain our desired $\rho'$, completing the proof:
\begin{equation}
\begin{aligned}
        \mathcal{I}_{A'} \otimes \left( \Gamma_3 \circ \Gamma_2 \circ \Gamma_1 \right) \left( \rho_{A'B'}' \right) 
            &= (F_B V_C) V_{A} \rho V_{A}^{\dagger} (F_B V_C)^{\dagger} \\
            &= \rho'.
\end{aligned}
\end{equation}

\end{proof}

\section{Proof details for Theorem~\ref{theorem:circuit_stability_3}}\label{appendix:circuit_stability_3}

In this Section, we provide more details about the proof of Theorem~\ref{theorem:circuit_stability_3}. The major steps of the proof are described schematically in Fig.~\ref{fig:circuit_stability_3}. The proof follows three major steps, similar to the proof of Theorem~\ref{theorem:circuit_stability_1}. (For convenience, Theorem~\ref{theorem:circuit_stability_3} is restated below, and its relevant partitions are redrawn in Fig.~\ref{fig:circuit_stability_partitions_3_redrawn}.)

\begin{figure}[h]
\centering
\begin{tikzpicture}[line width=1pt]
    \begin{scope}[xshift=0.0cm]
    \draw[] (-2, -1.5) -- ++ (4, 0) -- ++ (0, 3) -- ++ (-4, 0) -- cycle;
    \draw[] (0,0) circle (0.5cm);
    \draw[] (0,0) circle (1cm);
    \draw[] (30:0.5cm) -- (30:1cm);
    \draw[] (150:0.5cm) -- (150:1cm);
    \draw[] (-30:0.5cm) -- (-30:1cm);
    \draw[] (-150:0.5cm) -- (-150:1cm);
    \draw[] (60:0.5cm) -- (60:1cm);
    \draw[] (120:0.5cm) -- (120:1cm);
    \draw[] (-60:0.5cm) -- (-60:1cm);
    \draw[] (-120:0.5cm) -- (-120:1cm);
    
    \node[] () at (90:0.75cm) {\small $D$};
    \node[] () at (0:0.75cm) {\small $B$};
    \node[] () at (0,0) {\small $C$};
    \node[] () at (180:0.75cm) {\small $B$};
    \node[] () at (270:0.75cm) {\small $D$};
    \node[] () at (45:0.75cm) {\small $E$};
    \node[] () at (-45:0.75cm) {\small $E$};
    \node[] () at (135:0.75cm) {\small $E$};
    \node[] () at (-135:0.75cm) {\small $E$};
    \node[] () at (-1.5, 0) {\small $A$};
    \node[below] () at (0, -1.5) {(a)};
    \end{scope}
    \begin{scope}[xshift=4.5cm]
    \draw[] (-2, -1.5) -- ++ (4, 0) -- ++ (0, 3) -- ++ (-4, 0) -- cycle;
    \draw[] (0,0) circle (0.5cm);
    \draw[] (0,0) circle (1cm);
    \draw[] (30:0.5cm) -- (30:1cm);
    \draw[] (150:0.5cm) -- (150:1cm);
    \draw[] (-30:0.5cm) -- (-30:1cm);
    \draw[] (-150:0.5cm) -- (-150:1cm);
    \draw[] (60:0.5cm) -- (60:1cm);
    \draw[] (120:0.5cm) -- (120:1cm);
    \draw[] (-60:0.5cm) -- (-60:1cm);
    \draw[] (-120:0.5cm) -- (-120:1cm);
    \draw[style=dotted] (60:0.9cm) arc(60:120:0.9cm);
    \draw[style=dotted] (-60:0.9cm) arc(-60:-120:0.9cm);
    
    \node[] () at (90:0.7cm) {\small $\widetilde{D}$};
    \node[] () at (0:0.75cm) {\small $B$};
    \node[] () at (0,0) {\small $C$};
    \node[] () at (180:0.75cm) {\small $B$};
    \node[] () at (270:0.7cm) {\small $\widetilde{D}$};
    \node[] () at (45:0.75cm) {\small $E$};
    \node[] () at (-45:0.75cm) {\small $E$};
    \node[] () at (135:0.75cm) {\small $E$};
    \node[] () at (-135:0.75cm) {\small $E$};
    \node[] () at (-1.5, 0) {\small $A$};
    \node[below] () at (0, -1.5) {(b)};
    \end{scope}
    \begin{scope}[xshift=9.0cm]
    \draw[] (-2, -1.5) -- ++ (4, 0) -- ++ (0, 3) -- ++ (-4, 0) -- cycle;
    \draw[style=dashed] (0,0) circle (0.5cm);
    \draw[] (0,0) circle (0.4cm);
    \draw[style=dashed] (0,0) circle (1cm);
    \draw[] (18.69:0.4cm) -- (24.29:1.1cm) arc(24.29:65.71:1.1cm) -- (71.31:0.4cm);
    \draw[] (-18.69:0.4cm) -- (-24.29:1.1cm) arc(-24.29:-65.71:1.1cm) -- (-71.31:0.4cm);
    \draw[] (180-18.69:0.4cm) -- (180-24.29:1.1cm) arc(180-24.29:180-65.71:1.1cm) -- (180-71.31:0.4cm);
    \draw[] (180+18.69:0.4cm) -- (180+24.29:1.1cm) arc(180+24.29:180+65.71:1.1cm) -- (180+71.31:0.4cm);
    \draw[] (24.29:1.1cm) arc(24.29:-24.29:1.1cm);
    \draw[] (180-24.29:1.1cm) arc(180-24.29:180+24.29:1.1cm);
    \draw[] (67:1.0cm) arc(67:180-67:1.0cm);
    \draw[] (-67:1.0cm) arc(-67:-180+67:1.0cm);
    \draw[style=dashed] (30:0.5cm) -- (30:1cm);
    \draw[style=dashed] (150:0.5cm) -- (150:1cm);
    \draw[style=dashed] (-30:0.5cm) -- (-30:1cm);
    \draw[style=dashed] (-150:0.5cm) -- (-150:1cm);
    \draw[style=dashed] (60:0.5cm) -- (60:1cm);
    \draw[style=dashed] (120:0.5cm) -- (120:1cm);
    \draw[style=dashed] (-60:0.5cm) -- (-60:1cm);
    \draw[style=dashed] (-120:0.5cm) -- (-120:1cm);
    
    \node[] () at (90:0.75cm) {\small $D'$};
    \node[] () at (0:0.75cm) {\small $B'$};
    \node[] () at (0,0) {\small $C'$};
    \node[] () at (180:0.75cm) {\small $B'$};
    \node[] () at (270:0.75cm) {\small $D'$};
    \node[] () at (45:0.75cm) {\small $E'$};
    \node[] () at (-45:0.75cm) {\small $E'$};
    \node[] () at (135:0.75cm) {\small $E'$};
    \node[] () at (-135:0.75cm) {\small $E'$};
    \node[] () at (-1.5, 0) {\small $A'$};
    \node[below] () at (0, -1.5) {(c)};
    \end{scope}
\end{tikzpicture}
\caption{Fig.~\ref{fig:circuit_stability_partitions_3} redrawn. (a) The pre-circuit partition $ABCDE$, equivalent to the one in Fig.~\ref{fig:extendible_state_examples}(d). $A$ is a purifying system of $BCDE$. (b) The same partition $ABCDE$, with subsystem $\widetilde{D} \subset D$ shown. $\widetilde{D}$ consists of all sites in $D$ that are at distance more than $d$ from $A$. (c) The topologically equivalent post-circuit partition $A'B'C'D'E'$. The individual subsystems are: $B' := BE(d) \setminus E(d)$, $C' := C(-d)$, $D' := BCDE \setminus BE(d)C(-d)$, $E' := E(d)$, and $A'$ is a purifying system of $B'C'D'E'$. (Recall notation in Section~\ref{subsec:notation}.) For comparison, the partition $ABCDE$ is overlaid with dashed lines.}
\label{fig:circuit_stability_partitions_3_redrawn}
\end{figure}

\begin{theorem*}
    Theorem~\ref{theorem:circuit_stability_3} restated: Let us partition a system $\Lambda$ as in Fig.~\ref{fig:circuit_stability_partitions_3}(a)-(b) for $\Lambda = ABCDE$, and as in Fig.~\ref{fig:circuit_stability_partitions_3}(c) for $\Lambda = A'B'C'D'E'$. Let $\rho$ be a state on $\Lambda$, $U$ be a geometrically local depth-$d$ circuit supported on $\Lambda$, and $\rho' := U \rho U^{\dagger}$ be the post-circuit state. Suppose $\operatorname{dist}(A, C) > d$ and $\operatorname{dist}(B, D) > 2d$. If $\rho_{BC\widetilde{D}E}$ is locally extendible from $BE$ to $BC$, then $\rho_{B'C'D'E'}'$ is locally extendible from $B'E'$ to $B'C'$.
\end{theorem*}

The first step is depicted by Fig.~\ref{fig:circuit_stability_3}(a)-(c). Fig.~\ref{fig:circuit_stability_3}(a) shows $\rho' := U \rho U^{\dagger}$ over $A'B'C'D'E'$. The yellow shading depicts application of gates in $U$. Fig.~\ref{fig:circuit_stability_3}(b) shows the unitary $F_{BE}$, depicted in green. Note $B'E'$ as its support. Fig.~\ref{fig:circuit_stability_3}(c) is the state with $F_{BE}^{\dagger}$ applied, resulting in state $V_{ACD} \rho V_{ACD}^{\dagger}$, where $V_{ACD}$ is a unitary supported in $ACD$. The red shading indicates the gates supported on $B'E'$ not cancelled out by $F_{BE}^{\dagger}$. 

The second step is shown by Fig.~\ref{fig:circuit_stability_3}(d)-(f). Fig.~\ref{fig:circuit_stability_3}(d) is the partial trace $\Tr_{CD \setminus C'D'}$ applied, which combines with $\Tr_{C'D'}$ in $\rho_{A'B'E'}'$ into $\Tr_{CD \setminus C'D'} \circ \Tr_{C'D'} = \Tr_{CD}$; partial traces are depicted with graying the traced-out subsystems. Fig.~\ref{fig:circuit_stability_3}(e) shows that, under the partial trace, the reduced density matrix on $BC\widetilde{D}E$ is indistinguishable from $\rho_{BC\widetilde{D}E}$. The brown shading depicts the remaining gates on $D \setminus \widetilde{D}$. Fig.~\ref{fig:circuit_stability_3}(f) is the extending channel $\Phi$ of $\rho_{BC\widetilde{D}E}$ applied. We obtain the reduced density matrix $\Tr_{DE}\left(V_{AD\setminus\widetilde{D}}\rho V_{AD\setminus\widetilde{D}}^{\dagger}\right)$, where $V_{AD\setminus\widetilde{D}}$ is a unitary supported on $AD\setminus\widetilde{D}$. 

The final step is illustrated by Fig.~\ref{fig:circuit_stability_3}(g)-(i). Fig.~\ref{fig:circuit_stability_3}(g) shows the unitary $K$, depicted in green. Fig.~\ref{fig:circuit_stability_3}(h) is the state with $K$ applied. Fig.~\ref{fig:circuit_stability_3}(i) is the partial trace $\Tr_{D'E' \setminus DE}$ applied. The reduced density matrix on $A'B'C'$ is our desired $\rho_{A'B'C'}'$ [Eq.~\eqref{eq:K_restores}], so our proof is complete.

\begin{figure}[h]
\centering
\begin{tikzpicture}[line width=1pt, scale=0.9]
    \begin{scope}[xshift=0.0cm, yshift=0.0cm]
    \draw[fill=yellow] (-2, -1.5) -- ++ (4, 0) -- ++ (0, 3) -- ++ (-4, 0) -- cycle;
    \draw[style=dashed] (0,0) circle (0.5cm);
    \draw[] (0,0) circle (0.4cm);
    \draw[style=dashed] (0,0) circle (1cm);
    \draw[] (18.69:0.4cm) -- (24.29:1.1cm) arc(24.29:65.71:1.1cm) -- (71.31:0.4cm);
    \draw[] (-18.69:0.4cm) -- (-24.29:1.1cm) arc(-24.29:-65.71:1.1cm) -- (-71.31:0.4cm);
    \draw[] (180-18.69:0.4cm) -- (180-24.29:1.1cm) arc(180-24.29:180-65.71:1.1cm) -- (180-71.31:0.4cm);
    \draw[] (180+18.69:0.4cm) -- (180+24.29:1.1cm) arc(180+24.29:180+65.71:1.1cm) -- (180+71.31:0.4cm);
    \draw[] (24.29:1.1cm) arc(24.29:-24.29:1.1cm);
    \draw[] (180-24.29:1.1cm) arc(180-24.29:180+24.29:1.1cm);
    \draw[] (67:1.0cm) arc(67:180-67:1.0cm);
    \draw[] (-67:1.0cm) arc(-67:-180+67:1.0cm);
    \draw[style=dashed] (30:0.5cm) -- (30:1cm);
    \draw[style=dashed] (150:0.5cm) -- (150:1cm);
    \draw[style=dashed] (-30:0.5cm) -- (-30:1cm);
    \draw[style=dashed] (-150:0.5cm) -- (-150:1cm);
    \draw[style=dashed] (60:0.5cm) -- (60:1cm);
    \draw[style=dashed] (120:0.5cm) -- (120:1cm);
    \draw[style=dashed] (-60:0.5cm) -- (-60:1cm);
    \draw[style=dashed] (-120:0.5cm) -- (-120:1cm);
    
    \node[] () at (90:0.75cm) {\small $D'$};
    \node[] () at (0:0.75cm) {\small $B'$};
    \node[] () at (0,0) {\small $C'$};
    \node[] () at (180:0.75cm) {\small $B'$};
    \node[] () at (270:0.75cm) {\small $D'$};
    \node[] () at (45:0.75cm) {\small $E'$};
    \node[] () at (-45:0.75cm) {\small $E'$};
    \node[] () at (135:0.75cm) {\small $E'$};
    \node[] () at (-135:0.75cm) {\small $E'$};
    \node[] () at (-1.5, 0) {\small $A'$};
    \node[below] () at (0, -1.5) {(a)};
    \end{scope}
    \begin{scope}[xshift=4.5cm, yshift=0.0cm]
    \draw[fill=yellow] (-2, -1.5) -- ++ (4, 0) -- ++ (0, 3) -- ++ (-4, 0) -- cycle;
    \draw[draw=none, fill=green] (71.31:0.4cm) -- (65.71:1.1cm) arc(65.71:-65.71:1.1cm) -- (-71.31:0.4cm) arc(-71.31:71.31:0.4cm);
    \draw[draw=none, fill=green] (180-71.31:0.4cm) -- (180-65.71:1.1cm) arc(180-65.71:180+65.71:1.1cm) -- (180+71.31:0.4cm) arc(180+71.31:180-71.31:0.4cm);
    \draw[style=dashed] (0,0) circle (0.5cm);
    \draw[] (0,0) circle (0.4cm);
    \draw[style=dashed] (0,0) circle (1cm);
    \draw[] (18.69:0.4cm) -- (24.29:1.1cm) arc(24.29:65.71:1.1cm) -- (71.31:0.4cm);
    \draw[] (-18.69:0.4cm) -- (-24.29:1.1cm) arc(-24.29:-65.71:1.1cm) -- (-71.31:0.4cm);
    \draw[] (180-18.69:0.4cm) -- (180-24.29:1.1cm) arc(180-24.29:180-65.71:1.1cm) -- (180-71.31:0.4cm);
    \draw[] (180+18.69:0.4cm) -- (180+24.29:1.1cm) arc(180+24.29:180+65.71:1.1cm) -- (180+71.31:0.4cm);
    \draw[] (24.29:1.1cm) arc(24.29:-24.29:1.1cm);
    \draw[] (180-24.29:1.1cm) arc(180-24.29:180+24.29:1.1cm);
    \draw[] (67:1.0cm) arc(67:180-67:1.0cm);
    \draw[] (-67:1.0cm) arc(-67:-180+67:1.0cm);
    \draw[style=dashed] (30:0.5cm) -- (30:1cm);
    \draw[style=dashed] (150:0.5cm) -- (150:1cm);
    \draw[style=dashed] (-30:0.5cm) -- (-30:1cm);
    \draw[style=dashed] (-150:0.5cm) -- (-150:1cm);
    \draw[style=dashed] (60:0.5cm) -- (60:1cm);
    \draw[style=dashed] (120:0.5cm) -- (120:1cm);
    \draw[style=dashed] (-60:0.5cm) -- (-60:1cm);
    \draw[style=dashed] (-120:0.5cm) -- (-120:1cm);
    
    \node[] () at (90:0.75cm) {\small $D'$};
    \node[] () at (0:0.75cm) {\small $B'$};
    \node[] () at (0,0) {\small $C'$};
    \node[] () at (180:0.75cm) {\small $B'$};
    \node[] () at (270:0.75cm) {\small $D'$};
    \node[] () at (45:0.75cm) {\small $E'$};
    \node[] () at (-45:0.75cm) {\small $E'$};
    \node[] () at (135:0.75cm) {\small $E'$};
    \node[] () at (-135:0.75cm) {\small $E'$};
    \node[] () at (-1.5, 0) {\small $A'$};
    \node[below] () at (0, -1.5) {(b)};
    \end{scope}
    \begin{scope}[xshift=9.0cm, yshift=0.0cm]
    \draw[fill=yellow] (-2, -1.5) -- ++ (4, 0) -- ++ (0, 3) -- ++ (-4, 0) -- cycle;
    \draw[draw=none, fill=red!60!white] (71.31:0.4cm) -- (65.71:1.1cm) arc(65.71:-65.71:1.1cm) -- (-71.31:0.4cm) arc(-71.31:71.31:0.4cm);
    \draw[draw=none, fill=red!60!white] (180-71.31:0.4cm) -- (180-65.71:1.1cm) arc(180-65.71:180+65.71:1.1cm) -- (180+71.31:0.4cm) arc(180+71.31:180-71.31:0.4cm);
    \draw[draw=none, fill=white] (60:0.5cm) -- (60:1.0cm) arc(60:-60:1.0cm) -- (-60:0.5cm) arc(-60:60:0.5cm);
    \draw[draw=none, fill=white] (180-60:0.5cm) -- (180-60:1.0cm) arc(180-60:180+60:1.0cm) -- (180+60:0.5cm) arc(180+60:180-60:0.5cm);
    \draw[] (0,0) circle (0.5cm);
    \draw[] (0,0) circle (1cm);
    \draw[] (30:0.5cm) -- (30:1cm);
    \draw[] (150:0.5cm) -- (150:1cm);
    \draw[] (-30:0.5cm) -- (-30:1cm);
    \draw[] (-150:0.5cm) -- (-150:1cm);
    \draw[] (60:0.5cm) -- (60:1cm);
    \draw[] (120:0.5cm) -- (120:1cm);
    \draw[] (-60:0.5cm) -- (-60:1cm);
    \draw[] (-120:0.5cm) -- (-120:1cm);
    
    \node[] () at (90:0.75cm) {\small $D$};
    \node[] () at (0:0.75cm) {\small $B$};
    \node[] () at (0,0) {\small $C$};
    \node[] () at (180:0.75cm) {\small $B$};
    \node[] () at (270:0.75cm) {\small $D$};
    \node[] () at (45:0.75cm) {\small $E$};
    \node[] () at (-45:0.75cm) {\small $E$};
    \node[] () at (135:0.75cm) {\small $E$};
    \node[] () at (-135:0.75cm) {\small $E$};
    \node[] () at (-1.5, 0) {\small $A$};
    \node[below] () at (0, -1.5) {(c)};
    \end{scope}
    \begin{scope}[xshift=0.0cm, yshift=-4.0cm]
    \draw[fill=yellow] (-2, -1.5) -- ++ (4, 0) -- ++ (0, 3) -- ++ (-4, 0) -- cycle;
    \draw[draw=none, fill=red!60!white] (71.31:0.4cm) -- (65.71:1.1cm) arc(65.71:-65.71:1.1cm) -- (-71.31:0.4cm) arc(-71.31:71.31:0.4cm);
    \draw[draw=none, fill=red!60!white] (180-71.31:0.4cm) -- (180-65.71:1.1cm) arc(180-65.71:180+65.71:1.1cm) -- (180+71.31:0.4cm) arc(180+71.31:180-71.31:0.4cm);
    \draw[fill=red!30!white] (60:1.0cm) arc(60:180-60:1.0cm) -- (180-60:0.5cm) arc(180-60:180+60:0.5cm) -- (180+60:1.0cm) arc(180+60:360-60:1.0cm) -- (360-60:0.5cm) arc(360-60:360+60:0.5cm) -- (60:1.0cm);
    \draw[draw=none, fill=white] (60:0.5cm) -- (60:1.0cm) arc(60:-60:1.0cm) -- (-60:0.5cm) arc(-60:60:0.5cm);
    \draw[draw=none, fill=white] (180-60:0.5cm) -- (180-60:1.0cm) arc(180-60:180+60:1.0cm) -- (180+60:0.5cm) arc(180+60:180-60:0.5cm);
    \draw[line width=0.5pt, color=yellow!50!white, fill=yellow!50!white] (0,0) circle (0.4cm);
    \draw[draw=none, fill=yellow!50!white] (71.31:0.4cm) -- (66.28:1.0cm) arc(66.28:180-66.28:1.0cm) -- (180-71.31:0.4cm) arc(180-71.31:71.31:0.4cm);
    \draw[draw=none, fill=yellow!50!white] (-71.31:0.4cm) -- (-66.28:1.0cm) arc(-66.28:-180+66.28:1.0cm) -- (-180+71.31:0.4cm) arc(-180+71.31:-71.31:0.4cm);
    \draw[] (0,0) circle (1cm);
    \draw[] (60:0.5cm) arc(60:-60:0.5cm);
    \draw[] (180-60:0.5cm) arc(180-60:180+60:0.5cm);
    \draw[color=gray] (60:0.5cm) arc(60:180-60:0.5cm);
    \draw[color=gray] (-60:0.5cm) arc(-60:-180+60:0.5cm);
    \draw[] (30:0.5cm) -- (30:1cm);
    \draw[] (150:0.5cm) -- (150:1cm);
    \draw[] (-30:0.5cm) -- (-30:1cm);
    \draw[] (-150:0.5cm) -- (-150:1cm);
    \draw[] (60:0.5cm) -- (60:1cm);
    \draw[] (120:0.5cm) -- (120:1cm);
    \draw[] (-60:0.5cm) -- (-60:1cm);
    \draw[] (-120:0.5cm) -- (-120:1cm);
    
    \node[text=gray] () at (90:0.75cm) {\small $D$};
    \node[] () at (0:0.75cm) {\small $B$};
    \node[text=gray] () at (0,0) {\small $C$};
    \node[] () at (180:0.75cm) {\small $B$};
    \node[text=gray] () at (270:0.75cm) {\small $D$};
    \node[] () at (45:0.75cm) {\small $E$};
    \node[] () at (-45:0.75cm) {\small $E$};
    \node[] () at (135:0.75cm) {\small $E$};
    \node[] () at (-135:0.75cm) {\small $E$};
    \node[] () at (-1.5, 0) {\small $A$};
    \node[below] () at (0, -1.5) {(d)};
    \end{scope}
    \begin{scope}[xshift=4.5cm, yshift=-4.0cm]
    \draw[fill=yellow] (-2, -1.5) -- ++ (4, 0) -- ++ (0, 3) -- ++ (-4, 0) -- cycle;
    \draw[draw=none, fill=red!60!white] (71.31:0.4cm) -- (65.71:1.1cm) arc(65.71:-65.71:1.1cm) -- (-71.31:0.4cm) arc(-71.31:71.31:0.4cm);
    \draw[draw=none, fill=red!60!white] (180-71.31:0.4cm) -- (180-65.71:1.1cm) arc(180-65.71:180+65.71:1.1cm) -- (180+71.31:0.4cm) arc(180+71.31:180-71.31:0.4cm);
    \draw[fill=gray!10!white] (60:1.0cm) arc(60:180-60:1.0cm) -- (180-60:0.5cm) arc(180-60:180+60:0.5cm) -- (180+60:1.0cm) arc(180+60:360-60:1.0cm) -- (360-60:0.5cm) arc(360-60:360+60:0.5cm) -- (60:1.0cm);
    \draw[draw=none, fill=white] (60:0.5cm) -- (60:1.0cm) arc(60:-60:1.0cm) -- (-60:0.5cm) arc(-60:60:0.5cm);
    \draw[draw=none, fill=white] (180-60:0.5cm) -- (180-60:1.0cm) arc(180-60:180+60:1.0cm) -- (180+60:0.5cm) arc(180+60:180-60:0.5cm);
    \draw[line width=0.5pt, color=gray!10!white, fill=gray!10!white] (0,0) circle (0.4cm);
    \draw[] (60:0.5cm) arc(60:-60:0.5cm);
    \draw[] (180-60:0.5cm) arc(180-60:180+60:0.5cm);
    \draw[color=gray] (60:0.5cm) arc(60:180-60:0.5cm);
    \draw[color=gray] (-60:0.5cm) arc(-60:-180+60:0.5cm);
    \draw[draw=none, fill=brown!60!white] (60:0.9cm) -- (60:1.0cm) arc(60:180-60:1.0cm) -- (180-60:0.9cm) arc(180-60:60:0.9cm);
    \draw[draw=none, fill=brown!60!white] (-60:0.9cm) -- (-60:1.0cm) arc(-60:-180+60:1.0cm) -- (-180+60:0.9cm) arc(-180+60:-60:0.9cm);
    \draw[style=dotted, color=gray] (60:0.9cm) arc(60:120:0.9cm);
    \draw[style=dotted, color=gray] (-60:0.9cm) arc(-60:-120:0.9cm);
    \draw[] (0,0) circle (1cm);
    \draw[] (30:0.5cm) -- (30:1cm);
    \draw[] (150:0.5cm) -- (150:1cm);
    \draw[] (-30:0.5cm) -- (-30:1cm);
    \draw[] (-150:0.5cm) -- (-150:1cm);
    \draw[] (60:0.5cm) -- (60:1cm);
    \draw[] (120:0.5cm) -- (120:1cm);
    \draw[] (-60:0.5cm) -- (-60:1cm);
    \draw[] (-120:0.5cm) -- (-120:1cm);
    
    \node[text=gray] () at (90:0.7cm) {\small $\widetilde{D}$};
    \node[] () at (0:0.75cm) {\small $B$};
    \node[text=gray] () at (0,0) {\small $C$};
    \node[] () at (180:0.75cm) {\small $B$};
    \node[text=gray] () at (270:0.7cm) {\small $\widetilde{D}$};
    \node[] () at (45:0.75cm) {\small $E$};
    \node[] () at (-45:0.75cm) {\small $E$};
    \node[] () at (135:0.75cm) {\small $E$};
    \node[] () at (-135:0.75cm) {\small $E$};
    \node[] () at (-1.5, 0) {\small $A$};
    \node[below] () at (0, -1.5) {(e)};
    \end{scope}
    \begin{scope}[xshift=9.0cm, yshift=-4.0cm]
    \draw[fill=yellow] (-2, -1.5) -- ++ (4, 0) -- ++ (0, 3) -- ++ (-4, 0) -- cycle;
    \draw[draw=none, fill=red!60!white] (71.31:0.4cm) -- (65.71:1.1cm) arc(65.71:-65.71:1.1cm) -- (-71.31:0.4cm) arc(-71.31:71.31:0.4cm);
    \draw[draw=none, fill=red!60!white] (180-71.31:0.4cm) -- (180-65.71:1.1cm) arc(180-65.71:180+65.71:1.1cm) -- (180+71.31:0.4cm) arc(180+71.31:180-71.31:0.4cm);
    \draw[fill=gray!10!white] (60:1.0cm) arc(60:180-60:1.0cm) -- (180-60:0.5cm) arc(180-60:180+60:0.5cm) -- (180+60:1.0cm) arc(180+60:360-60:1.0cm) -- (360-60:0.5cm) arc(360-60:360+60:0.5cm) -- (60:1.0cm);
    \draw[draw=none, fill=white] (60:0.5cm) -- (60:1.0cm) arc(60:-60:1.0cm) -- (-60:0.5cm) arc(-60:60:0.5cm);
    \draw[draw=none, fill=white] (180-60:0.5cm) -- (180-60:1.0cm) arc(180-60:180+60:1.0cm) -- (180+60:0.5cm) arc(180+60:180-60:0.5cm);
    \draw[draw=none, fill=gray!10!white] (30:0.5cm) -- (30:1.0cm) arc(30:150:1.0cm) -- (150:0.5cm) arc(150:30:0.5cm);
    \draw[draw=none, fill=gray!10!white] (-30:0.5cm) -- (-30:1.0cm) arc(-30:-150:1.0cm) -- (-150:0.5cm) arc(-150:-30:0.5cm);
    \draw[fill=white] (0,0) circle (0.5cm);
    \draw[] (60:0.5cm) arc(60:-60:0.5cm);
    \draw[] (180-60:0.5cm) arc(180-60:180+60:0.5cm);
    \draw[draw=none, fill=brown!60!white] (60:0.9cm) -- (60:1.0cm) arc(60:180-60:1.0cm) -- (180-60:0.9cm) arc(180-60:60:0.9cm);
    \draw[draw=none, fill=brown!60!white] (-60:0.9cm) -- (-60:1.0cm) arc(-60:-180+60:1.0cm) -- (-180+60:0.9cm) arc(-180+60:-60:0.9cm);
    \draw[style=dotted, color=gray] (60:0.9cm) arc(60:120:0.9cm);
    \draw[style=dotted, color=gray] (-60:0.9cm) arc(-60:-120:0.9cm);
    \draw[] (0,0) circle (1cm);
    \draw[] (30:0.5cm) -- (30:1cm);
    \draw[] (150:0.5cm) -- (150:1cm);
    \draw[] (-30:0.5cm) -- (-30:1cm);
    \draw[] (-150:0.5cm) -- (-150:1cm);
    \draw[color=gray] (60:0.5cm) -- (60:1cm);
    \draw[color=gray] (120:0.5cm) -- (120:1cm);
    \draw[color=gray] (-60:0.5cm) -- (-60:1cm);
    \draw[color=gray] (-120:0.5cm) -- (-120:1cm);
    
    \node[text=gray] () at (90:0.7cm) {\small $\widetilde{D}$};
    \node[] () at (0:0.75cm) {\small $B$};
    \node[] () at (0,0) {\small $C$};
    \node[] () at (180:0.75cm) {\small $B$};
    \node[text=gray] () at (270:0.7cm) {\small $\widetilde{D}$};
    \node[text=gray] () at (45:0.75cm) {\small $E$};
    \node[text=gray] () at (-45:0.75cm) {\small $E$};
    \node[text=gray] () at (135:0.75cm) {\small $E$};
    \node[text=gray] () at (-135:0.75cm) {\small $E$};
    \node[] () at (-1.5, 0) {\small $A$};
    \node[below] () at (0, -1.5) {(f)};
    \end{scope}
    \begin{scope}[xshift=0.0cm, yshift=-8.0cm]
    \draw[fill=yellow] (-2, -1.5) -- ++ (4, 0) -- ++ (0, 3) -- ++ (-4, 0) -- cycle;
    \draw[draw=none, fill=red!60!white] (71.31:0.4cm) -- (65.71:1.1cm) arc(65.71:-65.71:1.1cm) -- (-71.31:0.4cm) arc(-71.31:71.31:0.4cm);
    \draw[draw=none, fill=red!60!white] (180-71.31:0.4cm) -- (180-65.71:1.1cm) arc(180-65.71:180+65.71:1.1cm) -- (180+71.31:0.4cm) arc(180+71.31:180-71.31:0.4cm);
    \draw[draw=none, fill=white] (60:0.5cm) -- (60:1.0cm) arc(60:-60:1.0cm) -- (-60:0.5cm) arc(-60:60:0.5cm);
    \draw[draw=none, fill=white] (180-60:0.5cm) -- (180-60:1.0cm) arc(180-60:180+60:1.0cm) -- (180+60:0.5cm) arc(180+60:180-60:0.5cm);
    \draw[draw=none, fill=gray!10!white] (30:0.5cm) -- (30:1.0cm) arc(30:150:1.0cm) -- (150:0.5cm) arc(150:30:0.5cm);
    \draw[draw=none, fill=gray!10!white] (-30:0.5cm) -- (-30:1.0cm) arc(-30:-150:1.0cm) -- (-150:0.5cm) arc(-150:-30:0.5cm);
    \draw[draw=none, fill=brown!60!white] (60:0.9cm) -- (60:1.0cm) arc(60:180-60:1.0cm) -- (180-60:0.9cm) arc(180-60:60:0.9cm);
    \draw[draw=none, fill=brown!60!white] (-60:0.9cm) -- (-60:1.0cm) arc(-60:-180+60:1.0cm) -- (-180+60:0.9cm) arc(-180+60:-60:0.9cm);
    \draw[draw=none, fill=green] (30:0.5cm) -- (30:1.1cm) arc(30:-30:1.1cm) -- (-30:0.5cm) arc(-30:-180+30:0.5cm) -- (180+30:1.1cm) arc(180+30:180-30:1.1cm) -- (180-30:0.5cm) arc(180-30:30:0.5cm);
    \draw[style=dashed] (0,0) circle (0.5cm);
    \draw[style=dotted, color=gray] (60:0.9cm) arc(60:120:0.9cm);
    \draw[style=dotted, color=gray] (-60:0.9cm) arc(-60:-120:0.9cm);
    \draw[style=dashed] (0,0) circle (1cm);
    \draw[] (0,0) circle (0.4cm);
    \draw[] (18.69:0.4cm) -- (24.29:1.1cm) arc(24.29:65.71:1.1cm) -- (71.31:0.4cm);
    \draw[] (-18.69:0.4cm) -- (-24.29:1.1cm) arc(-24.29:-65.71:1.1cm) -- (-71.31:0.4cm);
    \draw[] (180-18.69:0.4cm) -- (180-24.29:1.1cm) arc(180-24.29:180-65.71:1.1cm) -- (180-71.31:0.4cm);
    \draw[] (180+18.69:0.4cm) -- (180+24.29:1.1cm) arc(180+24.29:180+65.71:1.1cm) -- (180+71.31:0.4cm);
    \draw[] (24.29:1.1cm) arc(24.29:-24.29:1.1cm);
    \draw[] (180-24.29:1.1cm) arc(180-24.29:180+24.29:1.1cm);
    \draw[] (67:1.0cm) arc(67:180-67:1.0cm);
    \draw[] (-67:1.0cm) arc(-67:-180+67:1.0cm);
    \draw[style=dashed] (30:0.5cm) -- (30:1cm);
    \draw[style=dashed] (150:0.5cm) -- (150:1cm);
    \draw[style=dashed] (-30:0.5cm) -- (-30:1cm);
    \draw[style=dashed] (-150:0.5cm) -- (-150:1cm);
    \draw[style=dashed, color=gray] (60:0.5cm) -- (60:1cm);
    \draw[style=dashed, color=gray] (120:0.5cm) -- (120:1cm);
    \draw[style=dashed, color=gray] (-60:0.5cm) -- (-60:1cm);
    \draw[style=dashed, color=gray] (-120:0.5cm) -- (-120:1cm);
    
    \node[text=gray] () at (90:0.7cm) {\small $D'$};
    \node[] () at (0:0.75cm) {\small $B'$};
    \node[] () at (0,0) {\small $C'$};
    \node[] () at (180:0.75cm) {\small $B'$};
    \node[text=gray] () at (270:0.7cm) {\small $D'$};
    \node[text=gray] () at (45:0.75cm) {\small $E'$};
    \node[text=gray] () at (-45:0.75cm) {\small $E'$};
    \node[text=gray] () at (135:0.75cm) {\small $E'$};
    \node[text=gray] () at (-135:0.75cm) {\small $E'$};
    \node[] () at (-1.5, 0) {\small $A'$};
    \node[below] () at (0, -1.5) {(g)};
    \end{scope}
    \begin{scope}[xshift=4.5cm, yshift=-8.0cm]
    \draw[fill=yellow] (-2, -1.5) -- ++ (4, 0) -- ++ (0, 3) -- ++ (-4, 0) -- cycle;
    \draw[draw=none, fill=red!60!white] (71.31:0.4cm) -- (65.71:1.1cm) arc(65.71:-65.71:1.1cm) -- (-71.31:0.4cm) arc(-71.31:71.31:0.4cm);
    \draw[draw=none, fill=red!60!white] (180-71.31:0.4cm) -- (180-65.71:1.1cm) arc(180-65.71:180+65.71:1.1cm) -- (180+71.31:0.4cm) arc(180+71.31:180-71.31:0.4cm);
    \draw[draw=none, fill=white] (60:0.5cm) -- (60:1.0cm) arc(60:-60:1.0cm) -- (-60:0.5cm) arc(-60:60:0.5cm);
    \draw[draw=none, fill=white] (180-60:0.5cm) -- (180-60:1.0cm) arc(180-60:180+60:1.0cm) -- (180+60:0.5cm) arc(180+60:180-60:0.5cm);
    \draw[draw=none, fill=gray!10!white] (30:0.5cm) -- (30:1.0cm) arc(30:150:1.0cm) -- (150:0.5cm) arc(150:30:0.5cm);
    \draw[draw=none, fill=gray!10!white] (-30:0.5cm) -- (-30:1.0cm) arc(-30:-150:1.0cm) -- (-150:0.5cm) arc(-150:-30:0.5cm);
    \draw[draw=none, fill=brown!60!white] (60:0.9cm) -- (60:1.0cm) arc(60:180-60:1.0cm) -- (180-60:0.9cm) arc(180-60:60:0.9cm);
    \draw[draw=none, fill=brown!60!white] (-60:0.9cm) -- (-60:1.0cm) arc(-60:-180+60:1.0cm) -- (-180+60:0.9cm) arc(-180+60:-60:0.9cm);
    \draw[draw=none, fill=red!60!white] (30:0.5cm) -- (30:1.1cm) arc(30:-30:1.1cm) -- (-30:0.5cm) arc(-30:-180+30:0.5cm) -- (180+30:1.1cm) arc(180+30:180-30:1.1cm) -- (180-30:0.5cm) arc(180-30:30:0.5cm);
    \draw[draw=none, fill=yellow] (18.69:0.4cm) -- (24.29:1.1cm) arc(24.29:-24.29:1.1cm) -- (-18.69:0.4cm) arc(-18.69:180+18.69:0.4cm) -- (180+24.29:1.1cm) arc(180+24.29:180-24.29:1.1cm) -- (180-18.69:0.4cm) arc(180-18.69:18.69:0.4cm);
    \draw[style=dashed] (0,0) circle (0.5cm);
    \draw[style=dotted, color=gray] (60:0.9cm) arc(60:120:0.9cm);
    \draw[style=dotted, color=gray] (-60:0.9cm) arc(-60:-120:0.9cm);
    \draw[style=dashed] (0,0) circle (1cm);
    \draw[fill=yellow] (0,0) circle (0.4cm);
    \draw[] (18.69:0.4cm) -- (24.29:1.1cm) arc(24.29:65.71:1.1cm) -- (71.31:0.4cm);
    \draw[] (-18.69:0.4cm) -- (-24.29:1.1cm) arc(-24.29:-65.71:1.1cm) -- (-71.31:0.4cm);
    \draw[] (180-18.69:0.4cm) -- (180-24.29:1.1cm) arc(180-24.29:180-65.71:1.1cm) -- (180-71.31:0.4cm);
    \draw[] (180+18.69:0.4cm) -- (180+24.29:1.1cm) arc(180+24.29:180+65.71:1.1cm) -- (180+71.31:0.4cm);
    \draw[] (24.29:1.1cm) arc(24.29:-24.29:1.1cm);
    \draw[] (180-24.29:1.1cm) arc(180-24.29:180+24.29:1.1cm);
    \draw[] (67:1.0cm) arc(67:180-67:1.0cm);
    \draw[] (-67:1.0cm) arc(-67:-180+67:1.0cm);
    \draw[style=dashed] (30:0.5cm) -- (30:1cm);
    \draw[style=dashed] (150:0.5cm) -- (150:1cm);
    \draw[style=dashed] (-30:0.5cm) -- (-30:1cm);
    \draw[style=dashed] (-150:0.5cm) -- (-150:1cm);
    \draw[style=dashed, color=gray] (60:0.5cm) -- (60:1cm);
    \draw[style=dashed, color=gray] (120:0.5cm) -- (120:1cm);
    \draw[style=dashed, color=gray] (-60:0.5cm) -- (-60:1cm);
    \draw[style=dashed, color=gray] (-120:0.5cm) -- (-120:1cm);
    
    \node[text=gray] () at (90:0.7cm) {\small $D'$};
    \node[] () at (0:0.75cm) {\small $B'$};
    \node[] () at (0,0) {\small $C'$};
    \node[] () at (180:0.75cm) {\small $B'$};
    \node[text=gray] () at (270:0.7cm) {\small $D'$};
    \node[text=gray] () at (45:0.75cm) {\small $E'$};
    \node[text=gray] () at (-45:0.75cm) {\small $E'$};
    \node[text=gray] () at (135:0.75cm) {\small $E'$};
    \node[text=gray] () at (-135:0.75cm) {\small $E'$};
    \node[] () at (-1.5, 0) {\small $A'$};
    \node[below] () at (0, -1.5) {(h)};
    \end{scope}
    \begin{scope}[xshift=9.0cm, yshift=-8.0cm]
    \draw[fill=yellow] (-2, -1.5) -- ++ (4, 0) -- ++ (0, 3) -- ++ (-4, 0) -- cycle;
    \draw[draw=none, fill=red!30!white] (71.31:0.4cm) -- (65.71:1.1cm) arc(65.71:-65.71:1.1cm) -- (-71.31:0.4cm) arc(-71.31:71.31:0.4cm);
    \draw[draw=none, fill=red!30!white] (180-71.31:0.4cm) -- (180-65.71:1.1cm) arc(180-65.71:180+65.71:1.1cm) -- (180+71.31:0.4cm) arc(180+71.31:180-71.31:0.4cm);
    \draw[draw=none, fill=white] (60:0.5cm) -- (60:1.0cm) arc(60:-60:1.0cm) -- (-60:0.5cm) arc(-60:60:0.5cm);
    \draw[draw=none, fill=white] (180-60:0.5cm) -- (180-60:1.0cm) arc(180-60:180+60:1.0cm) -- (180+60:0.5cm) arc(180+60:180-60:0.5cm);
    \draw[draw=none, fill=gray!10!white] (30:0.5cm) -- (30:1.0cm) arc(30:150:1.0cm) -- (150:0.5cm) arc(150:30:0.5cm);
    \draw[draw=none, fill=gray!10!white] (-30:0.5cm) -- (-30:1.0cm) arc(-30:-150:1.0cm) -- (-150:0.5cm) arc(-150:-30:0.5cm);
    \draw[draw=none, fill=brown!60!white] (60:0.9cm) -- (60:1.0cm) arc(60:180-60:1.0cm) -- (180-60:0.9cm) arc(180-60:60:0.9cm);
    \draw[draw=none, fill=brown!60!white] (-60:0.9cm) -- (-60:1.0cm) arc(-60:-180+60:1.0cm) -- (-180+60:0.9cm) arc(-180+60:-60:0.9cm);
    \draw[draw=none, fill=red!30!white] (30:0.5cm) -- (30:1.1cm) arc(30:-30:1.1cm) -- (-30:0.5cm) arc(-30:-180+30:0.5cm) -- (180+30:1.1cm) arc(180+30:180-30:1.1cm) -- (180-30:0.5cm) arc(180-30:30:0.5cm);
    \draw[draw=none, fill=yellow] (18.69:0.4cm) -- (24.29:1.1cm) arc(24.29:-24.29:1.1cm) -- (-18.69:0.4cm) arc(-18.69:180+18.69:0.4cm) -- (180+24.29:1.1cm) arc(180+24.29:180-24.29:1.1cm) -- (180-18.69:0.4cm) arc(180-18.69:18.69:0.4cm);
    \draw[style=dashed, color=gray] (30:0.5cm) arc(30:180-30:0.5cm);
    \draw[style=dashed, color=gray] (180+30:0.5cm) arc(180+30:360-30:0.5cm);
    \draw[style=dotted, color=gray] (60:0.9cm) arc(60:120:0.9cm);
    \draw[style=dotted, color=gray] (-60:0.9cm) arc(-60:-120:0.9cm);
    \draw[style=dashed, color=gray] (30:1.0cm) arc(30:180-30:1.0cm);
    \draw[style=dashed, color=gray] (180+30:1.0cm) arc(180+30:360-30:1.0cm);
    \draw[] (18.69:0.4cm) -- (24.29:1.1cm) arc(24.29:65.71:1.1cm) -- (67:1.0cm);
    \draw[color=gray] (67:1.0cm) -- (71.31:0.4cm);
    \draw[] (-18.69:0.4cm) -- (-24.29:1.1cm) arc(-24.29:-65.71:1.1cm) -- (-67:1.0cm);
    \draw[color=gray] (-67:1.0cm) -- (-71.31:0.4cm);
    \draw[] (180-18.69:0.4cm) -- (180-24.29:1.1cm) arc(180-24.29:180-65.71:1.1cm) -- (180-67:1.0cm);
    \draw[color=gray] (180-67:1.0cm) -- (180-71.31:0.4cm);
    \draw[] (180+18.69:0.4cm) -- (180+24.29:1.1cm) arc(180+24.29:180+65.71:1.1cm) -- (180+67:1.0cm);
    \draw[color=gray] (-180+67:1.0cm) -- (-180+71.31:0.4cm);
    \draw[fill=yellow] (0,0) circle (0.4cm);
    \draw[] (24.29:1.1cm) arc(24.29:-24.29:1.1cm);
    \draw[] (180-24.29:1.1cm) arc(180-24.29:180+24.29:1.1cm);
    \draw[] (67:1.0cm) arc(67:180-67:1.0cm);
    \draw[] (-67:1.0cm) arc(-67:-180+67:1.0cm);
    \draw[style=dashed, color=gray] (30:0.5cm) -- (30:1cm);
    \draw[style=dashed, color=gray] (150:0.5cm) -- (150:1cm);
    \draw[style=dashed, color=gray] (-30:0.5cm) -- (-30:1cm);
    \draw[style=dashed, color=gray] (-150:0.5cm) -- (-150:1cm);
    \draw[style=dashed, color=gray] (60:0.5cm) -- (60:1cm);
    \draw[style=dashed, color=gray] (120:0.5cm) -- (120:1cm);
    \draw[style=dashed, color=gray] (-60:0.5cm) -- (-60:1cm);
    \draw[style=dashed, color=gray] (-120:0.5cm) -- (-120:1cm);
    
    \node[text=gray] () at (90:0.7cm) {\small $D'$};
    \node[] () at (0:0.75cm) {\small $B'$};
    \node[] () at (0,0) {\small $C'$};
    \node[] () at (180:0.75cm) {\small $B'$};
    \node[text=gray] () at (270:0.7cm) {\small $D'$};
    \node[text=gray] () at (45:0.75cm) {\small $E'$};
    \node[text=gray] () at (-45:0.75cm) {\small $E'$};
    \node[text=gray] () at (135:0.75cm) {\small $E'$};
    \node[text=gray] () at (-135:0.75cm) {\small $E'$};
    \node[] () at (-1.5, 0) {\small $A'$};
    \node[below] () at (0, -1.5) {(i)};
    \end{scope}
\end{tikzpicture}
\caption{The construction of extending channel $\Gamma$ for Theorem~\ref{theorem:circuit_stability_3}, illustrated.}
\label{fig:circuit_stability_3}
\end{figure}
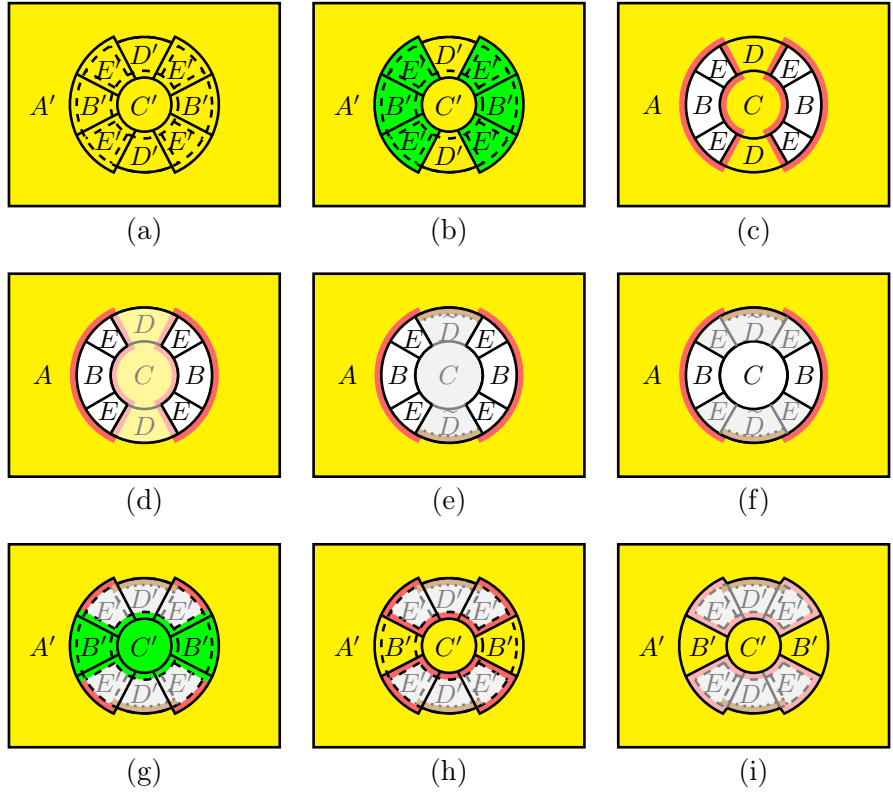

\end{document}